\documentclass[12pt]{ociamthesis}  

\usepackage[utf8]{inputenc}

\usepackage{hyperref}

\usepackage{amsthm}
\usepackage{amssymb}
\usepackage{enumitem}
\usepackage{amsmath}
\usepackage{amsthm}
\usepackage{tikz-cd}
\usepackage{mathtools}

\usepackage{subcaption}
\usepackage{caption}

\newtheorem{theorem}{Theorem}[section]

\newtheorem{innercustomthm}{Theorem}
\newenvironment{customthm}[1]
  {\renewcommand\theinnercustomthm{#1}\innercustomthm}
  {\endinnercustomthm}

\theoremstyle{definition}
\newtheorem{definition}{Definition}[section]
\newtheorem{exmp}{Example}[section]
\newtheorem{example}{Example}[section]

\theoremstyle{plain}
\newtheorem{lemma}{Lemma}[section]

\newcommand{\res}[2]{#1\rvert_{#2}}
\newcommand{\resmap}[2]{\text{res}^{#1}_{#2}}

\newcommand{\support}[1]{\left\lvert #1 \right\rvert}
\newcommand{\supp}[1]{{\support{#1}}}

\newcommand{\cf}[1]{\text{CF}(#1)}

\newcommand{\lcb}[1]{\text{LC}_{#1}^{\leftarrow}}
\newcommand{\lcf}[1]{\text{LC}_{#1}^{\rightarrow}}

\newcommand{\nats}{\mathbb{N}}
\newcommand{\zn}[1]{\mathbb{Z}_{#1}}
\newcommand{\pmodel}{\mathcal{S}}
\newcommand{\mcvx}{\mathcal{M}}

\newcommand{\ints}{\mathbb{Z}}

\newcommand{\ghz}{\ket{\text{GHZ}}}
\newcommand{\pauli}[2]{\sigma_{\MakeLowercase{#1}}^{#2}}
\newcommand{\sheaf}[1]{\mathcal{#1}}
\newcommand{\dist}{\mathcal{D}}

\newcommand{\im}[1]{\text{im}\, #1}

\DeclarePairedDelimiter\ket{\lvert}{\rangle}
\DeclarePairedDelimiterX\braket[2]{\langle}{\rangle}{#1 \delimsize\vert #2}

\newcommand{\weylScenario}{(\mathcal{O}, \mathcal{M}_{\mathcal{O}}, \mathbb{Z}_d)}
\newcommand{\cWeylOps}{\mathcal{O}}

\newcommand{\wPhaseGrp}{\mathbb{Z}_d}

\newcommand{\sExtNoArgs}{\eta}
\newcommand{\sExt}[1]{\eta(#1)}
\newcommand{\sExtFail}[2]{\Delta \eta\, (#1, #2)}
\newcommand{\sExtFailNoArgs}{\Delta \eta}

\title{Comparing two cohomological obstructions for contextuality, and a generalised construction of quantum advantage with shallow circuits}
\author{Sivert Aasn\ae ss}

\degree{Doctor of Philosophy}     
\degreedate{Trinity 2021}         

\college{Worcester College}  

\begin{document}

\baselineskip=18pt plus1pt

\setcounter{secnumdepth}{3}
\setcounter{tocdepth}{3}

\maketitle                  
\begin{abstract}
    \emph{Contextuality} is a fundamental non-classical feature
    of quantum mechanics. Abramsky et al.\ showed that contextuality
    in a range of examples is detected by a cohomological invariant
    based on \v{C}ech cohomology. However, the approach
    does not give a complete cohomological characterisation
    of contextuality. Bravyi, Gosset, and K\"{o}nig (BGK)
    gave the first unconditional proof that a restricted
    class of quantum circuits is more powerful than its classical analogue.
    The result, for the class of circuits of bounded depth and fan-in
    (shallow circuits), exploits a particular family of examples
    of contextuality.
    
    A different cohomological approach to contextuality was introduced
    by Okay et al.\ Their approach exploits the particular
    algebraic structure of the Pauli operators and their qudit generalisations
    known as Weyl operators. We give an abstract account of the algebraic
    structure of the Weyl operators, that Okay et al.\ exploit to define
    their cohomological invariant. We then generalise
    their approach to any example of contextuality with this structure.
    We prove at this general level that the approach
    does not give a more complete characterisation of contextuality
    than the \v{C}ech cohomology approach.
    
    BGK's quantum circuit and computational problem
    is derived from a family of non-local games related to the well
    known GHZ non-local game. We present a generalised version of their construction.
    A systematic way of taking examples of contextuality and producing
    unconditional quantum advantage results with shallow circuits.
\end{abstract}          

\begin{romanpages}          
\tableofcontents            
\listoffigures              
\end{romanpages}            

\chapter{Introduction}

Quantum contextuality \cite{specker_e_1960,kochen_problem_1975}, and
in particular nonlocality \cite{bell_einstein_1964},
has been highly influential in shaping our understanding 
of the distinction between quantum and classical physics.
Contextuality is a feature of the empirical data created by
measurement experiments. This is a key difference
between contextuality and certain other features
of quantum mechanics, for example, entanglement and superposition,
which are internal to the theory itself.
Because contextuality is an empirical phenomenon it
says something about any physical theory that is consistent with the 
predictions of quantum mechanics.
This is part of why contextuality was seen as so profound,
and in the era of quantum computing, it makes contextuality relevant for proving
quantum advantage. Because it is an empirical phenomenon it makes
sense to talk about classical models creating contextuality,
but something like entanglement and superposition doesn't have any
classical analogue.

No unconditional proof of quantum advantage is known for a general 
computational model. This appears to be well beyond the limits of current
techniques. A recent breakthrough by Bravyi, Gosset, and K\"{o}nig (BGK)
gave the first unconditional quantum advantage result
for a restricted class of circuits \cite{bravyi_quantum_2018}. 
A \emph{shallow circuit} is a family of circuits of bounded depth and fan-in.
BGK explicitly defines a shallow quantum family $\{Q_n\}_{n \in \nats}$
and a family of computational problems $\{\text{GHZ-2D(n)}\}_{n \in \nats}$
that are solved perfectly by the quantum circuit,
but not with high accuracy by any classical shallow circuit.

\begin{figure}
    \centering
    \begin{subfigure}{0.45\textwidth}
        \includegraphics[width=1\textwidth]{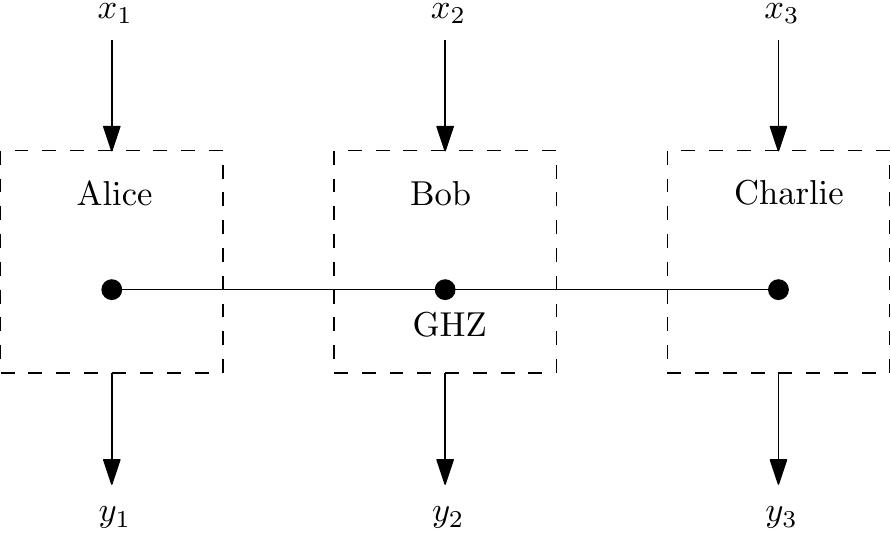}
        \caption{}
    \end{subfigure}
    \hfill
    \begin{subfigure}{0.45\textwidth}
        \includegraphics[width=\textwidth]{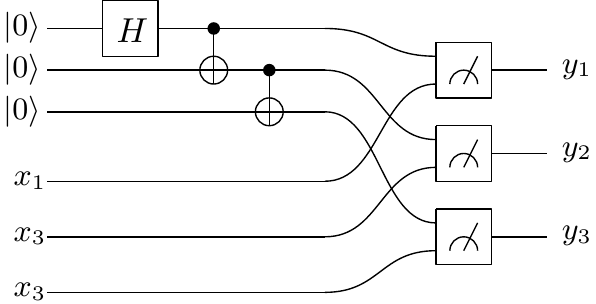}
        \caption{}
    \end{subfigure}
    \caption{Circuit version of the the \emph{GHZ game}. Inputs $x_1, x_2, x_3 \in \{0,1\}$
        are selected uniformly at random. A Haramard gate
        followed by two controlled not gates initialises the state 
        $\ket{GHZ} := \ket{000} + \ket{111}$.
        Each qubit is measured with measurement settings $0 \mapsto X, 1 \mapsto Y$.
        Outcomes $y_1, y_2, y_3 \in \{0,1\}$ are returned according
        to $1 \mapsto 0, -1 \mapsto 1$. The circuit wins if $x_1 \oplus x_2 \oplus x_3 = 1 \oplus y_1 \oplus y_2 \oplus y_3$.
        }
    \label{fig:ghz_and_circuit}
\end{figure}
It is well known that certain examples of contextuality can be recast
as cooperative games called \emph{nonlocal games}. An example
is Greenberger-Horn-Zeilling (GHZ) game 
\cite{greenberger_bells_1990, cleve_consequences_2010}. It was
observed by BGK that quantum strategies for nonlocal games
can be recast as circuits (Figure \ref{fig:ghz_and_circuit}). 
The computational problems
$\text{GHZ-2D}(n)$ that BGK considered can be seen as ``distributed'' versions of the 
GHZ game played on an $n \times n$ grid. 
This raises the question if every nonlocal game
can be turned into a family of ``distributed'' games, that
gives rise to an unconditional quantum advantage result with shallow
circuits. We show that this is the case. We describe this result in more detail
in Section \ref{section:circuits-intro}

\emph{Cohomology} can be a powerful technique for detecting structure
in data. It could therefore be a useful tool
for studying the empirical data associated with contextuality.
Two prominent cohomological approaches to contextuality
is the \emph{\v{C}ech cohomology approach} introduced
by Abramsky, Mansfield, and Barbosa \cite{abramsky_cohomology_2012}
and the \emph{topological approach} of Okay, Roberts, Bartlett, and 
Raussendorf \cite{okay_topological_2017}.
The \v{C}ech cohomology approach was further developed
by, for example, Abramsky et al.\ \cite{abramsky_contextual_2017},
and by Caru \cite{caru_towards_2018, caru_cohomology_2017,caru_logical_2019}.
The insight that contextuality has a topological structure
\cite{mansfield_contextuality_2020}
has lead to a range of results,
for example, the homotopical approach
of Okay and Raussendorf \cite{okay_homotopical_2020},
the connection with resource theory made
by Okay, Tyhurst, and Raussendorf \cite{okay_cohomological_2018},
the classifying space for contextuality \cite{okay_classifying_2021},
and more recent work by Okay, Kharoof, and Ipek 
has uncovered the simplicial structure \cite{okay_simplicial_2022}.

The \v{C}ech cohomology approach is based
on the sheaf theoretic framework
of Abramsky and Brandenberger, which
describes contextuality as a feature
of abstract families of empirical data, known
as \emph{empirical models} \cite{abramsky_sheaf-theoretic_2011}.
The \v{C}ech cohomology approach is very general. However,
this generality comes at the cost of completeness.
The issue of completeness was a main point of interest in the later
work of Abramsky et al.\ and Caru.
The topological approach lacks some of the generality of the \v{C}ech
approach. But the additional structure that the
approach requires gives the potential for a more refined approach.
In particular, we are interested in the possibility that the structure
used by the topological approach can help alleviate the
issue of incompleteness in the \v{C}ech cohomology approach.

In Section \ref{section:cohomology-intro} we give
an account of this structure used by the topological approach
within the sheaf theoretic
framework. We show that Okay et al.'s invariant can be generalised
to any empirical model equipped with this structure.
We then show that, in fact, at this level of generality
the two approaches are equivalent with resepect to
the question of completeness.

\section{Comparing two obstruction for contextuality} \label{section:cohomology-intro}
\begin{figure}
    \centering
    \begin{subfigure}{0.45\textwidth}
        \includegraphics[width=0.9\textwidth]{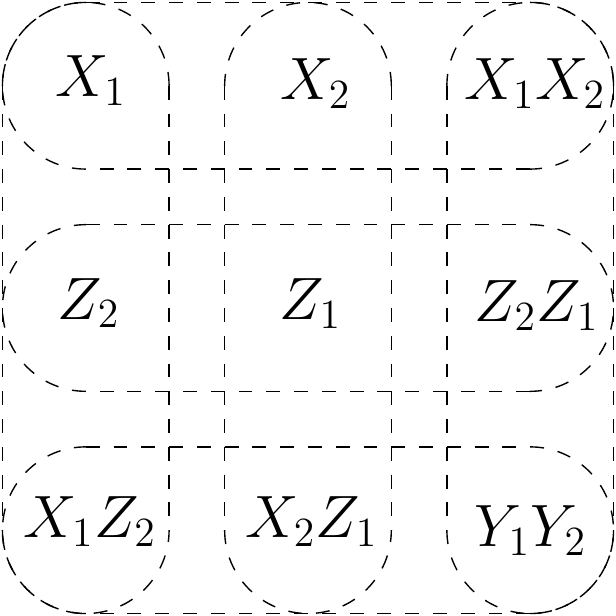}
        \caption{Mermin's square}
        \label{fig:m-square}
    \end{subfigure}
    \hfill
    \begin{subfigure}{0.45\textwidth}
        \includegraphics[width=0.9\textwidth]{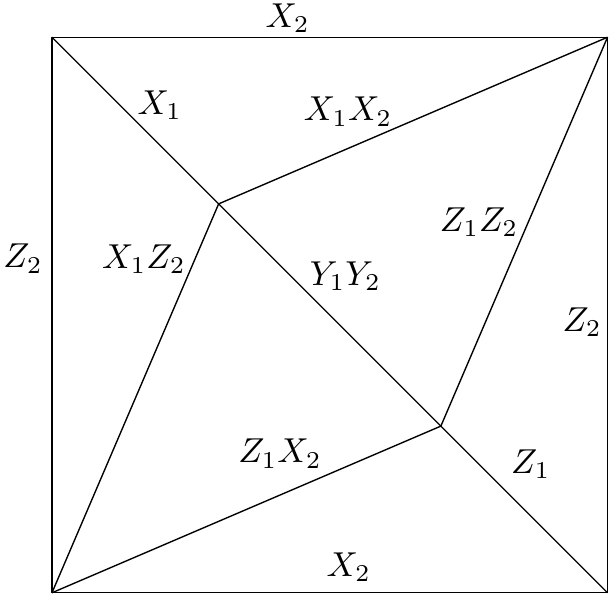}
        \caption{Classifying space for Mermin's square}
        \label{fig:m-square-space}
    \end{subfigure}
    \caption{
        (a) The set of quantum measurement operators known as \emph{Mermin's square}, 
        and (b) the associated classifying
        space of the topological approach. In (a) Each row and column represents a context of commuting operators. In (b) each operator labels a loop attached to a single point, and each context a surface.
    }
    \label{fig:topological_realisations}
\end{figure}

In the sheaf-theoretic framework contextuality
is seen as the failure of a locally compatible family of data
to be given a globally consistent description. In sheaf theory cohomology
is a powerful tool for studying the transition from local to global. 
It is therefore natural to consider the application of cohomological
methods to contextuality. Abramsky, Mansfield, and Barbosa \cite{abramsky_cohomology_2012} shows
that in a range of examples contextuality can be detected by the non-vanishing
of a cohomological invariant based on \emph{\v{C}ech cohomology}.

The sheaf-theoretic framework distinguishes between possibilitstic and probabilistic 
empirical models. A possibilistic model only keeps track of
which outcomes are possible, and not their particular probabilities.
The \v{C}ech cohomology invariant can be defined for any possibilistic empirical model.
However, it is generally not a complete invariant for contextuality.
There are so called ``false negatives'', contextual empirical models where the cohomological
invariant vanishes. False negatives can occur because
empirical models lack the required algebraic structure to directly
define the obstruction. An empirical model is a presheaf of sets,
while \v{C}ech cohomology requires a presheaf of abelian groups.
Abramsky et al.\ therefore considers the \v{C}ech cohomology
of the free abelian presheaf associated with an empirical model.

The \v{C}ech cohomology invariant lead to further work on developing
a complete cohomological invariant for contextuality.
It was shown by Abramsky, Barbosa, Kishida, Lal, and 
Mansfield \cite{abramsky_contextuality_2015} that \v{C}ech cohomology
is complete for a large class of examples captured by generalised AvN arguments.
Several other invariants have been proposed,
for example Roumen \cite{roumen_cohomology_2017}
and Caru \cite{caru_towards_2018}.

The Pauli operators and their qudit generalisations known as \emph{Weyl operators}
have a special role in quantum computing. They are used in for example
error correcting codes and 
measurement based quantum computing \cite{nielsen_quantum_2010}.
It is well known that the Pauli operators is a rich source of
examples of contextuality, this is also the case
when $d > 2$, see for example De Silva for examples \cite{de_silva_logical_2017}.
In dimension $d \geq 2$
the Weyl operators form a group $P_{n,d}$, 
called the generalised $n$-qudit Pauli group.
$P_{n,d}$ is closed under the \emph{phase} action of $\mathbb{Z}_d$.
\begin{align}
    \Omega : \mathbb{Z}_d \times P_{n,d} \to P_{n,d} :: (q, O) \mapsto \omega^q O
\end{align}
where $\omega := e^{2 \pi i / d}$.

The \emph{topological approach} of Okay et al. \cite{okay_topological_2017}
studies sets of Weyl operators that are closed under certain operations.
A set of operators $\mathcal{O} \subset P_{n,d}$ is \emph{closed}
if it satisfies the following conditions:
\begin{enumerate}
    \item $\mathcal{O}$ contains the identity: $I \in \mathcal{O}$.
    \item $\mathcal{O}$ is closed under the phase action: $\Omega(\mathbb{Z}_d, \mathcal{O}) \subset \mathcal{O}$.
    \item $\mathcal{O}$ is closed under commuting products: If $O_1, O_2 \in \mathcal{O}$ and $O_1O_2 = O_2O_1$ then $O_1O_2 \in \mathcal{O}$.
\end{enumerate}
Okay et al.\ show that questions
about contextuality for closed sets of Weyl operators can be given
a topological characterisation (Figure \ref{fig:topological_realisations}). 
The result generalises an earlier
characterisation for Pauli operators by Arkhipov \cite{arkhipov_extending_2012}.

The topological approach uses ideas from group cohomology.
Recall that a group extension of a group $K$ by a group $G$
is a short exact sequence of groups
\begin{equation}
    \begin{tikzcd}
        G \arrow[r, "i"] & H \arrow[r, "j"] & K
    \end{tikzcd} 
\end{equation}
generalising the direct product of groups $G \times K$. A left splitting, right splitting,
or trivialisation are homomorphisms $l, r, h$ respectively making
the following diagram commute:
\begin{equation}
    \begin{tikzcd}
        G \arrow[loop left, "\text{id}_G"] \arrow[r, "i"] \arrow[dr, "\text{in}_G"] 
        & H \arrow[d, "h"] \arrow[r, "j"] \arrow[l, "l" above, bend right] & K \arrow[loop right, "\text{id}_K"]
        \arrow[l, "r" above, bend right] \\
        & G \times K \arrow[ur, "\pi_2"] &
    \end{tikzcd}
\end{equation}
Group cohomology is an elegant solution to the problem
of classifying group extensions for fixed $G$ and $K$ \cite{brown_cohomology_2012}.

A closed set of Weyl operators is not a group because it is not
closed under inverses and only under commuting products.
However, Okay et al. shows that for any such set one can define
a classifying space similar to that of group cohomology.
Using this space they show that both state dependent and state
independent proofs of contextuality can be given a topological characterisation.
They show that state dependent and state independent contextuality
can be detected by the non-vanishing of a cohomology class.

\paragraph{Results}
We first give a more abstract account of the algebraic structure used by 
Okay et al.'s approach.

A \emph{bundle over a commutative partial monoid}
is a generalisation of group extensions to commutative partial monoids.
A closed set of Weyl operators comes with the structure of a bundle
over a commutative partial monoid. Proofs of contextuality
for a closed set of Weyl operators correspond to extending
\emph{local} left splittings, defined on a sub-bundle,
to \emph{global} left splittings defined on the whole bundle.
For closed sets of Weyl operators the problem
of extending a local left splitting globally can therefore
be used as a test for contextuality.

We prove a version of the splitting lemma for commutative partial
monoids, and we generalise group cohomology to partial commutative monoids.
The splitting lemma shows that for the problems of extending
either a local left splitting, right splitting, or trivialisation
of a sub-bundle to the whole bundle are equivalent.
Furthermore, the problem of extending a local splitting to a global
splitting can be given a cohomological characterisation.

We then generalise the cohomological obstruction
to any empirical model with the structure of a bundle
over a commutative partial monoid.
For such an empirical model the problem
of extending a local splitting is a test for contextuality.
We can therefore use the cohomological obstruction
for extending a local splitting.
There can be global splittings that don't correspond
to valid outcome assignments.
This raises the possibility
of false negatives.

We finally show that any false negative of the \v{C}ech approach
induces a global splitting that is not 
consistent with the model.

In summary, our results are:
\begin{itemize}
    \item Closed sets of Weyl operators come with the structure of a bundle
        over a commutative partial monoid. Local (resp. global) outcome assignments
        induce local (resp. global) left splittings of the bundle.
    \item The topological obstruction can be generalised to a class of empirical models
        equipped with the structure of a bundle over a commutative partial monoid.
    \item The vanishing of the \v{C}ech cohomology obstruction
        implies the vanishing of the generalised topological obstruction.
\end{itemize}

\section{A general construction of quantum advantage with  shallow circuits}
\label{section:circuits-intro}
Bravyi, Gosset, and K\"{o}nig's initial result was quickly
improved in several ways.
For example, it was shown to be noise robust \cite{bravyi_quantum_2020},
and it was extended to the more powerful classical circuit class AC0 \cite{watts_exponential_2019}, 
of circuits of bounded depth and 
unbounded fan-in AND, OR, and NOT gates.
It has also inspired several results for
\emph{interactive circuits},
that is circuits with more than one round 
of input and output \cite{grier_interactive_2020}.
\begin{figure}
    \centering
    \begin{subfigure}{0.48\textwidth}
        \centering
        \includegraphics[width=1\textwidth]{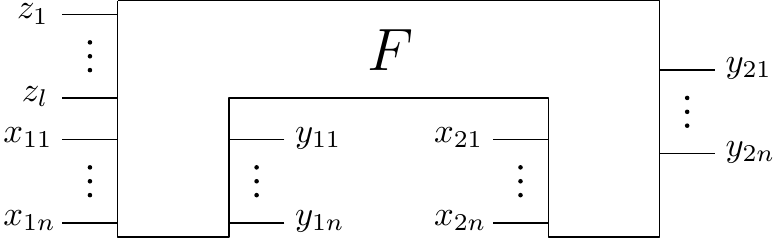}
        \caption{} \label{fig:generator-two-round}
    \end{subfigure}
    \hfill
    \begin{subfigure}{0.48\textwidth}
        \centering
        \includegraphics[width=1\textwidth]{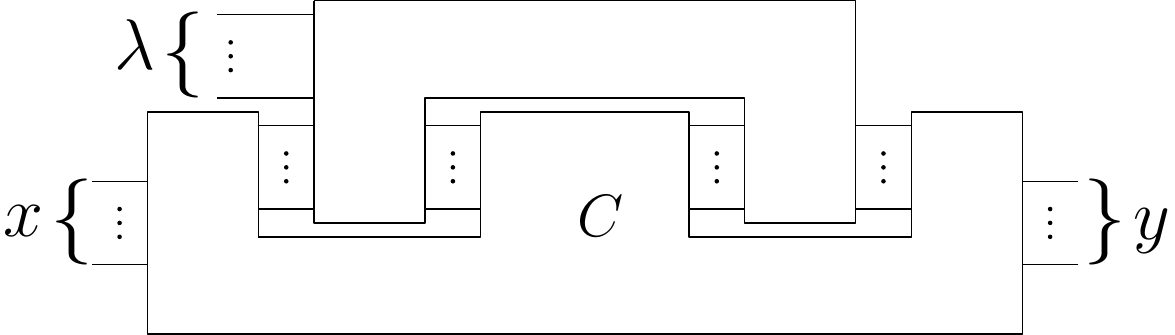}
        \caption{} \label{fig:two-round-evaluation}
    \end{subfigure}
    \caption{An interactive circuit (a) is a circuit with several rounds of inputs and outputs.
        The circuit is evaluated by composing with a classical circuit as in (b).
        }
    \label{fig:interactive_circuit_illustration}
\end{figure}

AC0 is currently at the edge of unconditional circuit separations
for classical circuits. It, therefore, seems unlikely that the techniques
used by BGK can be extended to prove much stronger complexity
theoretic results. However, in the lack of stronger results,
we should try to learn as much as possible.

BGK's result extends an earlier result by 
Barrett et al.\ \cite{barrett_modeling_2007}.
An interesting point is that after BGK's result was published
it was observed that Barrett et al.'s construction
solves an open problem about quantum advantage in distributed
computing \cite{gall_quantum_2019}.

Nonlocality is a particular type of contextuality that arise in scenarios where 
compatible measurements are performed
at distinct locations called measurement sites.
We observe that nonlocality can be recast in terms of circuits.
A quantum realisation gives rise to a circuit (Figure \ref{fig:contextuality-circuit})
that prepares an entangled state and then implements local measurements.
The circuit takes a classical
input $x_i$ and returns a classical output $y_i$ for each measurement
site $i$.
There is no path through the circuit $Q_\text{NC}$
from input $x_i$ to a different output $y_{j}$, where $i \neq j$.
$Q_\text{NC}$ is \emph{contextual} if it is not equivalent to any classical
circuit with the same inputs and outputs, such that there is no path
from an input to a different output (Figure \ref{fig:classical-circ}).

A \emph{nonlocal game} is usually thought of as being played
by a set of spatially separated players against Verifier.
We can equivalently think of a nonlocal game
as a computational problem where some quantum
circuit of the form $Q_\text{NC}$ achieves advantage over 
any classical circuit of the form $C_\text{NC}$.
In a nonlocal game $\Phi$ we randomly select an input
and an accepting condition $(x_1, \dots, x_n, A)$.
We then evaluate the circuit on inputs $x_1, \dots, x_n$. The circuit
wins if $(y_1, \dots, y_n) \in A$. The \emph{success probability}
is the likelihood of the accepting condition being satisfied.
A nonlocal game $\Phi$ is violated by a quantum strategy $Q_\text{NC}$
if there exists a bound $\gamma$ such that
\begin{align}
    p_S(C_\text{NC}, \Phi) \leq \gamma < p_S(Q_\text{NC}, \Phi)
\end{align}
where $p_S$ denotes success probability and $C_\text{NC}$ is any classical circuit
of the same form.

\begin{figure}
    \centering
    \begin{subfigure}{0.45\textwidth}
        \centering
        \includegraphics{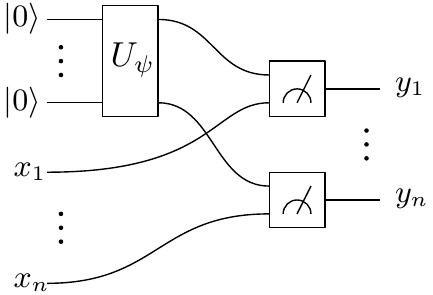}
        \caption{}
        \label{fig:contextuality-circuit}
    \end{subfigure}
    \begin{subfigure}{0.45\textwidth}
        \centering
        \includegraphics{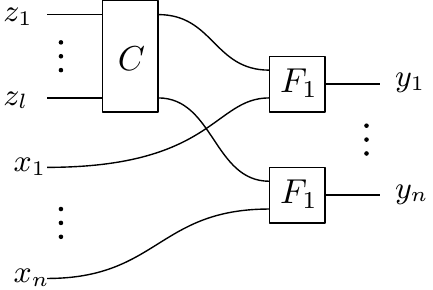}
        \caption{}
        \label{fig:classical-circ}
    \end{subfigure}
    \caption{
        The quantum circuit (a) is \emph{contextual} if it is not equivalent to any
        classical circuit (b) of the same shape, with the ability to sample
        an arbitrary random seed $z_1, \dots, z_l$.
    }
\end{figure}

Bravyi, Gosset, and K\"{o}nig introduced a family of nonlocal games
$\{\text{2D-GHZ}(n)\}_{n \in \nats}$ and a shallow quantum circuit
$\{Q_n\}_{n \in \nats}$ (Figure \ref{fig:bgk_problem_and_circuit}). 
$\text{2D-GHZ}(n)$ is a version of the GHZ-game
played on an $n \times n$ grid. The circuit $Q_n$ prepares
$n^2$ qubits in the graph state of the $n \times n$ grid and applies
classically controlled Pauli $X,Y$ or $Z$ measurements to each qubit.
It can be shown that the graph state can be prepreaed by a single Hadamard gate
on each qudit, and four layers of controlled $Z$ gates. The circuit $\{Q_n\}$
is therefore shallow. 

The inputs and accepting condition is chosen by Verifier
in the nonlocal game $\text{2D-GHZ}(n)$ is related to the inputs and accepting
condition in the GHZ-game. At the beginning of each round
Verifier randomly selects inputs $x_A,x_B,x_C$ for the GHZ-game,
nodes $v_A,v_B,v_C \in \text{Grid}(n,n)$,
and paths $u_{AB}: v_A \to v_B, u_{BC} : v_B \to v_C, u_{CA} : v_C \to v_A$.
Players $v_A,v_B,v_C$ are then given inputs $x_A,x_B,x_C$
and the remaining players are given inputs that encode
that paths $u_{AB}, u_{BC}, u_{CA}$.
An output for the players $y_1, \dots y_{n^2}$ is accepted if it satisfies
a constraint
\begin{align}
    x_A \oplus x_B \oplus x_C = 1 \oplus (y_A \oplus k_A(y)) \oplus (y_B \oplus k_B(y)) 
    \oplus (y_C \oplus k_C(y))
\end{align}
where $y_A,y_B,y_C$ are the outputs of $v_A,v_B,v_C$ and $k_A(y),k_B(y),k_C(y)$
are ``correction factors'' that only depend on the outcomes
of players along the paths close to each respective node.

BGK shows that for each $n \in \nats$ the game $\text{2D-GHZ}(n)$
is solved perfectly by the quantum circuit $Q_n$,
and that it is not solved with high accuracy by any 
classical shallow circuit $\{C_n\}_{n \in \nats}$.
\begin{figure}
    \centering
    \begin{subfigure}{0.45\textwidth}
        \centering
        \includegraphics[width=1\textwidth]{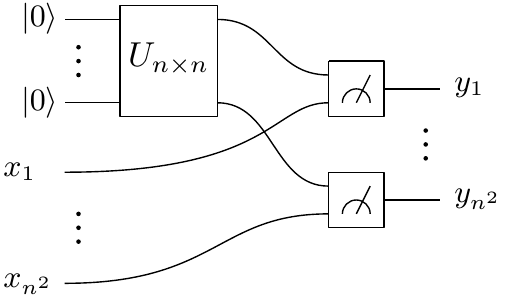}
        \caption{}
    \end{subfigure}
    \hfill
    \begin{subfigure}{0.45\textwidth}
        \centering
        \includegraphics{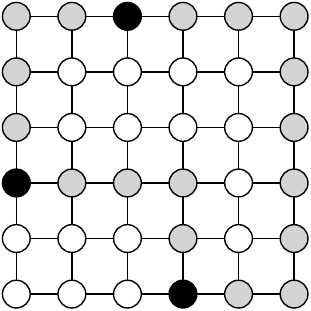}
        \caption{}
    \end{subfigure}
    \caption{The 2D-GHZ game. The quantum circuit strategy (a)
        prepares $n^2$ qubits in the cluster state of the $n \times n$
        grid. It takes inputs $x_1, \dots, x_{n^2} \in \{1,2,3\}$,
        performs controlled Pauli measurements according
        to $1 \mapsto X, 2 \mapsto Y, 3 \mapsto Z$,
        and returns outcomes $y_1, \dots, y_{n^2} \in \{0,1\}$ according
        to $1 \mapsto 0, -1 \mapsto 1$.
        The input is randomly sampled as follows.
        First select nodes $A,B,C \in n \times n$ and paths
        $u_{AB}, u_{BC}, u_{CA} \subset n \times n$ in a ``sufficiently uniform''
        way.
        The circuit wins if the output $y$ satisfies
        $x_A \oplus x_B \oplus x_C = 1 \oplus (y_A \oplus k_A(y)) \oplus (y_B \oplus k_B(y)) \oplus (y_C \oplus k_C(y))$,
        where $k_A(y), k_B(y), k_C(y) \in \{0,1\}$ depends only on the value
        of $y$ close to $A,B,C$ respectively.
    }
    \label{fig:bgk_problem_and_circuit}
\end{figure}

\begin{customthm}{\cite{bravyi_quantum_2018}}
    The shallow quantum circuit $\{Q_n\}_{n \in \nats}$ solves
    the 2D-GHZ game perfectly for all $n$.
    However, the 2D-GHZ game is not solved with high accuracy by
    any classical shallow circuit $\{C_n\}_{n \in \nats}$.
    \begin{align}
        p_S(Q_n, \text{2D-GHZ}(n)) &= 1\\
        p_S(C_n, \text{2D-GHZ}(n)) &\leq 3/4 + \epsilon_n
    \end{align}
    where $p_S$ denotes \emph{success probability} and $\epsilon_n \in O(1/n)$.
\end{customthm}

The key to BGK's result is that single-qubit measurements on an entangled
state with only local entanglement can create
entanglement between qubits that are far away.
Depth and fan-in constrain the nonlocal correlations
that a classical circuit can produce,
but it also constrains the entangled states
and the measurements that a quantum circuit can use.
Observe that in the circuit $C_\text{NC}$ there can only
be a path from input $x_i$ to output $y_i$, while in a circuit
of depth $D$ and maximal fan-in $K$ there can be a path from at most
$K^D$ inputs to any given output. As Verifier makes different
choices of players $v_A,v_B,v_C$ in the 2D-GHZ game this
forces the depth and fan-in of a classical circuit to be large.
On the classical side it can be shown that when the measurement
along the paths $u_{AB}, u_{BC}, u_{CA}$ are made, the effect
is to create an entangled $\ket{GHZ}$ state at qubits
$v_A,v_B,v_C$,
up to a local Pauli factors given by
$k_A, k_B, k_C$. Furthermore, these corrections can be made
classically post measurement.

In summary, the technique relies upon two key properties of the GHZ game:
The use of the GHZ state and Pauli measurements.
The choice of state is important
because it can be realised by local measurements on a graph state
in different ways, and the measurements are important because it allows
for the corrections $k_A,k_B,k_C$ to be performed post-measurement.

\paragraph{Results}
We first present a quantum protocol that uses teleportation
to both distribute an entangled state on a graph and perform measurements
on the distributed qudits (Figure \ref{fig:teleportation_in_introduction}).
For any multi-qudit state $\psi$ with qudit $i$
and graph $G$ with nodes $V$
we consider a scenario where a number of agents $I \times V$, 
one for each qudit of $\psi$
and node of $G$, share entanglement. Each qudit of $\psi$ is held by some
node on the graph, and each pair of nodes $(i,v), (i,w)$ such that
$v,w$ are adjacent in $G$ share a two-qudit entangled state.
By choosing a path through the graph for each qudit we can then distribute each
qudit of $\psi$ to an arbitrary node on the graph, up to a random single-qudit phase
for each qudit.
An important observation is that this can be done in a constant number of rounds
of quantum measurements.
\begin{figure}
    \centering
    \begin{subfigure}{0.48\textwidth}
    \centering
        \includegraphics[width=1\textwidth]{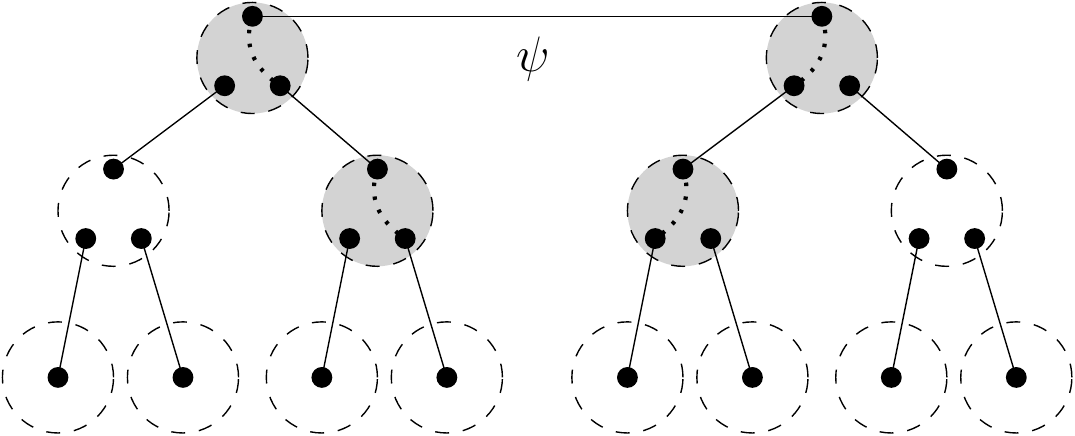}
        \caption{
        }
        \label{fig:measurement-setting}
    \end{subfigure}
    \hfill
    \begin{subfigure}{0.48\textwidth}
    \centering
        \includegraphics[width=1\textwidth]{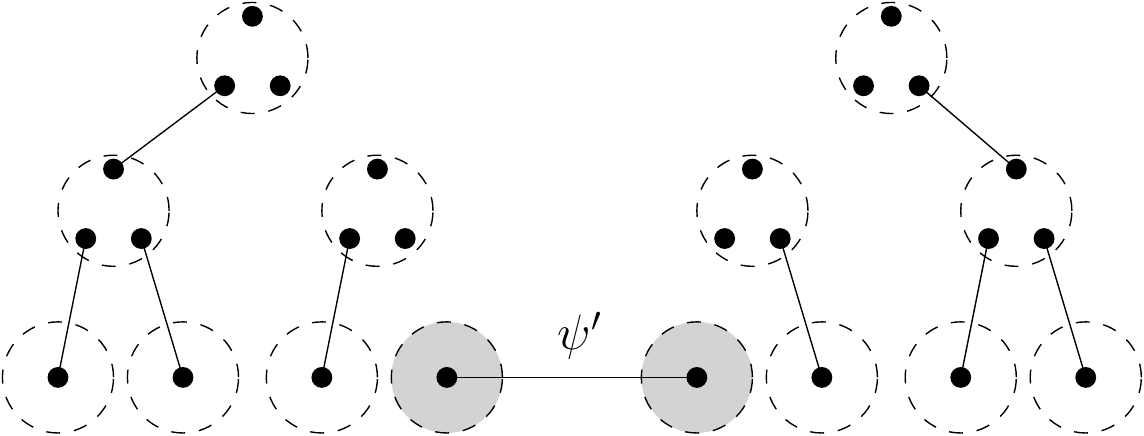}
        \caption{}
        \label{fig:step-two}
    \end{subfigure}
    \caption{Given a nonlocal game with players $I$ and state $\psi$, and a graph
    $G$ with nodes $V$ we consider the scenario (a) consisting
    of players $I \times V$, here each dotted circle indicate a player
    and $\bullet$ a qudit in either a maximally entangled state
    or the state $\psi$.
    (b) By performing local measurement on the qudits held by each
    player we can distribute each qudit, up to a random factor
    on each qudit, to any player on the graph.
    }
    \label{fig:teleportation_in_introduction}
\end{figure}

We then consider the family of protocols arising from a fixed
state and a family of graphs. Using this construction
we show that any nonlocal game gives rise to a family
of distributed games. We then show that for certain families
of graphs distributed games gives rise to unconditional
quantum advantage results with shallow circuits.

We present two versions of this construction.
The first is completely general, but the distributed games
have two rounds (Figure \ref{fig:teleportation_adn_circuit_two_rounds}). 
It is a result about interactive circuits 
(Figure \ref{fig:interactive_circuit_illustration}).
The second result is less general, but for circuits in the usual sense
having only a single round of inputs and outputs (Figure \ref{fig:teleportation_and_circuit}).
In the second result we consider nonlocal games with quantum
strategies given by measurements of single-qudit Weyl operators.
Note that the states are still completely general.

The outline of the two results is as follows.
Suppose that $(Q_\text{NC}, \Phi)$ is any nonlocal game
with classical bound $\gamma$. For any family of graphs
$\{G_n\}_{n \in \nats}$ we define a family of two-round cooperative
games $\{\Phi_n\}_n$ and two-round interactive quantum circuits $\{Q_n\}_n$,
such that for each $n \in \nats$ the quantum circuit
$Q_n$ violates the bound $\gamma$.
For certain families of graphs we show that the quantum circuit
is shallow
and that a classical shallow circuit $\{C_n\}_n$
violates the bound $\gamma$ only up to a small factor $\epsilon_n$. 
Where $\text{lim}_{n \to \infty} \epsilon_n = 0$.
The rate of convergence is a property of the graphs.

\begin{customthm}{I (Informal)}
    For any nonlocal game $\Phi$ and quantum strategy $Q$
    we define a family of two-round interactive games $\{\Phi_n\}_{n \in \nats}$
    and a shallow two-round quantum circuit $\{Q_n\}_{n \in \nats}$
    such that for any classical two-round interactive shallow circuit 
    $\{C_n\}_{n \in \nats}$
    \begin{align}
        p_S(Q_n, \Phi_n) &= p_S(Q, \Phi)\\
        p_S(C_n, \Phi_n) &\leq \gamma + \epsilon_n
    \end{align}
    for some small $\epsilon_n$.
\end{customthm}

Next, we show that if the quantum strategy uses only single-qudit Weyl
measurements (Figure \ref{fig:weyl-circ})
\begin{figure}[h]
    \centering
    \includegraphics[scale=1.2]{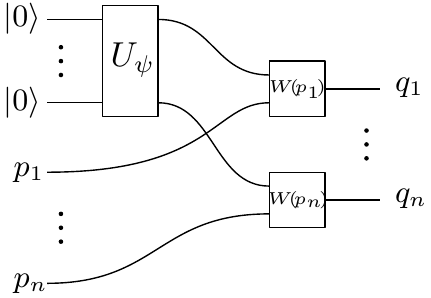}
    \caption{
        A \emph{Weyl measurement strategy} is a special quantum strategy using measurements
        in the basis of Weyl operators. For dimension $d \geq 2$ and $n$-qudit state $\psi$
        the Weyl measurement strategy takes inputs
        $p_1, \dots, p_n \in \mathbb{Z}_d^2$ and return outcomes
        $q_1, \dots, q_n \in \mathbb{Z}_d$ of performing the single-qudit
        Weyl measurement $W(p_i)$ on qudit $i$.
    }
    \label{fig:weyl-circ}
\end{figure}
then the number of input-output rounds can be reduced to one.
\begin{customthm}{II (Informal)}
    For any nonlocal game $\Phi$ and Weyl measurement strategy $Q$
    we define a family of nonlocal games $\{\Phi_n\}_{n \in \nats}$
    and a shallow quantum circuit $\{Q_n\}_{n \in \nats}$
    such that for any classical shallow circuit $\{C_n\}_{n \in \nats}$
    \begin{align}
        p_S(Q_n, \Phi_n) &= p_S(Q, \Phi)\\
        p_S(C_n, \Phi_n) &\leq \gamma + \epsilon_n
    \end{align}
    for some small $\epsilon_n$.
\end{customthm}
where $\epsilon_n$ is a different bound.

\begin{figure}
    \centering
    \begin{subfigure}{0.48\textwidth}
    \centering
        \includegraphics[width=1\textwidth]{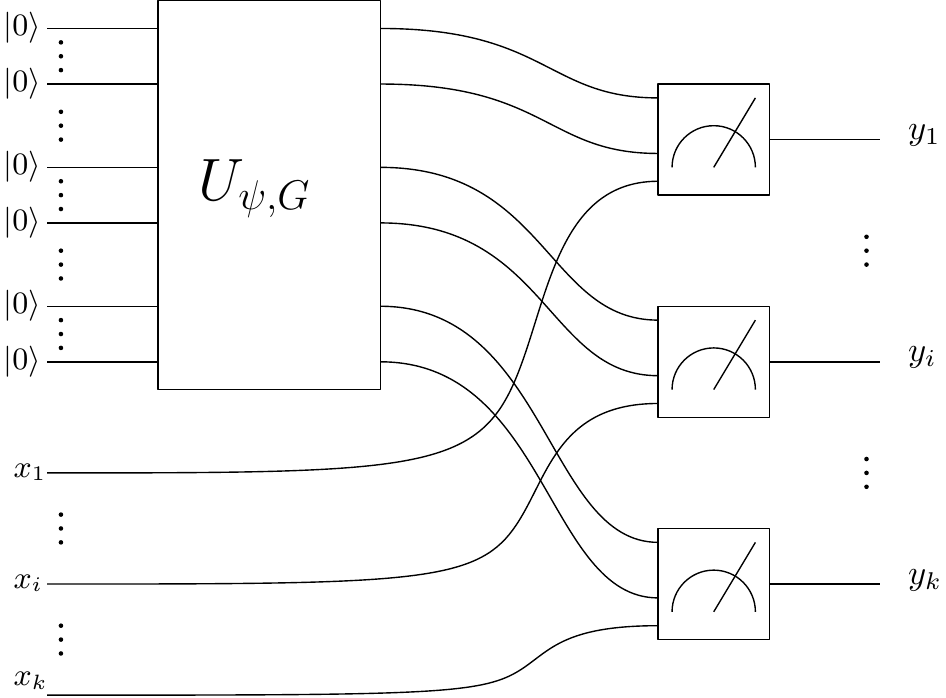}
        \caption{}
        \label{fig:step-two}
    \end{subfigure}
    \hfill
    \begin{subfigure}{0.48\textwidth}
    \centering
        \includegraphics[width=1\textwidth]{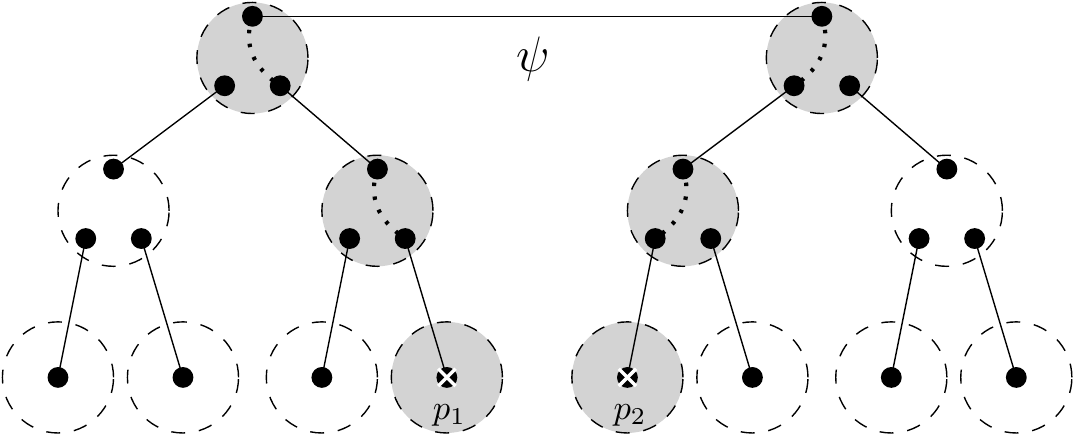}
        \caption{}
        \label{fig:measurement-setting}
    \end{subfigure}
    \caption{Let $\psi$ be an $I$-qudit state and $G = (V, E, r)$ a rooted graph.
    The quantum circuit strategy (a) prepares a single instance of $\psi$
    and a maximally entangled pair of qudits $\ket{\phi} := \frac{\ket{00} + \ket{11}}{\sqrt{2}}$
    for each $i \in I$ and edge $e \in E$.
    The circuit has an input for each $(i,v) \in I \times V$ which controls a measurement
    on a subset of qudits. This subset includes one of the two qudits of the state
    $\ket{\phi}$ associated with each edge adjacent to $v$, and when $v=r$ also includes
    qudit $i$ of $\psi$.
    The possible measurement settings are either a Weyl operator measurement on a single qudit,
    or a Bell basis measurement on a pair of qudits.
    (b) In the nonlocal game Verifier selects inputs $p_1, \dots, p_n$ and an accepting
    condition $A$ according to the nonlocal game $\Phi$. Verifier then
    randomly selects a rooted path $(v_{i1}, \dots, v_{il_i})$
    for each $i \in I$ and sends each $(i,v_{il_i})$
    the input corresponding to a Bell basis measurement,
    and $(i,v_{il_i})$ the Weyl measurement setting $p_i$.
    Verifier accepts the outputs $p_{i1}', \dots, p_{i(l_i-1)}', q_i$
    if $(q_1 - [p_1,p_1'], \dots, q_n - [p_n, p_n']) \in A$,
    where $p_i' := p_{i1}' + \dots + p_{il_i}'$.
    }
    \label{fig:teleportation_and_circuit}
\end{figure}

\begin{figure}
    \centering
    \begin{subfigure}{\textwidth}
        \centering
        \includegraphics[width=\textwidth]{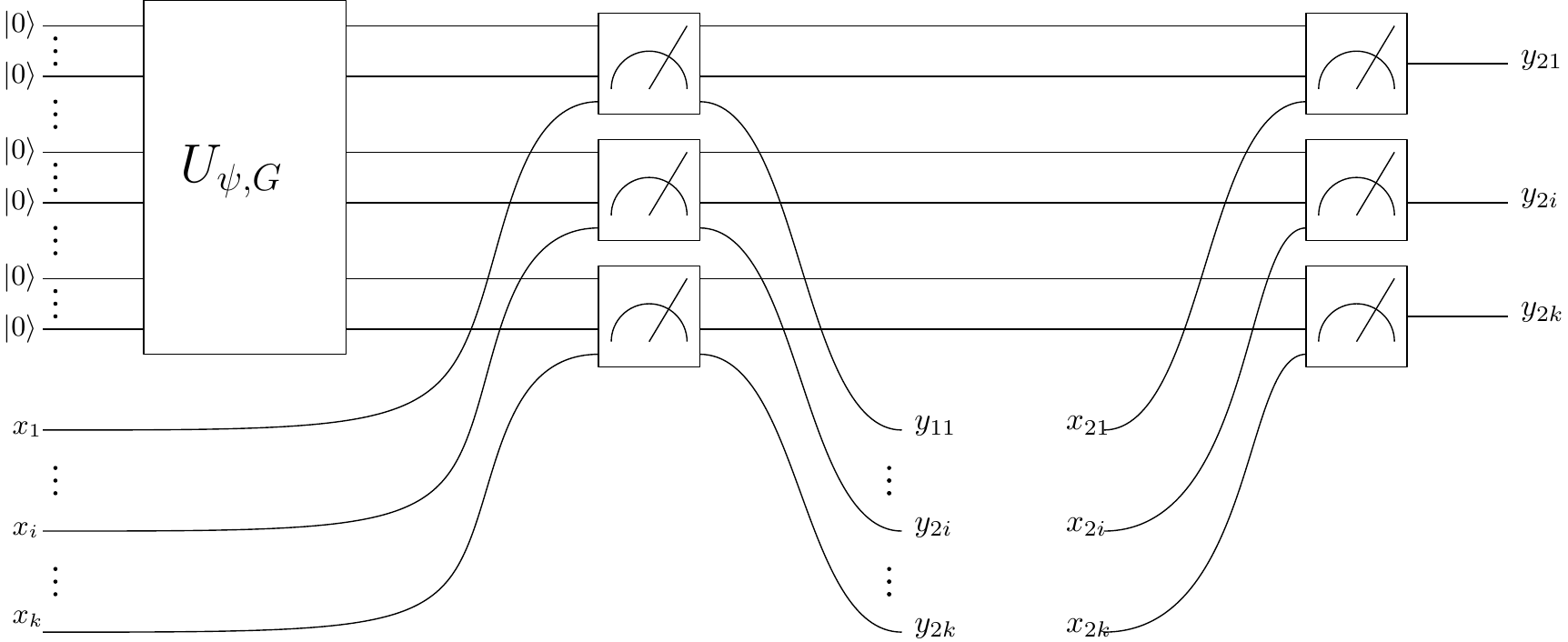}
        \caption{}
        \label{fig:my_label}
    \end{subfigure}
    \begin{subfigure}{\textwidth}
    \centering
    \begin{subfigure}{0.48\textwidth}
    \centering
        \includegraphics[width=1\textwidth]{resource-state-step-1_2.pdf}
        \caption{
        }
    \end{subfigure}
    \hfill
    \begin{subfigure}{0.48\textwidth}
    \centering
        \includegraphics[width=1\textwidth]{resource-state-step-2_2.pdf}
        \caption{}
    \end{subfigure}
\end{subfigure}
    \caption{
    Let $\psi$ be an $I$-qudit state and $G = (V, E, r)$ a rooted graph.
    The quantum circuit strategy (a) prepares a single instance of $\psi$
    and a maximally entangled pair of qudits $\ket{\phi} := \frac{\ket{00} + \ket{11}}{\sqrt{2}}$
    for each $i \in I$ and edge $e \in E$.
    The circuit first takes an input for $(i,v) \in I \times V$ which controls
    a non-destructive measurement on a subset of qudits,
    it then takes another round of inputs for each $(i,v) \in I \times V$
    which control destructive measurements on each subset of qudits.
    For each $(i,v)$ the subset of qudits which is measured
    includes one of the two maximally entangled qudits
    associated with each edge adjacent to $v$, and when $v=r$ also includes
    qudit $i$ of $\psi$.
    In the first round the measurement settings are either nothing
    or a Bell basis measurement on a pair of qudits.
    In the second round the possible measurements are either nothing
    or a conjugated measurement $W(p) M W(p)^\dagger$ on a single qudit,
    where $M$ is one of the measurement settings.
    (b) In the nonlocal game Verifier selects inputs $x_1, \dots, x_n$ and an accepting
    condition $A$ according to the nonlocal game $\Phi$. Verifier then
    randomly selects a rooted path $(v_{i1}, \dots, v_{il_i})$
    for each $i \in I$ and first sends each $(i,v_{il_i})$
    the input corresponding to a Bell basis measurement.
    If the outcomes of this are $p_{i1}, \dots, p_{i(l_i-1)}$
    Verifier sends $(i,v_{il_i})$ the input for the conjugated measurement
    $W(p_i) M_{x_i} W(p_i)^\dagger$,
    where $p_i' := p_{i1}' + \dots + p_{il_i}'$.
    Verifier accepts the output $y_1, \dots, y_n$ if $(y_1, \dots, y_n) \in A$.
    }
    \label{fig:teleportation_adn_circuit_two_rounds}
\end{figure}

\section{Structure of this text}
In Chapter 2 we present some technical background material on the 
sheaf-theoretic framework.
We then present the results on cohomology and circuits in Chapters 3 and 4 
respectively, and we conclude with some remarks in
Chapter 5.
\chapter{The sheaf-theoretic framework}

An early influential paper on contextually is
John Bell's famous paper on the Einstein-Podolsky-Rosen
(EPR) paradox \cite{bell_einstein_1964}. The ``paradox'' of EPR purportedly
showed that quantum mechanics should not be seen as a complete description of physical
reality \cite{einstein_can_1935}. Bell's insight could be understood to be that
the incompleteness highlighted by EPR is not simply a feature of quantum mechanics,
but of any physical theory that is consistent with the empirical predictions
of quantum mechanics. Other influential papers by
Kochen and Specker \cite{kochen_problem_1975}, Mermin \cite{mermin_extreme_1990},
and Greenberger-Horne-Zeillinger \cite{greenberger_bells_1990}, to mention a few.

This early work on contextuality focused on particular examples. Our
interest in contextuality stems from the wish to prove general connections between
contextuality and quantum advantage. It is therefore necessary to work with
a more abstract definition of contextuality. Our approach
uses the \emph{sheaf theoretic} framework of Abramsky and Brandenberger \cite{abramsky_sheaf-theoretic_2011}.
The sheaf theoretic approach is among several general definitions of contextuality.
For example, Robert Spekken's ontological models
framework \cite{spekkens_contextuality_2005},
Cabello, Severini, and Winter's graph theoretic approach 
\cite{cabello_graph-theoretic_2014}, and
the \emph{contextuality by default} approach of Dzhafarov, Kujala, and Cervantes
\cite{dzhafarov_contextuality-by-default_2015}.
Further work on the contextuality by default approach was
carried out by Dzhafarov, Kujala, and 
Cervantes \cite{dzhafarov_contextuality-by-default_2015}
and connections with psychology were investigated by Dzhafarov
and Kujala \cite{dzhafarov_contextcontent_2016},
to mention some. A graph theoretic approach that
refines that of Cabelo, Severini and Winter's is the approach
of Acín, Fritz, Leverrier, and Sainz \cite{acin_combinatorial_2015}.

The sheaf theoretic framework has proved useful for linking contextuality
to constraint satisfaction and database theory \cite{abramsky_logical_2012, abramsky_relational_2013}.

In this chapter, we give an introduction to contextuality
using the sheaf-theoretic framework, and we introduce several technical notions
that will be used in the following chapters.

\paragraph{Overview}
The two basic concepts in the sheaf theoretic framework are \emph{measurement scenarios}
and \emph{empirical models}. We introduce measurement scenarios
in Section \ref{section:measurement-scenarios} and empirical models
in Section \ref{section:empirical-models}.
In Section \ref{section:simulations} we define \emph{simulations},
a class of structure preserving transformations
between empirical models.
In Section \ref{section:cech-cohomology} we introduce
the \v{C}ech cohomology obstruction for contextuality.
In Section \ref{section:bell-inequalities} we define
non-local games.
In Section \ref{section:resource-inequalities} we
introduce the contextual fraction, and give an example
of a resource inequality.

\section{Measurement scenarios} \label{section:measurement-scenarios}
\begin{figure}
    \centering
    \includegraphics[width=0.5\textwidth]{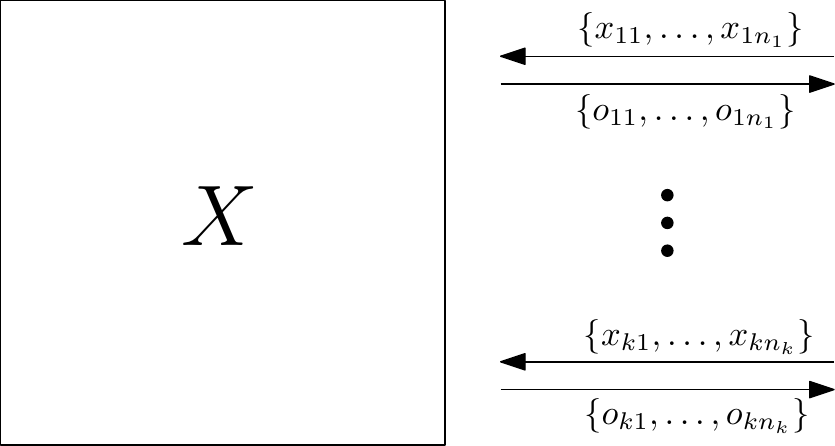}
    \caption{Let $(X, \mathcal{M}, O)$ be a measurement scenario.
        In an experimental run with $k$ rounds
        a sequence of contexts$\{x_{11}, \dots, x_{1n_1}\},
        \dots, \{x_{k1}, \dots, x_{kn_k}\} \in \mathcal{M}$
        satisfying Eq.\ (\ref{eq:contexts}) are performed,
        giving outcomes $\{o_{11}, \dots, o_{1n_i}\}, \{o_{k1}, \dots, o_{kn_k}\}$.}
    \label{fig:system-agent-diagram}
\end{figure}

In the sheaf theoretic approach of Abramsky and 
Brandenberger \cite{abramsky_sheaf-theoretic_2011} a
measurement scenario represents the abstract type of an experiment.
In this type of experiment some, but not necessarily all, 
combinations of measurements
can be performed together, either sequentially or in parallel (Figure \ref{fig:system-agent-diagram}). 
We will first give the general definition and then consider two 
types of scenarios: quantum scenarios (Section \ref{section:quantum-scenarios})
and multipartite scenarios (Section \ref{section:multipartite-scenarios}).

A measurement scenario is specified by a set of measurements, a family of subsets
called the \emph{measurement cover} specifying which measurements are compatible,
and a set of outcomes for each measurement.
\begin{definition}
A \emph{measurement scenario} is a tuple $(X, \mathcal{M}, \{O_x\}_{x \in X})$
where
\begin{itemize}
    \item $X$ is a set of \emph{measurements}.
    \item $\mathcal{M} \subset \mathcal{P}(X)$ is a family
    of subsets of measurements, called the \emph{measurement cover},
    such that:
        \begin{enumerate}
            \item $\mathcal{M}$ covers $X$: $\bigcup_{C \in \mathcal{M}} C = X$.
            \item $\mathcal{M}$ is downwards closed: If $C \in \mathcal{M}$
            and $C' \subset C$ then $C' \in \mathcal{M}$.
        \end{enumerate}
    \item $O_x$ is a set of \emph{outcomes}.
\end{itemize}
The elements of the measurement cover are called \emph{contexts}.
\end{definition}
\newcommand{\mscenario}{(X, \mathcal{M}, \{O_x\}_{x \in X})}

Let $(X, \mathcal{M}, O)$ be a measurement scenario.
Each context $C \in \mathcal{M}$ represents a set of compatible measurements
that can be performed either sequentially in any order, or in parallel.
We make the restriction that a measurement can only be performed once.
A sequence of contexts $C_1, \dots, C_n \in \mathcal{M}$ is valid if 
it has no repeated measurements and its union is a context:
\begin{align}
    \label{eq:contexts}
    \bigcup_i C_i \in \mathcal{M} \quad \text{ and }
    \quad
    C_i \cap C_j = \emptyset \quad \text{ for all } i \neq j
\end{align}

A joint outcome $s \in \prod_{x \in X'} O_x$ to a subset of measurements
is sometimes called a \emph{local section}.
This assignment is called the \emph{event sheaf}.
\begin{definition}
    Let $S = (X, \mathcal{M}, O)$ be a measurement scenario.
    The \emph{event sheaf}, denoted by $\mathcal{E}_S$,
    assigns to each $U \subset X$
    the set of \emph{local sections}
    $\mathcal{E}(U) := \prod_{x \in U} O_x$,
    and for each $V \subset U$
    restrictions $s \in \mathcal{E}(U)$
    to a local section $\res{s}{V}$
    by the usual functional restriction.
\end{definition}

Recall that a \emph{presheaf} on a topological space
$X$ is a contravariant function $F:X^\text{op} \to \text{Set}$.
Here $X$ is seen as a category with objects given by the open sets,
and morphisms inclusion. 
For each inclusion $U \subset V$ the map $F(U \subset V):F(V) \to F(U)$
is called the restriction map.
A \emph{sheaf} is a presheaf satisfying the following the \emph{sheaf condition}.
A compatible family for the open cover $\mathcal{U}$
is a family$\{f_U \in F(U)\}_{U \in \mathcal{U}}$ whose restrictions
on overlaps are compatible:
\begin{align}
    F(U \cap V \subset U)(f_U) = F(U \cap V \subset V)(f_V)
\end{align}
for all $U,V \in \mathcal{U}$.
The \emph{sheaf condition} states that any compatible family
arises as the family of restrictions
\begin{align}
    f_U = F(U \subset X)(f)
\end{align}
of some \emph{global section} $f \in F(U)$.

\subsection{Quantum scenarios} \label{section:quantum-scenarios}

The first example of a measurement scenario that we work
with arise from sets of \emph{projective measurements}.
A projective measurement is a family of projectors 
$M = \{M_o\}_{o \in O}$,
where $O$ labels the outcomes,
such that $\sum_{o \in O} M_o = I$.
Two measurements 
$M = \{M_o\}_{o \in O}$,
and 
$N = \{N_p\}_{p \in P}$ commute if their
projective elements commute:
\begin{align}
    M_oN_p = N_pM_o \quad \text{for all $o \in O$, $p \in P$}
\end{align}

A set of pairwise commuting projective measurements is said to be compatible.

\begin{example}
    Let $\pmb{M}$ be a set of projective measurements.
    $(\pmb{M}, \mathcal{M}, O)$
    is the measurement scenario with measurement
    cover the maximal subsets of pairwise commuting measurements,
    and outcomes $O$ given by the outcomes of each measurement.
\end{example}

The \emph{$n$-Pauli group}, denoted by $P_n$, is the group of $n$-qubit unitary operators
generated by the single-qubit Pauli operators
\begin{align*}
\hspace{-2.0 cm}
    I := \begin{bmatrix}
    1 & 0 \\
    0 & 1 \\
    \end{bmatrix}\quad
    \pauli{X}{} := \begin{bmatrix}
        0 & 1 \\
        1 & 0 \\
    \end{bmatrix} \quad
    \pauli{Y}{} := \begin{bmatrix}
        0 & -i \\
        i & 0 \\
    \end{bmatrix}
    \quad
    \pauli{Z}{} := \begin{bmatrix}
        1 & 0 \\
        0 & -1 \\
    \end{bmatrix}
\end{align*}
We denote the application of the
Pauli operator $\pauli{P}{} = \pauli{X}{}, \pauli{Y}{}, \pauli{Z}{}$
to qubit $i$ as
\begin{align}
    \pauli{P}{i} = I \otimes \dots \otimes \pauli{P}{} \otimes \dots \otimes I
\end{align}
recall that the single-qubit Pauli operators satisfy the commutativity relation
\begin{align}
    \pauli{X}{}\pauli{Y}{} = -\pauli{Y}{}\pauli{X}{}\\
    \pauli{X}{}\pauli{Z}{} = -\pauli{Z}{}\pauli{X}{}\\
    \pauli{Y}{}\pauli{Z}{} = -\pauli{Z}{}\pauli{Y}{}
\end{align}
It therefore follows that two $n$-qubit Pauli operators
commute if and only if they anti-commute at an even number of qubits.
An $n$-qubit Pauli operator specifies a projective measurement 
with outcomes $\{0,1\}$.
Two $n$-qubit Pauli operators commute, and therefore their projective measurements there also commute,
if and only if they anti-commute at an even number of qubits.
\begin{example}
    A quantum scenario $(\mathcal{O}, \mathcal{M}, \mathbb{Z})$
    is given by the set of two-qubit Pauli operators
    \begin{align}
        \begin{tabular}{c  c  c  c  }
        $\pauli{X}{1}$ & $\pauli{X}{2}$ & $\pauli{X}{1}\pauli{X}{2}$ \\ &&&\\
        $\pauli{Z}{2}$ & $\pauli{Z}{1}$ & $\pauli{Z}{1}\pauli{Z}{2}$ \\ &&&\\
        $\pauli{X}{1}\pauli{Z}{2}$ & $\pauli{Z}{1}\pauli{X}{2}$
        & $\pauli{Y}{1}\pauli{Y}{2} $\\\ &&& \\
        \end{tabular}
    \end{align}
    where each row and column make up a maximal context.
\end{example}

\subsection{Multipartite scenarios} \label{section:multipartite-scenarios}

The second type of measurement scenario that we consider represents
scenarios where measurements can be performed independently
at a number of locations (Figure \ref{fig:abstract-scenario}).
This is sometimes called a \emph{non-locality scenario}.
We prefer the terminology ``multipartite'' because it avoids
the implication that the locations are necessarily spatially separated.

A multipartite scenario is specified by a set of \emph{measurement sites} $I$,
for each measurement site $i$ a set of \emph{measurement settings} $X_i$,
and for each measurement setting $x \in X_i$ a set of 
\emph{measurement outcomes} $Y_{i,x}$.
Two measurements are compatible if and only if they belong to a different
measurement site.
First a comment about notation.
Recall that $\coprod_{i \in I} X_i$
is defined as
\begin{align}
    \coprod_{i \in I} X_i := \{(i, x) \mid i \in I, x \in X_i\}
\end{align}

\begin{definition}
    A \emph{multipartite scenario} $(I, X, Y)$
    is the measurement scenario $(\coprod_{i \in I} X_i, \mathcal{M}, Y)$,
    where the measurement cover $\mathcal{M}$
    is defined by
    \begin{align}
        \mathcal{M} := \{C \subset \coprod_{i \in I} X_i \mid (i,x), (i,x') \in C \Rightarrow x = x'\}    
    \end{align}
    and $Y_{i,x}$ is the set of outcomes for each $(i,x)$.
\end{definition}
\begin{figure}
    \centering
    \includegraphics[width=0.3\textwidth]{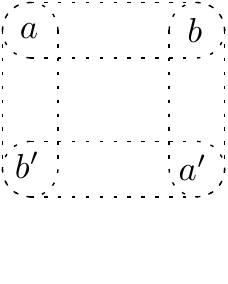}
    \caption{
        Consider a multipartite scenario with two measurement sites
        with measurement settings $\{a,a'\}$, $\{b,b'\}$ respectively.
        The maximal contexts are then $\{a,b\}, \{a,b'\}, \{b',a'\}, \{b,a'\}$.
    }\label{fig:abstract-scenario}
\end{figure}

\section{Empirical models} \label{section:empirical-models}

While a measurement scenario describes an experimental setup
Abramsky and Brandenberger introduced
the concept of an empirical model to
capture the empirical data generated in an experiment.
They introduced two types of empirical models,
capturing different types of data.
In Section \ref{section:probabilistic-models}
we define \emph{probabilistic} empirical models,
and \emph{probabilistic contextuality}.
In Section \ref{section:possibilistic-models}
we define \emph{possibilistic} empirical models,
and \emph{possibilistic contextuality}.

\subsection{Probabilistic empirical models} \label{section:probabilistic-models}

For any set $X$ write $\mathcal{D}(X)$ for the set of probability distributions 
over $X$.
It is sometimes convenient to write a probability distribution $d \in \mathcal{D}(X)$
as a \emph{formal sum} $d = \sum_{x \in X} d(x) \cdot x$ over the elements of $X$.
Note that $\mathcal{D}$ is a functor with action on functions $f:X \to Y$
given by the \emph{pushforward}
\begin{align}
    f_* : \mathcal{D}(X) \to \mathcal{D}(Y) &:: \sum_{x \in X} d(x) \cdot x \mapsto \sum_{x \in X} d(x) \cdot f(x)
\end{align}

Let $S = (X, \mathcal{M}, O)$ be a measurement scenario.
Consider the assignment $U \mapsto \mathcal{D}(\mathcal{E}_S(U)$ of the 
set of probability distributions over the local sections at a set of
measurements $U \subset X$.
$\mathcal{D} \circ \mathcal{E}_S : X^\text{op} \to \text{Set}$
is a presheaf, but not in general a sheaf.
For each $U \subset V$ and $d \in \mathcal{D}(\mathcal{E}_S(V)$
the restriction map is the marginal distribution
$\res{d}{U} \in \mathcal{D}(\mathcal{E}_S(U))$
\begin{align}
    \res{d}{U} := \sum_{s \in \mathcal{E}_S(V)} d_s \cdot \res{s}{U}
\end{align}
where $\res{s}{U}$ is the restriction of the section $s$ to $U$.

\begin{definition}
    Let $(X, \mathcal{M}, O)$ be a measurement scenario.
    A probabilistic empirical model
    is a family of probability distributions
    $e = \{e_C \in \dist(\mathcal{E}(C))\}_{C \in \mathcal{M}}$
    such that
    \begin{equation} \label{eq:local-compatibility}
        C \subset C' \Rightarrow \res{e_{C'}}{C} = e_{C}
    \end{equation}
\end{definition}

An experimental run for a scenario $S = (X, \mathcal{M}, O)$
is a sequence $(C_1, s_1), \dots, (C_n, s_n)$
where $C_1, \dots, C_n \in \mathcal{M}$ is a valid
sequence of contexts and $s_1, \dots, s_n$ are local
sections for the respective contexts.
If $e$ is an empirical model then the probability
of the run is
\begin{align}
    e(C_1, s_1, \dots, C_n,s_n) := e(C_1 \cup \dots \cup C_n)(s_1 \cup \dots \cup s_n)
\end{align}

Local compatibility is motivated by the ``no-disturbance'' principle in quantum mechanics. For multipartite scenarios, this is more commonly
called ``no-signalling''.
Let $\pmb{M}$ be a set of projective measurements and $\psi$ a state.
The measurement postulate of quantum mechanics defines an empirical model
$e$ for the scenario $(\pmb{M}, \mathcal{M}, O)$
given by
\begin{align}
    e(C)(s) := \lVert \big[ \prod_{M \in C} M_{s(M)} \big] \psi \rVert^2
\end{align}
The no-disturbance principle is the observation that the probability
distribution given by a context $C \subset \pmb{M}$ is independent
of which other compatible measurements it is performed in.
Marginalising from the maximal contexts, therefore, gives the correct behavior
for quantum measurements.

For a multipartite scenario, the measurement settings are not themselves
quantum measurements. We, therefore, have to choose
some interpretations of them as quantum measurements.
To ensure that the measurements are compatible
we do this on independent subsystems.
\begin{definition}
    Let $S = (I, X, Y)$ be a multipartite scenario.
    A \emph{quantum realised} empirical model
    $e_{\psi,\pi}$ is given by an $I$-qudit state
    $\psi$, and a single-qudit measurement $\pi(i,x)$
    for each $i \in I$, $x \in X_i$
    with outcomes $Y_{i,x}$.
    $e_{\psi,\pi}$ is defined by
    \begin{equation}
        e(C) = \sum_{s \in \mathcal{E}_S(C)} \lVert \big[ \bigotimes_{(i,x) \in C} 
            \pi(i,x)_{s(i,x)} \big] \psi \rVert ^2 \cdot s
    \end{equation}
\end{definition}

The following example of an abstract empirical model is taken from \cite{abramsky_sheaf-theoretic_2011}.
\begin{example}
    An empirical model for the two-partite scenario with measurement sites
    $A,B$ and measurement settings $\{a,a'\}, \{b,b'\}$ respectively
    and outcome $\{0,1\}$ is given by the table
    \begin{equation}
        \begin{tabular}{c c |c c c c c}
             A & B & (0,0) & (0,1) & (1, 0) & (1,1)\\
             \hline
             a & b & 1/2 & 0 & 0 & 1/2\\
             a & b' & 3/8 & 1/8 & 1/8 & 3/8\\
             a' & b & 3/8 & 1/8 & 1/8 & 3/8\\
             a' & b' & 1/8 & 3/8 & 3/8 & 1/8
        \end{tabular}
        \label{table:prob-table}
    \end{equation}
    The entries of the table give a probability
    to the outcomes of each maximal context. If the measurement $a$
    is performed on its own, then the probability of $0$
    is $1/2$, which can be seen by marginalising from either
    context $(a,b)$ or $(a,b')$.
    \begin{align}
        P(a = 0) &= P((a,b) = (0,1), (0,0)) = 1/2 + 0\\
        P(a = 0) &= P((a,b') = (0,1), (0,0)) = 3/8 + 1/8 = 1/2
    \end{align}
\end{example}

\paragraph{Contextuality}
Although $\mathcal{D} \circ \mathcal{E}_S : X^\text{op} \to \text{Set}$
is a presheaf, it is not necessarily a sheaf.
There can be compatible families that do not arise as a family
of restrictions from a global section.
Probabilistic contextuality is defined
as the failure of an empirical model to be explained as a family
of restrictions.
\begin{definition}
    Let $S = (X, \mathcal{M}, O)$ be a measurement
    scenario and $e$ an empirical model.
    $e$ is \emph{contextual} if there is no $d \in \mathcal{D}(\mathcal{E}_S(X))$
    such that for all maximal contexts $C \in \mathcal{M}_*$
    \begin{equation}
        e_C = \res{d}{C}
    \end{equation}
\end{definition}

An example of contextuality can therefore be thought of as a family
of locally compatible data that cannot be ``glued together''
to a consistent global picture of the data. We now give some examples.
\begin{example} \label{example:contextuality-1}
    Consider the multipartite scenario $(\{A,B\}, \{0,1\}, \{0,1\})$
    with measurement sites $A,B$ and two measurement settings
    each with two outcomes.
    The empirical model given by the following
    probability table is contextual.
    \begin{equation}
        \begin{tabular}{c c |c c c c c}
             A & B & (0,0) & (0,1) & (1, 0) & (1,1)\\
             \hline
             a & b & 1/2 & 0 & 0 & 1/2\\
             a & b' & 1/2 & 0 & 0 & 1/2\\
             a' & b & 1/2 & 0 & 0 & 1/2\\
             a' & b' & 0 & 1/2 & 1/2 & 0
        \end{tabular}
        \label{table:ppr-box}
    \end{equation}
\end{example}
\begin{proof}
    Suppose that there exists a probability
    distribution $d$ over the global sections
    whose restriction to each maximal context gives the table.
    From the probability table we have that
    each of the following events occur with certainty:
    \begin{align}
        d(a = b) = 1\\
        d(a = b') = 1\\
        d(a' = b) = 1\\
        d(a' = b') = 0\\
    \end{align}
    However, this is not possible because the constraints are mutually exclusive.
\end{proof}

\begin{exmp}[The GHZ model \cite{greenberger_bells_1990}] \label{example:contextuality-ghz}
    Consider the multipartite scenario
    $(\{0,1,2\}, \{ \mathbb{Z}_2 \}, \{\mathbb{Z}_2\})$
    with three measurement sites, two measurement settings at each measurement site,
    and two outcomes for each measurement setting.
    Let $\ket{\text{GHZ}} := \frac{1}{\sqrt{2}} (\ket{000} + \ket{111}$
    and $\pi$ the mapping of measurement setting $0$ to a Pauli $X$-basis measurement,
    and $1$ to a Pauli $Y$-basis measurement
    \begin{equation}
        \pi ::= (i,0) \mapsto X, (i,1) \mapsto Y
    \end{equation}
    The quantum realised empirical model 
    $e_\text{GHZ}$ given by measurements $\pi$ on the state
    $\ket{GHZ}$ is contextual.
\end{exmp}
\begin{proof}
    Write $X$ for the total set of measurements,
    suppose that there exists a probability distribution
    $d$ on the set of global sections $g:X \to \mathbb{Z}_2$
    that gives the empirical model $e_\text{GHZ}$.
    $\ghz$ is a $+1$-eigenstate of $\pauli{X}{1}\pauli{X}{2}\pauli{X}{3}$
    while it is a $-1$-eigenstate of 
        $\pauli{X}{1}\pauli{Y}{2}\pauli{Y}{3}$, 
        $\pauli{Y}{1}\pauli{X}{2}\pauli{Y}{3}$, and
        $\pauli{Y}{1}\pauli{Y}{2}\pauli{X}{3}$. 
    With the identification $\{-1,1\} \cong \zn{2}$ this means that 
    any global section $g:X \to \mathbb{Z}_2$
    satisfies 
\begin{align}
    \pauli{X}{1} \oplus \pauli{X}{2} \oplus \pauli{X}{3} 
        &= 0 \\
    \pauli{X}{1} \oplus \pauli{Y}{2} \oplus \pauli{Y}{3} 
        &= 1 \\ \pauli{Y}{1} \oplus \pauli{X}{2} \oplus \pauli{Y}{3} 
        &= 1 \\
    \pauli{Y}{1} \oplus \pauli{Y}{2} \oplus \pauli{X}{3} 
        &= 1
\end{align}
However, summing them together results in $0 = 1$. There is therefore
no global section $g$, and in particular no distribution $d$.
\end{proof}

Another famous example is the so-called CHSH model \cite{clauser_proposed_1969}. This illustrates an important
technique for proving contextuality. It uses an argument
involving an inequality satisfied by all non-contextual models.
\begin{example}[The CHSH model]
    Consider the multipartite scenario $(\mathbb{Z}_2, \{\mathbb{Z}_2\}, \{\mathbb{Z}_2\})$
    with two measurement sites, two measurement settings at each measurement site,
    and two outcomes for each measurement setting.
    The CHSH model, $e_\text{CHSH}$, is the empirical model
    realised by the state $\Phi := \frac{1}{\sqrt{2}}(\ket{00} + \ket{11})$
    and
    \begin{equation}
        \pi ::= \begin{cases}
            (0,0) \mapsto Z,\\
            (0,1) \mapsto X,\\
            (1,0) \mapsto A,\\
            (1,1) \mapsto B
        \end{cases}
    \end{equation}
    where
    \begin{align}
        \ket{a_0} &:= \cos{\frac{\pi}{8}}\ket{0} + \sin{\frac{\pi}{8}}\ket{1},
        \quad
        \ket{a_1} := -\sin{\frac{\pi}{8}}\ket{0} + \cos{\frac{\pi}{8}}\ket{1}\\
        \ket{b_0} &:= \cos{\frac{\pi}{8}}\ket{0} - \sin{\frac{\pi}{8}}\ket{1},
        \quad
        \ket{b_1} := \sin{\frac{\pi}{8}}\ket{0} + \cos{\frac{\pi}{8}}\ket{1}
    \end{align}
\end{example}
\begin{lemma}
    The CHSH model is contextual.
    For any non-contextual model the sum $\sum_{x,y} x_1 \oplus x_2 = y_1 \land y_2 \leq 0.75$.
    For $e_\text{CHSH}$ the sum is $\cos^2{\frac{\pi}{8}} \approx 0.85 > 0.75$.
\end{lemma}

\subsection{Possibilistic empirical models} \label{section:possibilistic-models}
\begin{figure}
    \centering
    \begin{subfigure}{0.45\textwidth}
        \centering
        \includegraphics[width=0.5\textwidth]{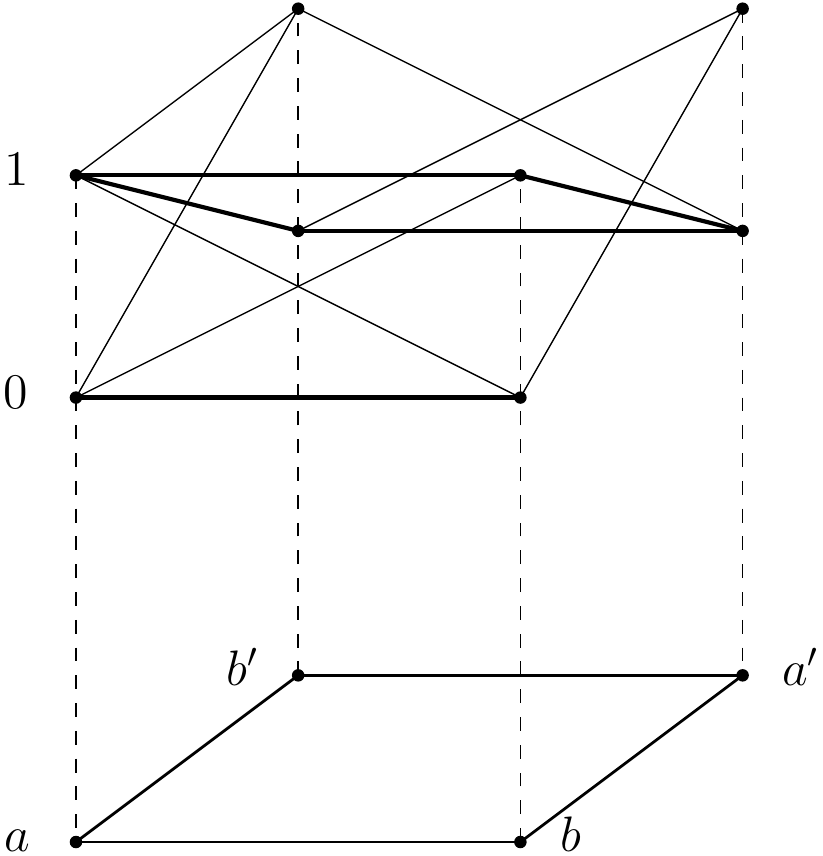}
        \caption{Hardy's model}
    \end{subfigure}
    \begin{subfigure}{0.45\textwidth}
        \centering 
        \includegraphics[width=0.5\textwidth]{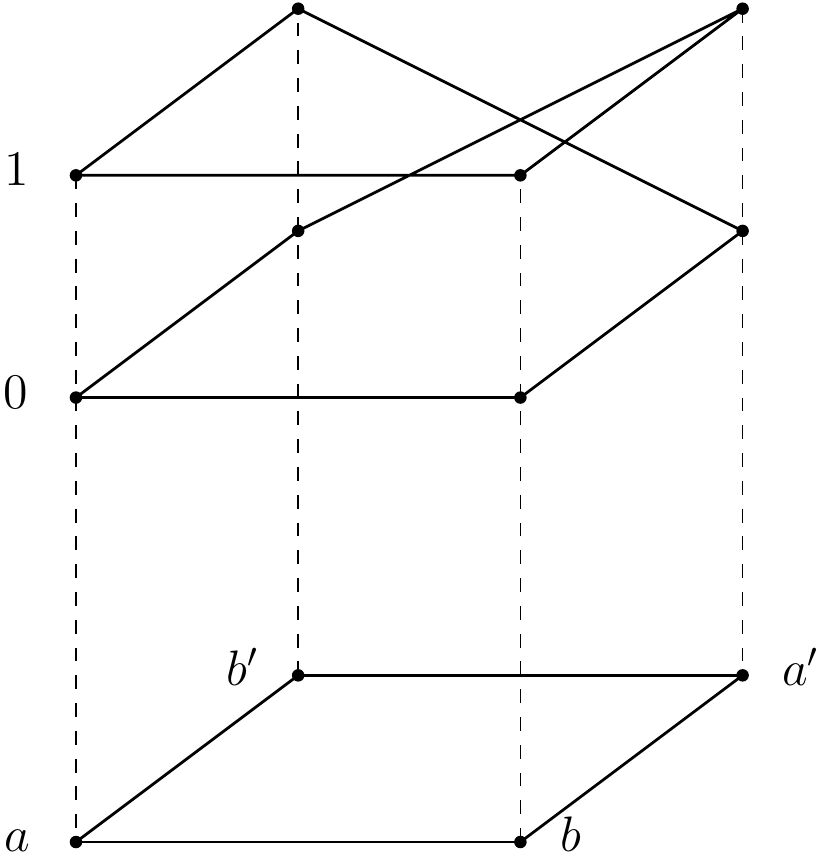}
        \caption{PR-box}
    \end{subfigure}
    \caption{Bundle diagrams of Hardy's model and the PR-box. The points above
    each measurement represents the two outcomes $0,1$ and the line
    segments the possible joint value assigmnets allowed by each
    model. A global section corresponds to a cycle visiting each
    measurement exactly once, for example the section $(a,a',b,b') = (1,1,0,0)$
    in (a). However, (a) is logically contextual at the local section
    $(a,b) = (0,0)$ as can be seen. It can be seen that
    (b) is strongly contextual because no cycle visiting each measurement
    once is possible.
    .}
    \label{fig:bundle-diagram-caru}
\end{figure}
Let $d \in \mathcal{D}(X)$ be a probability distribution over a set $X$.
The \emph{support} of $d$ is the subset $\supp{d} := \{x \in X \mid d(x) > 0\}$.

Observe that examples \ref{example:contextuality-1}
and \ref{example:contextuality-ghz} we do not refer to particular probabilities.
The proofs show that there is no global section $s \in \mathcal{E}_S(X)$
that is consistent with the \emph{support} of the models.
Given any probabilistic empirical model $e$ we can ``forget''
about the probabilities and only consider the family
of supports $\{\supp{e_C} \subset \mathcal{E}_S(C)\}_{C \in \mathcal{M}_*}$.
This is an example of a \emph{probabilistic} empirical model.
Probabilistic empirical models can be seen as presheafs in the following way.
\begin{definition}
    A \emph{possibilistic empirical model} $\pmodel:(X, \mcvx, O)$ is a
    subpresheaf of $\mathcal{E}_S$ such that 
    \begin{enumerate}
        \item Every compatible family for the measurement cover $\mcvx$ induces
            a global section.
        \item  $\pmodel$ is \emph{flasque beneath the cover}: 
            If $C,C' \in \mcvx$ and $C \subset C'$
            then every $s \in \mathcal{S}(C)$
            is the restriction of some $s' \in \mathcal{S}(C')$.
\end{enumerate}
\end{definition}

Note that although any probabilistic empirical model
gives rise to a probabilistic empirical model, the converse is not necessarily
true. Not all possibilistic empirical models is the support of a probabilistic
model.

Possibilistic empirical models can be represented as \emph{boolean tables}.
An example is given by the \emph{Hardy model}
\begin{equation}
    \begin{tabular}{c c |c c c c c}
         A & B & (0,0) & (0,1) & (1, 0) & (1,1)\\
         \hline
         a & b & 1 & 1 & 1 & 1\\
         a & b' & 0 & 1 & 1 & 1\\
         a' & b & 0 & 1 & 1 & 1\\
         a' & b' & 1 & 1 & 1 & 0
    \end{tabular}
    \label{table:possibilistic-hardy}
\end{equation}
and the \emph{Popescu-Rohrlich (PR) box}.
\begin{equation}
    \begin{tabular}{c c |c c c c c}
         A & B & (0,0) & (0,1) & (1, 0) & (1,1)\\
         \hline
         a & b & 1 & 0 & 0 & 1\\
         a & b' & 1 & 0 & 0 & 1\\
         a' & b & 1 & 0 & 0 & 1\\
         a' & b' & 0 & 1 & 1 & 0
    \end{tabular}
    \label{table:pr-box}
\end{equation}

In the possibilistic setting there are two natural forms of contextuality
that we can consider. Logical and strong.
\begin{definition}
    Let $\mathcal{S}$ be a possibilistic empirical model for a measurement scenario
    $(X, \mathcal{M}, O)$.
    We say that $\mathcal{S}$ is
    \begin{itemize}
        \item logically contextual at $s \in \pmodel(C)$
            if there is no global section $g \in \mathcal{S}(X)$ such that $\res{g}{C} = s$.
        \item logically contextual if $\mathcal{S}$ is logically contextual at \emph{some}
            local section. Otherwise, it is \emph{non-contextual}.
        \item strongly contextual if $\mathcal{S}$ has no global section: $\mathcal{S}(X) = \emptyset$.
    \end{itemize}
\end{definition}

We consider three types of contextuality forming a hierarchy:
\begin{equation}
    \text{Probabilistic Contextuality}
    <
    \text{Logical Contextuality}
    <
    \text{Strong Contextuality}
\end{equation}

Possibilistic empirical models can be represented
by bundle diagrams.
When an empirical model is represented as a bundle diagram
logical and strong contextuality have particularly elegant
interpretations. Logical contextuality is the failure
of a single line to extend to a path, and strong
contextuality is the property of every line extending to a path 
(Figure \ref{fig:bundle-diagram-caru}).
\begin{lemma}
    The Hardy model (\ref{table:possibilistic-hardy}) is logically
    contextual, but not strongly contextual.
    The PR-box is strongly contextual.
\end{lemma}
\begin{proof}
    This can be seen by inspecting the bundle diagrams.
\end{proof}

\subsection{State dependent contextuality}
Let $\pmb{M}$ be a set of projective measurements, $C \subset \pmb{M}$
a context of commuting measurements, and $\psi$ a state.
An outcome assignment $s$ for $C$
is \emph{consistent with $\psi$}
if $s$ has non-zero probability according to the measurement postulate
\begin{align}
    \big[ \prod_{M \in \pmb{M}} M_{s(M)} \big] \psi \neq 0
\end{align}
\begin{definition}
    Let $\pmb{M}$ be a set of projective measurements
    and $\psi$ a state. The \emph{state dependent model} $\mathcal{S}_{\pmb{M},\psi}$
    is the possibilistic empirical model
    \begin{equation}
        \pmodel_{\pmb{M},\psi}(V) := \{
            s \in \mathcal{E}_{\pmb{M}, \mathcal{M}, O}(V) \mid
                \text{$s$ is consistent with $\psi$} \}
    \end{equation}
\end{definition}

The set of measurements $\pmb{M}$ is state dependently contextual if $\pmodel_{\pmb{M}, \psi}$
is contextual for some state $\psi$. An example of a state-dependent
contextuality proof is the GHZ-example.

\subsection{State independent contextuality}
For some sets of quantum measurements, the state $\psi$ is not needed for the proof of contextuality.
\begin{definition}
    Let $\pmb{M}$ be a set of projective measurements.
    The \emph{state independent} model
    $\pmodel_X$ is defined at any \emph{below the cover}
    by 
    \begin{equation}
    \pmodel_X(V) := 
        \{s \in \mathcal{E}_{\pmb{M}, \mathcal{M}, O}(V) \mid \text{$s$ is consistent with \emph{some} 
            state} \}
    \end{equation}
\end{definition}

The set of measurements $\pmb{M}$ is said to be state-independently contextual
if $\pmodel_X$ is contextual. 
\begin{exmp}[Mermin's square \cite{mermin_extreme_1990}]
Let $\pmodel_X:(X, \mcvx, \zn{2})$ be the state independent model induced by the set of
measurements displayed in
\emph{Mermin's square}
\begin{center}
    \begin{tabular}{c  c  c  c  c }
    $\pauli{X}{1}$ & $\pauli{X}{2}$ & $\pauli{X}{1}\pauli{X}{2}$ & $I$ \\ &&&& \\
    $\pauli{Z}{2}$ & $\pauli{Z}{1}$ & $\pauli{Z}{1}\pauli{Z}{2}$ & $I$ \\&&&& \\
    $\pauli{X}{1}\pauli{Z}{2}$ & 
    $\pauli{Z}{1}\pauli{X}{2}$ & 
    $\pauli{Y}{1}\pauli{Y}{2}$ & $I$ \\&&&& \\
    
    $I$ & $I$ & $-I$ & \\
    \end{tabular}
\end{center}
Observe that the measurements displayed in any row or column
$M_1, M_2, M_3, M_4$ defines a context and furthermore 
satisfies $M_1 M_2 M_3 = M_4$, where $M_4 = \pm I$.
By Lemma 2.1
any local section $s \in \pmodel(C)$ therefore satisfies one of the following equations
\begin{align}
    \pauli{X}{1} \oplus \pauli{X}{2} \oplus \pauli{X}{1}\pauli{X}{2} &=  0\\ 
    \pauli{Z}{1} \oplus \pauli{Z}{2} \oplus \pauli{Z}{1}\pauli{Z}{2} &=  0\\ 
    \pauli{X}{1} \oplus \pauli{Z}{2} \oplus \pauli{X}{1}\pauli{Z}{2} &=  0\\ 
    \pauli{Z}{1} \oplus \pauli{X}{2} \oplus \pauli{Z}{1}\pauli{X}{2} &=  0\\ 
    \pauli{X}{1}\pauli{Z}{2} \oplus 
        \pauli{Z}{1}\pauli{X}{2} \oplus 
         \pauli{Y}{1}\pauli{Y}{2} &= 
            0\\ 
    \pauli{X}{1}\pauli{X}{2} \oplus 
        \pauli{Z}{1}\pauli{Z}{2} \oplus \pauli{Y}{1}\pauli{Y}{2} &= 
            1 
\end{align}
Any global section $g \in \pmodel_X(C)$ therefore simultaneously satisfies all equations.
However, these equations are mutually inconsistent. Summing together
all of the equations gives $0 = 1$, because each measurement appears in exactly
two equations.
$\pmodel_X$ is therefore strongly contextual.
\end{exmp}

\section{Simulations} \label{section:simulations}
\begin{figure}
    \centering
    \includegraphics[width=0.8\textwidth]{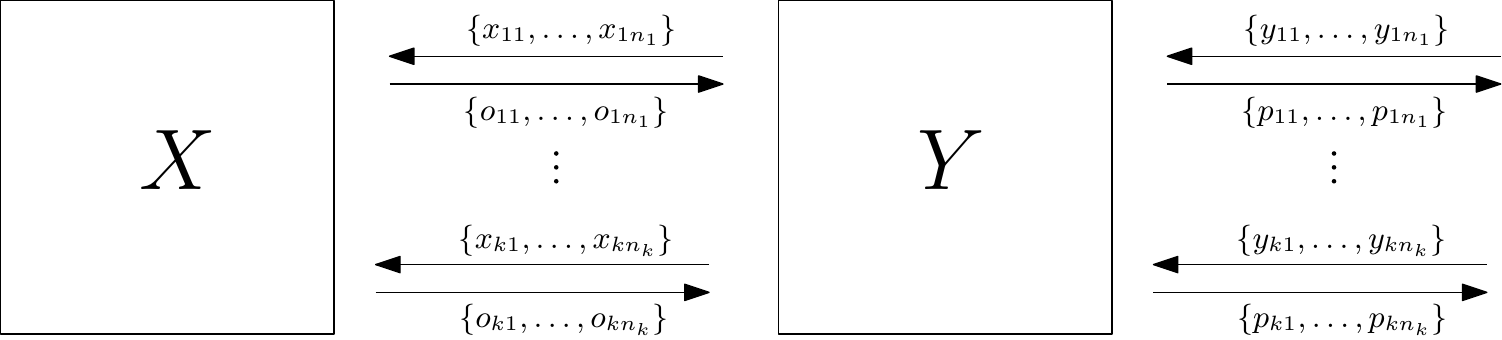}
    \caption{ Consider a setup involving two measurement scenarios $S = (X, \mcvx, O),
    T = (Y, \mathcal{N}, P)$. A simulation from $S$ to $T$
    maps each measurement $y \in Y$ to a measurement protocol on $S$,
    and each possible outcome of this protocol to an outcome $p \in P_y$.
    This induces a map on empirical models of $S$ to empirical models
    of $T$.}
    \label{fig:system-agent-simulation}
\end{figure}

The motivation behind introducing simulations is to equip
the sheaf-theoretic framework with a class of structure-preserving
transformations. The problem
of what the right notion of structure-preserving
transformation is for empirical models
was considered by Karvonen \cite{karvonen_categories_2019}.
The work of Karvonen later formed the basis
for the more developed idea of simulation
laid out by Abramsky, Barbosa, Karvonen, and Mansfield \cite{abramsky_comonadic_2019}. See also
For further work on simulations see the work
of Barbosa, Karvonen, and Mansfield \cite{barbosa_closing_2021}
and Abramsky, Barbosa, Karvonen, and Mansfield \cite{abramsky_simulations_2019}.

The notion of simulation that we present
here was defined by Abramsky et al.\ \cite{abramsky_comonadic_2019},
based on earlier work by 
The idea of studying examples of contextuality up to a class of 
structure-preserving transformations have also
been considered by others, 
for example Amaral et al.\  \cite{amaral_noncontextual_2018}.

Informally, a simulation $s$ from a measurement
scenario $S$ to another scenario $T$
describes how we can translate measurements on $T$ into measurements
on $S$, and outcomes of these measurements on $S$ into outcomes in $T$
(Figure \ref{fig:system-agent-simulation}).
This defines a map at the level of empirical models
called the \emph{pushforward}.

This section is structured as follows.
In Section \ref{section:simulations-1r} we introduce
the most simple example of a simulation,
deterministic single-round simulations.
We then introduce measurement protocols, describing adaptive
sequence of measurements.
We finally present the general notion of simulation.

\subsection{Single-round simulations} \label{section:simulations-1r}

We will now present the notion of simulation that
Karvonen \cite{karvonen_categories_2019} considered.

\begin{definition}
    A deterministic single-round simulation
    from a measurement
    scenario $S = (X_S, \mcvx_S, O_S)$ to
    another measurement scenario $T = (X_T, \mathcal{M}_T, O_T)$
    is a pair
    \begin{align}
        f &: X_T \to \mathcal{M}_S\\
        g &= \{g_y:\mathcal{E}_S(f(x)) \to (O_T)_x\}_{x \in X_T}
    \end{align}
    such that $\bigcup_{x \in C} f(x) \in \mathcal{M}_S$
    for every $C \in \mathcal{M}_T$.
\end{definition}

Let $e$ be an empirical model for the scenario $S$.
The pushforward $(f,g)_*(e)$ is then the empirical model
for the scenario $T$, defined by
\begin{align}
    (f,g)_*(e)(C) :=
        \sum_{s \in \mathcal{E}_S(\cup_{x \in C} f(x))} e(\cup_{x \in C} f(x))(s) \cdot (x \mapsto g_x(\res{s}{f(x)})
\end{align}
for all contexts $C$ in $T$. Suppose now
that $e$ is non-contextual, and therefore
a convex combination of global sections
\begin{align}
    e(C) = \sum_{\lambda \in \mathcal{E}(X_S)} p_x \cdot \res{\lambda}{C}
\end{align}
$(f,g)_*(e)$ is then a convex combination
\begin{align}
    (f,g)_*(e)(C) = 
        \sum_{\lambda \in \mathcal{E}(X_S)} p_\lambda \cdot 
        g_U(\res{\lambda}{f(U)})
\end{align}
If we define the global section 
$\lambda' \in \mathcal{E}_T(Y)$ by
\begin{align}
    \lambda'(y) = g_y(\res{\lambda}{C_y})
\end{align}
for each global section $\lambda \in \mathcal{E}(X_S)$,
then 
\begin{align}
    g_U(\res{\lambda}{f(U)}) = \res{\lambda'}{U}
\end{align}
hence $e$ is a convex combination of global sections, and hence non-contextual.

We, therefore, observe that if translate each measurement
$y \in Y$ into a fixed measurement $f(y) \subset \mathcal{M}$
that is independent of the measurement context that $y$
is performed in, then the induced map on empirical models
preserve non-contextuality.
A \emph{simulation} extends this in two ways, by allowing
for randomness and several rounds of measurements.

\subsection{Measurement protocols}

While Karvonen initially only considered single-round simulations
it is natural to consider simulations with more than
one round of measurements. To capture
this Abramsky, Mansfield, Barbosa, and Karvonen introduced
what they called \emph{measurement protocols}.

A \emph{measurement protocol} of length $n$ on a measurement
scenario $S$
\begin{align}
    C = \{C_1, \dots, C_i(s_1, \dots, s_{i-1}), \dots, 
    C_n(s_1, \dots, s_{n-1}) \in \mathcal{M}_X\}_{s_1 \in C_1, \dots, (s_1, \dots, s_{n-1}) \in C_{n-2}(s_1, \dots, s_{n-2})}
\end{align}
represents a deterministic strategy that someone can follow
to perform measurements on $S$, in an adaptive way. The measurement
setting $C_i$ is a function of the previous $i-1$ measurement outcomes.
We require that for all outcomes $s_1, \dots, s_{n-1}$
the sequence of contexts is valid, that is satisfying
Eq.\ \ref{eq:contexts}.
We write $\text{MP}_n(S)$ for the set of measurement protocols
of length $n$.
A \emph{run} of an adaptive measurement sequence
$\{C_i(s_1, \dots, s_{i-1})\}_{n \in \nats}$
is a sequence of contexts
and local sections $\{(U_i, s_i)\}_{n \in \nats}$
such that $s_i \in \mathcal{E}_S(U_i)$
and $U_i = C_i(s_1, \dots, s_{i-1})$ for all $i \in \nats$.
We write $\mathcal{E}_S(C)$ for the set of runs of a measurement
protocol $C$.

A set of measurement protocols
$\{C^j\}_{j \in J}$
that can be performed in parallel is said to be \emph{compatible}.
For any compatible set of measurement protocols
$\{C^j\}_{j \in J}$ their parallel product
is denoted by $\otimes_{i \in I} C^j$.

When a measurement protocol $C$ is performed the outcome
is a run. By the no-disturbance assumption
the probability of a given run can be defined by
\begin{align}
    e(C)(r) := e(U_1 \cup \dots \cup U_n)(s_1 \cup \dots \cup s_n)
\end{align}
where $r = \{(U_i, s_i)\}_{i=1}^N$ is a run.

\subsection{General simulations}

The idea of describing probabilistic simulations as probability
distributions over deterministic simulations is how
Karvonen described simulations. Although only for single-round
simulations. Later this was also how
Abramsky, Barbosa, Mansfield, and Karvonen formalised
probabilistic simulations with more than one round.

We now define deterministic $n$-round simulations,
and general simulations as probability distributions
over deterministic simulations.
\begin{definition}
    A deterministic simulation from a measurement scenario $S$
    to another measurement scenario $T$ of depth $n$ is
    a pair $(f,g)$ where
    \begin{itemize}
        \item $f : X_T \to \text{MP}_n(S)$ is a function
            such that $\{f(x)\}_{x \in C}$ is compatible for 
            all $C \in \mathcal{M}_T$.
        \item $g = \{g_x:\mathcal{E}(f(x)) \to (O_T)_x\}_{x \in X_T}$ is a family of functions.
    \end{itemize}
\end{definition}

Let $S = (X, \mathcal{M}, O)$, $T = (Y, \mathcal{N}, P)$ be two measurement
scenarios, and $t = (f,g) : S \to T$ a deterministic $n$-round
simulation. For each context $C \in \mathcal{N}$
write $f_C$ for the parallel product of the measurement
protocols $\{f(y)\}_{y \in C}$
\begin{align}
    f_C := \otimes_{y \in C} f(y)
\end{align}
The family of functions $\{g_y\}_{y \in C}$
defines a function
\begin{align}
    g_C : \mathcal{E}_S(f_C) \to \mathcal{E}_T(C)
\end{align}
defined at each $y \in C$ by the function $g_y$.
For any empirical model $e$ of $S$ we define the pushforward
$t_*(e)$ to be the empirical model for the scenario
$T$ given by the convex combination
\begin{align}
    t_*(e)_C := \sum_{r \in \mathcal{E}_S(f(C))} e(f(C))(r) \cdot g_C(r)
\end{align}
for each context $C$ of $T$.

\begin{definition}
    Let $S$ and $T$ be measurement scenarios. An \emph{$n$-round simulation}
    from $S$ to $T$, denoted $s:S \to T$,
    is a probability distribution over the set of deterministic $n$-round 
    simulations  from $S$ to $T$.
\end{definition}

We generalise the definition of the pushforward model
by taking the convex combination of empirical
models:
\begin{align}
    s_*(e) = \sum_{t:S \to T} s(t) \cdot t_*(e)
\end{align}
where $s = \sum_{t:S \to T} s(t) \cdot t$ is a simulation,
and $e$ is an empirical model.

\section{The cohomology of contextuality} \label{section:cech-of-contextuality}
A cohomology theory assigns an algebraic invariant to each element of some
class of objects. Cohomology theories are useful when
one can find invariants that can be computed easily, yet
characterise an important property of the objects we are studying.
An example is the \emph{simplicial cohomology} of a topological space.
Using for example triangulation we can compute the simplicial cohomology
of a large class of spaces. In topology this is an invaluable tool
for resolving many questions in a simple way.

In the sheaf-theoretic framework, a possibilistic empirical model
is a sheaf of sets $\mathcal{S} : X^\text{op} \to \textbf{Set}$. 
Contextuality is seen as the failure of a local section $s \in \mathcal{S}(C)$
to extend to a global section $g \in \mathcal{S}(X)$.
For presheafs of \emph{abelian groups} $\mathcal{F}:X^\text{op} \to \textbf{AbGrp}$ 
this transition from local to global is characterised by a 
\emph{cohomological obstruction}.
It is therefore natural to consider if this obstruction can detect contextuality.
Abramsky et al.\ showed that this is the case in a range of 
examples \cite{abramsky_cohomology_2012}, but also that it is not complete.
A more precise characterisation of the class of models where it is complete
was later given \cite{abramsky_contextuality_2015}.

In this section we present the \emph{\v{C}ech cohomology obstruciton}
of Abramsky et al.\ \cite{abramsky_cohomology_2012}. We first
define the cohomology groups of a cochain complex in Section \ref{section:c-groups}.
In Section \ref{section:cech-cohomology} we define the \emph{\v{C}ech cohomology}
groups of a presheaf of abelian groups.
In Section \ref{section:cech-of-contextuality} we define the obstruction
for contextuality.

\subsection{Cohomology groups of a cochain complex} \label{section:c-groups}

To define the cohomology groups of an object we use a family
of abelian groups connected by homomorphisms. This is called
a \emph{cochain complex}.
\begin{definition}
A \emph{cochain complex} is a sequence
\begin{equation}
    \vcenter{
        \hbox{
            \begin{tikzcd}
                0 \arrow[r, "d^{-1} := 0" above] &
                C^0 \arrow[r, "d^0" above] &
                C^1 \arrow[r, "d^1" above] &
                C^2 \arrow[r, "d^2" above] &
                \cdots
            \end{tikzcd}
        }
    }
\end{equation}
where $C^0, C^1, \dots $ are abelian groups, and
$d^0, d^1, \dots$ are homomorphisms
such that $d^{n+1} \circ d^n = 0$.
The elements of $C^n$ are known as the \emph{$n$-cochains} and
$d^n$ is the \emph{$n$'th coboundary map}.
$\text{im}(d^n)$ are the \emph{$n$-coboundaries} and $\text{ker}(d^{n+1})$
the \emph{$n$-cocycles}.
\end{definition}

The requirement that $d^{n+1} \circ d^n = 0$ equivalently
says that every $n$-coboundary is an $n$-cocycle,
$\text{im}(d^n) \subset \text{ker}(d^{n+1})$.
A sequence such that $\text{im}(d^n) = \text{ker}(d^{n+1})$
is said to be \emph{exact} at $n$. Cohomology measures the failure
of a sequence to be exact.
\begin{definition}
    The $n$-th cohomology group
    is the quotient of the coboundaries to the cocycles.
    $H^n := \text{im}(d^n) / \text{ker}(d^{n-1})$.
\end{definition}

The cohomology class $[x] \in H^n$ of a cocycle
$x$ can be thought of as an obstruction for $x$ to be a coboundary,
because $[x] = 0$ if and only if $x$ is a coboundary.

When we assign a cochain complex to some mathematical object
it is common to use a free construction. This free construction
loses some of the structure of the original object.
However, it can also be the case that the cohomology groups
capture some interesting feature of the object.
The classic example is simplicial cohomology, which
relates to the number of ``holes'' in a topological space.

\subsection{\v{C}ech cohomology} \label{section:cech-cohomology}
Let $X$ be a topological space, and 
$\mathcal{F}:X^\text{op} \to \mathbf{AbGrp}$ a presheaf of abelian
groups. In this section we define the \emph{\v{C}ech cohomology}
groups of $\mathcal{F}$. To do this we assign to $\mathcal{F}$
a cochain complex. This complex is defined using
an open cover $\mathcal{U}$ of $X$. First, we define an 
object encoding the combinatorial structure of the open cover.

\begin{definition}
    Let $\mathcal{U}$ be an open cover of a topological space
    $X$. The \emph{$n$-simplices} of the \emph{nerve}
    of $\mathcal{U}$, denoted by $\mathcal{N}_n(\mathcal{U})$, 
    are $n+1$-tuples of intersecting open sets.
    \begin{align}
        \mathcal{N}_n(\mathcal{U})
            &:= \{(U_0, \dots, U_n) \in \mathcal{U}^{n+1}
                \mid
                U_0 \cap \dots \cap U_n \neq \emptyset
                \}
    \end{align}
    The boundary maps $\partial_i:\mathcal{N}_{n+1}(\mathcal{U})
        \to \mathcal{N}_{n}(\mathcal{U})$
    remove the $i$'th open set:
    \begin{equation}
        \partial_i :: (U_0, \dots, U_{n+1}) \mapsto (U_0, \dots, U_{i-1},
            U_{i+1}, \dots, U_{n+1}) 
    \end{equation}
\end{definition}

Let $\supp{(U_0, \dots, U_n)} := \cap_{i} U_i$.

\begin{definition}
    Let $X$ be a topological space,
    $\mathcal{U}$ an open cover of
    $X$ and $\mathcal{F}$ a presheaf of abelian groups
    on $X$. The \v{C}ech cohomology
    group $H^n(\mathcal{F})$ is the $n$'th cohomology
    group of the cochain complex
    \begin{equation}
        \vcenter{
            \hbox{
                \begin{tikzcd}
                    0 \arrow[r, "d^{-1} := 0" above] &
                    C^0(\mathcal{U}, \sheaf{F}) \arrow[r, "d^0" above] &
                    C^1(\mathcal{U}, \sheaf{F}) \arrow[r, "d^1" above] &
                    C^2(\mathcal{U}, \sheaf{F}) \arrow[r, "d^2" above] &
                    \cdots
                \end{tikzcd}
            }
        }
    \end{equation}
    where 
    \begin{itemize}
        \item The $n$-cochains 
            $C^n(\mathcal{U}, \sheaf{F}) := \bigoplus_{U \in \mathcal{N}_n(\mathcal{U})} \sheaf{F}(\supp{U})$.
        \item The coboundary map $d^n(\omega)(U) := 
                    \sum_{i=0}^q {(-1)}^{i} 
                    \mathcal{F}(\supp{\partial_i U} \subset U)(\omega(\partial_i U))$
    \end{itemize}
\end{definition}
It can be verified that $d^{q+1} \circ d^q = 0$.

\subsection{The obstruction to the extension of a local section}

Let $S = (X, \mathcal{M}, O)$ be a measurement scenario, $\mathcal{S}$
a possibilistic empirical model, and $s_0 \in \mathcal{S}(C_0)$
a local section. We define the \v{C}ech cohomology
obstruction for $s_0$ to extend to a global section.

We can give $\mathcal{S}$ the structure of an abelian presheaf
by composing with the 
the functor
$F_\mathbb{Z}:\textbf{Set} \to \textbf{AbGrp}$
assigning to each set $X$ the \emph{free abelian group} on $X$,
that is, the group of formal linear combinations of $X$.
\begin{align}
    F_\mathbb{Z}(X) &:= \{\sum_{x \in X} k_x \cdot x \mid k_x \neq 0 \text{ for finitely many $x \in X$} \}\\
    F_\mathbb{Z}(f:X \to Y) &:= \sum_{x \in X} k_x \mapsto \sum_{x \in X} k_x \cdot f(x)
\end{align}

Let $\mathcal{F} := F_\mathbb{Z} \circ \mathcal{S}$. Note that
even though $\mathcal{S}$ is a sheaf, $\mathcal{F}$ is generally only
a presheaf.

The construction employs two auxiliary presheaves. For any
subset $U \subset X$ we define $\res{\mathcal{F}}{C_0}$ to be the
restriction of each $U$ to $U \cap C_0$,
and $\mathcal{F}_{\tilde{C_0}}$ assigns to each $U$
the subset of elements whose restriction to $U \cap C_0$ vanishes.
\begin{definition}
    Let $\mathcal{F}$ be an abelian presheaf and $C_0$ an open set.
    \begin{equation}
     \sheaf{F}_{\tilde{C_0}} :: U \mapsto \ker{\sheaf{F}(U \cap C_0 \subset U)}
     \quad
     \res{\sheaf{F}}{C_0} :: U \mapsto \sheaf{F}(C_0 \cap U)
     \end{equation}
\end{definition}

At any $U \subset X$ these presheaves are related to $\sheaf{F}$
by a sequence
\begin{center}
    \begin{tikzcd}
    0 \arrow[r] &
    \sheaf{F}_{\tilde{C_0}}(U) \arrow[r, hookrightarrow] &
    \sheaf{F}(U) \arrow[r, "\resmap{U}{U \cap C_0}"] &
    \res{\sheaf{F}}{C_0}(U) \arrow[r] &
    0
    \end{tikzcd}
\end{center}
which in fact is exact, because $\sheaf{F}$ is flasque beneath the cover.
When lifted to the level of cochain complexes it, therefore, gives rise to
a short exact sequence
\begin{center}
    \begin{tikzcd}
    0 \arrow[r] &
    C^*(\mcvx, \sheaf{F}_{\tilde{C_0}}) \arrow[r] &
    C^*(\mcvx, \sheaf{F}) 
    \arrow[r] &
    C^*(\mcvx, \res{\sheaf{F}}{C_0}) \arrow[r] &
    0
    \end{tikzcd}
\end{center}
Using standard techniques from homological algebra this short exact sequence
of cochain complexes induces
a \emph{long exact sequence} of cohomology groups
\begin{center}
    \begin{tikzcd}[column sep=small]
    0 \arrow[r] &
    H^0(\mcvx, \sheaf{F}_{\tilde{C_0}}) \arrow[d, phantom, ""{coordinate, name=Z}] \arrow[r] &
    H^0(\mcvx, \sheaf{F}) \arrow[r] &
    H^0(\mcvx, \res{\sheaf{F}}{C_0}) \arrow[dll,
        "\gamma",
            rounded corners,
            to path={ -- ([xshift=2ex]\tikztostart.east)
            |- (Z) [near end]\tikztonodes
            -| ([xshift=-2ex]\tikztotarget.west)
        -- (\tikztotarget)}] \\
    &
    H^1(\mcvx, \sheaf{F}_{\tilde{C_0}}) \arrow[r] &
    H^1(\mcvx, \sheaf{F}) \arrow[r] &
    H^1(\mcvx, \res{\pmodel}{C_0}) \arrow[r] &
    \cdots
    \end{tikzcd}
\end{center}
where $\gamma$ is the connecting homomorphism. For details about
this see for example \cite{weibel_introduction_1994}.
Using the identification
$\sheaf{F}(C_0) \cong H^0(\mcvx, \res{\sheaf{F}}{C_0})$ we define
the \emph{obstruction for $s_0$} to extend to a global section to be
$\gamma(1 \cdot s_0) \in H^1(\mcvx, \sheaf{F}_{\tilde{C_0}})$.

\begin{lemma}[\cite{abramsky_cohomology_2012}]
If the cover $\mcvx$ is \emph{connected}\footnote{
i.e. All pairs $C,C' \in \mcvx$ are connected by a sequence
$C_0 = C, C_1, C_2, \cdots, C_{n-1}, C_n = C'$
with $C_i \cap C_{i+1} \neq \emptyset$.
This assumption is harmless because non-connected components
are completely independent in terms of contextuality.
Incidentally, all of the scenarios we will consider are connected.
}
then $\gamma(1 \cdot s_0) = 0$
if and only if $1 \cdot s_0$ extends to a compatible family of
$F_\ints\pmodel$.
\end{lemma}

\begin{definition}
    Let $\pmodel:(X, \mcvx, O)$ be a possibilistic empirical model and $s_0 \in \pmodel(C_0)$
    a local section. The \emph{cohomological obstruction}
    to $s_0$ lifting to a global section is the cohomological obstruction
    to $1 \cdot s_0$ extending to a compatible family in $F_{\mathbb{Z}} \circ \mathcal{S}$.
\end{definition}

Observe that if $s_0$ extends to a global section $s$
in $\mathcal{S}$, then $1 \cdot s_0$ extends to a global
section $1 \cdot s$ in $\mathcal{F}$, hence the obstruction
is \emph{sound}.
\begin{lemma}
    The \v{C}ech cohomology obstruction for contextuality
    is sound:
    If $\gamma(g) \neq 0$
    then $\mathcal{S}$ is logically contextual
    at $s$.
\end{lemma}

\subsection{Generalised AvN arguments} \label{section:generalised-avn}

The \v{C}ech cohomology obstruction is not complete.
There are so-called \emph{false negatives},
contextual empirical models where the obstruction vanishes.
The approach detects contextuality in many cases, but an example
where the approach is not complete is Hardy's paradox.
Work has been carried out by Caru on understanding false negatives and
refining the approach \cite{caru_towards_2018}.

The \v{C}ech cohomology obstruction is complete for a large fragment
of models that can be described by generalised AvN models. Abramsky
et al.\ \cite{abramsky_contextuality_2015} take this terminology
from Mermin \cite{mermin_extreme_1990} who
used the term ``all versus nothing'' to describe his proof of contextuality.
These proofs can be understood as exhibiting an inconsistent set of equations
over $\zn{2}$ that is locally satisfied by the model.
The all versus nothing terminology was also used by for example
Cabello \cite{cabello_all_2001}.
The \v{C}ech cohomology obstruction is complete for the \emph{generalised AvN models},
the class of models that locally satisfies a system of inconsistent 
equations over any ring $R$ \cite{abramsky_contextuality_2015}.

\begin{definition}
    Let $(X, \mcvx, R)$ be a measurement scenario where $R$ is a ring.
    An \emph{$R$-linear equation} is a triple $(C, r, a)$
    where $C \in \mcvx$ is a context,
    $r:C \to R$ assigns a coefficient in $R$ to each $x \in C$,
    and $a \in R$ is a constant.
    A local section $s:C \to R$ \emph{satisfies} $(C, r, a)$
    if
    \begin{align}
        \sum_{x \in C} r(x) \cdot s(x) = a
    \end{align}
    where $\cdot$ denotes multiplication in $R$.
\end{definition}

Let $\mathcal{S}$ be an empirical model.
The \emph{$R$-linear theory of $\pmodel$} is the set of all $R$-linear equations 
that are consistent with $\pmodel$.
\begin{align}
    \text{Th}_R(\pmodel) := 
        \bigcup_{C \in \mcvx} 
        \{(C, r, a) \mid s \text{ satisfies $(C, r, a)$ for all $s \in \mathcal{S}(C)$}
        \}
\end{align}

\begin{definition}
$\pmodel$ is $\text{AvN}_R$ if its $R$-linear theory is \emph{inconsistent}.
i.e. there is no $s:X \to R$
such that $\res{s}{C} \models \phi$, for every context $C \in \mcvx$
and formula $\phi \in \text{Th}_R(\pmodel)$ at $C$.
\end{definition}
\begin{theorem}[\cite{abramsky_contextuality_2015}]
If $\pmodel$ is $\text{AvN}_R$ then $\gamma(1 \cdot s) \neq 0$ for all
$C \in \mcvx$ and $s \in \pmodel(C)$.
\end{theorem}

\section{Witnessing contextuality through cooperative games} \label{section:bell-inequalities}
\begin{figure}
    \centering
    \begin{subfigure}{0.45\textwidth}
        \centering
        \includegraphics[width=1\textwidth]{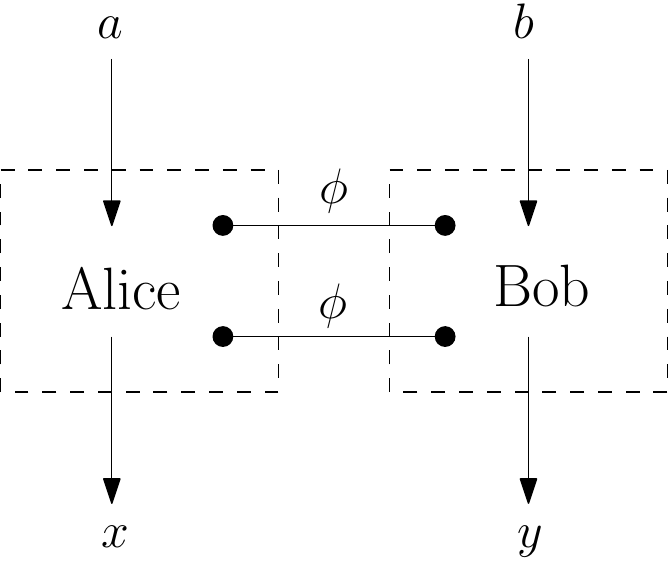}
        \caption{}\label{fig:msg-setup}
    \end{subfigure}
    \hfill
    \begin{subfigure}{0.45\textwidth}
        \centering
        \begin{tabular}{c | c | c | c}
                  & $a = 1$ & $a = 2$ & $a = 3$\\
             \hline
             $b = 1$ &  $X_1$ & $X_2$ & $X_1X_2$ \\
             \hline
             $b = 2$ & $Z_2$ & $Z_1$ & $Z_1Z_2$ \\
             \hline
             $b = 3$ & $-X_1Z_2$ & $-Z_1X_2$ & $Y_1Y_2$
        \end{tabular}
        \caption{}
    \end{subfigure}
    \caption{
        The Magic Square game. 
        Alice and Bob each hold one of the two qubits of two maximally entangled states $\phi$.
        Verifier sends Alice and Bob $a,b \in \{1,2,3\}$. Alice performs the three observables
        $(M_1,M_2,M_3)$ in column $b$ of (b) and Bob performs the observables
        in row $a$. Given outcome $x = (x_1,x_2,x_3), y = (y_1,y_2,y_3)$ they win
        if $x_1 \oplus x_2 \oplus x_3 = 1$ and $y_1 \oplus y_2 \oplus y_3 = 0$,
        and $x_b = y_a$.
    }
    \label{fig:magic_square_game}
\end{figure}
There are different ways of proving that an empirical model is contextual.
For example, using inequalities \cite{clauser_proposed_1969,bell_einstein_1964}
or using systems of logical formulas \cite{mermin_extreme_1990}.
A systematic treatment of contextuality proofs is given by
Abramsky and Hardy \cite{abramsky_logical_2012}.
It is well known that certain contextuality
proofs can be recast as cooperative games known as \emph{non-local games}.
For example, the Magic Square game (Figure \ref{fig:magic_square_game})
\cite{cleve_consequences_2010}.

In this section, we first define cooperative games and non-local games.
We then explain that simulations can be used to translate
a cooperative game from one scenario to another.

\subsection{Cooperative games} \label{section:co-games}
Let $S = (I, X, Y)$ be a multipartite scenario.

A \emph{game} is played by $I$, thought of as players, against \emph{Verifier}.
A game is played over one or more rounds of the following form.
Verifier sends each player $i \in I'$, in a subset $I' \subset$,
a value $x_i \in X_i$, and each player responds with a value $y_i \in Y_x$.
We assume that the players are not allowed to communicate
and that each player is sent at most one value.
A strategy for Verifier is therefore an $n$-round
measurement protocol $C$, and a strategy for the players
is an empirical model $e$.

At the beginning of each game Verifier randomly selects a strategy $C$ 
and an \emph{accepting condition}
$A \subset \mathcal{E}_S(m)$. The goal of the players
is to maximize the probability that their responses
$s_1, \dots, s_n$ satisfies the accepting condition.
\begin{definition}
    Let $S = (I, X, Y)$ be a multipartite measurement scenario.
    An $n$-round \emph{game} is a convex combination 
    $\Phi = \sum_{C \in \text{MP}_n(S), A \subset \mathcal{E}_S(C)} \Phi_{C,A} \cdot (m, A)$.
    The \emph{success probability} of an empirical model $e$ is
    \begin{align}
        p_S(e, \Phi) := \sum_{C \in \text{MP}_n(S), A \subset \mathcal{E}_S(C)}     
            \Phi_{C,A} e(C)(A)
    \end{align}
\end{definition}

A \emph{non-local game} is a single-round cooperative game along
with a quantum strategy exceeding that of any non-contextual strategy.
\begin{definition}
    Let $S = (I, X, Y)$ be a multipartite scenario.
    A \emph{non-local game} is a pair $(e, \Phi)$ 
    where $e$ is a quantum realised empirical model,
    and $\Phi$ is a single-round game, such that
    there exists a $\gamma$ such that
    for all non-contextual empirical models $e_\text{NC}$
    \begin{align}
        p_S(e_\text{NC}, \Phi) \leq \gamma < p_S(e, \Phi)
    \end{align}
    the least such $\gamma^*$, is called the classical upper bound.
\end{definition}

A well-known example is the Greenberger-Horne-Zeillinger (GHZ) game \cite{greenberger_bells_1990}.
\begin{exmp}
    The GHZ game is played by three players $A,B,C$. Verifier
    selects inputs $x_A,x_B,x_C \in \mathbb{Z}_2$ with uniform probability.
    The players win if their outputs $y_A,y_B,y_C \in \mathbb{Z}_2$
    satisfies
    \begin{equation}
        A_\text{GHZ}(x_A, x_B, x_C)(y_A, y_B, y_C) \iff x_A \lor x_B \lor x_B = y_A \oplus y_B \oplus y_C
    \end{equation}
    A winning quantum strategy is given where each player performs
    a Pauli $X$ measurement if the input is $0$ and a Pauli $Y$
    measurement if the input is $1$. However, any non-contextual
    strategy solves the game with at most $3/4$.
\end{exmp}

\subsection{The pullback of a game}
\begin{figure}
    \centering
    \includegraphics[width=\textwidth]{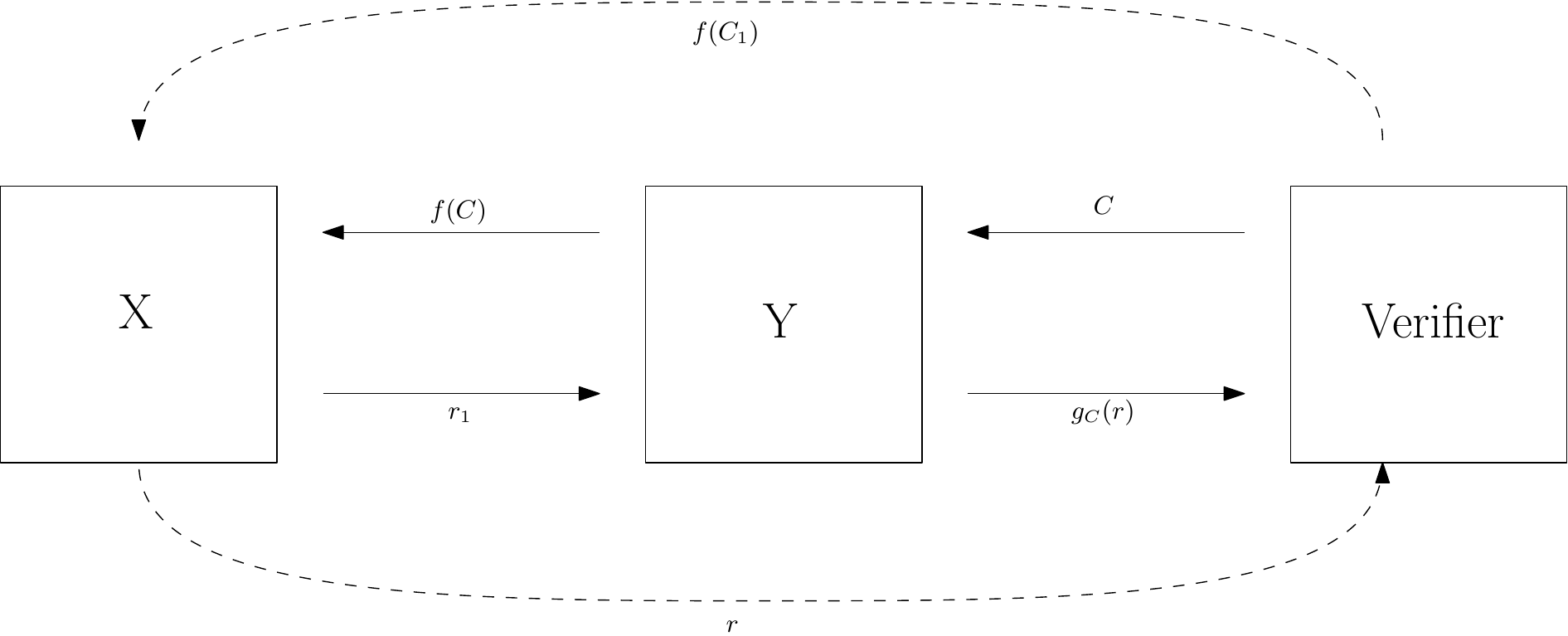}
    \caption{The pullback of a game $\Phi$.
    Consider a game where Verifier plays against a set of players $Y$. Verifier sends $Y$ a context $C$,
    the players then interact with another set of players $X$
    through a measurement protocol $f(C)$.
    If the result of $f(C)$ is a run $r$ then they respond with
    $g_C(r)$ to Verifier. Verifier accepts
    if $g_C(r) \in A$ satisfies the accepting condition.
    This game is equivalent to the game
    where Verifier interacts directly with $X$ by performing
    the measurement protocol $f(C)$
    and accepts a run $r$ if $r \in g_C^{-1}(A)$.
    }
    \label{fig:pb-diagram}
\end{figure}

Let $S$ and $T$ be measurement scenarios, and $s:S \to T$ an $n$-round
simulation.
We have explained that $s$ induces a map on empirical models
going from $S$ to $T$, called the pushforward.
Simulations also have a natural action on games
(Figure \ref{fig:pb-diagram}).
The \emph{pullback} $s^*$ maps $k$-round games of $T$ to $kn$-round games on $S$.

The defining property of the pullback is that
for any empirical model $e$ of $S$ and game $\Phi$ of $T$, the success probability
of $e$ on $s^*(\Phi)$ is the success probability of $s_*(e)$
on $\Phi$:
\begin{align}
    p_S(e, s^*(\Phi)) = p_S(s_*(e), \Phi)
\end{align}

We can define the pullback directly as follows.
\begin{definition}
    Let $S$ and $T = (Y, \mathcal{N}, P)$ be measurement scenarios,
    $s:S \to T$ an $n$-round simulation,
    and $\Phi$ a single-round game.
    The \emph{pullback} $s^*(\Phi)$
    is the $n$-round game for the scenario
    $S$, defined as
    \begin{align}
        s^*(\Phi) := \sum_{(f,g):S \to T, C \in \mathcal{N}, A \subset     
            \mathcal{E}_T(C)} s(f,g) \Phi_{C,A} \cdot (f_C, g_C^{-1}(A))
    \end{align}
    where $\Phi_{C,A}$ is the probability of Verifier
    selecting the context $C$ and accepting condition $A$,
    $s(f,g)$ is the probability of the deterministic
    simulation $(f,g)$
    given by $s$,
    and $f_C \in \text{MP}_n(S)$,
    $g_C:\mathcal{E}_S(f_C) \to \mathcal{E}_T(C)$
    are the maps defined by the deterministic simulation.
\end{definition}

\section{The contextual fraction} \label{section:resource-inequalities}
\begin{figure}
    \centering
    \includegraphics{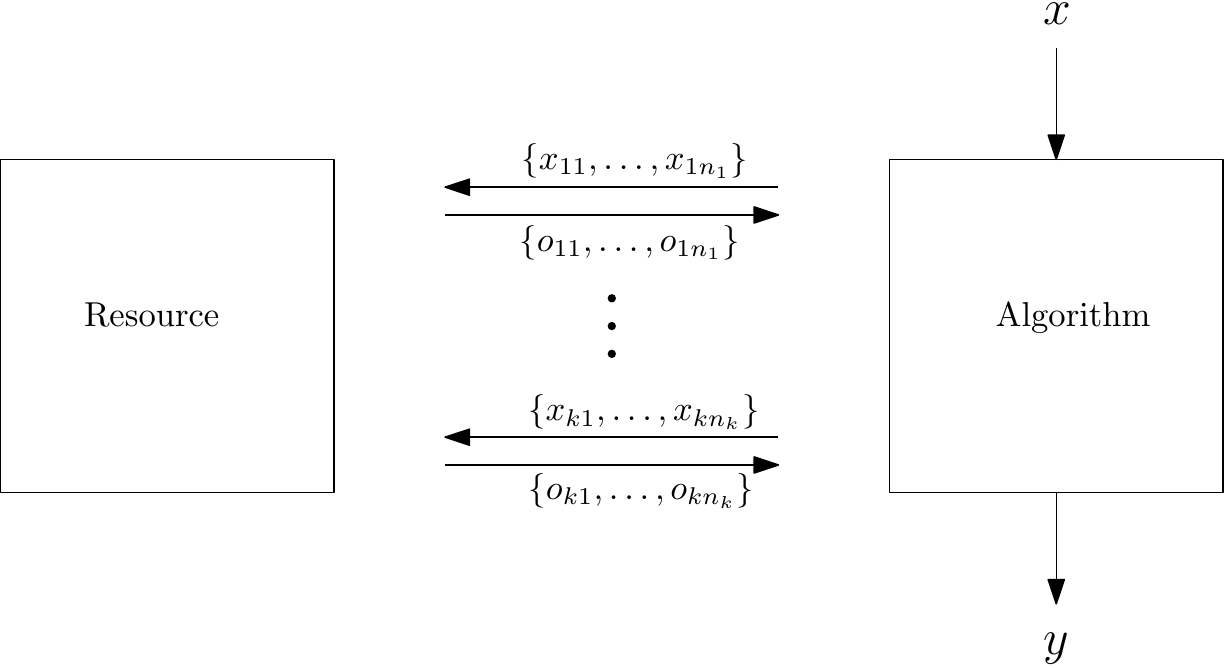}
    \caption{In the resource view we think about an empirical model
    as a resource that can be consumed by a classical algorithm
    solving a computational problem.}
    \label{fig:resource-situation}
\end{figure}

The \emph{contextual fraction} is a measure of contextuality
introduced by Abramsky, Barbosa, and Mansfield \cite{abramsky_contextual_2017}.
See Barbosa, Douce, Emeriau, Kashefi, and Mansfield
for a generalisation of the contextual fraction for continuous
variables \cite{barbosa_continuous-variable_2022}.

The contextual fraction was motivated by
the consideration of situations where a source of contextuality
is consumed to solve a computational problem 
(Figure \ref{fig:resource-situation}).
Abramsky, Barbosa, and Mansfield observed that several results of this type
can be refined to give \emph{resource inequalities} on the form
\begin{equation}
    p_F \geq (1 - \cf{e}) v(f)
    \label{eq:resource-ineq}
\end{equation}
relating the degree of failure $p_F$
in a situation where an empirical model $e$ is 
consumed to solve a problem $f$,
to the contextual fraction $\cf{e}$
and some intrinsic measure $v(f)$ of the hardness of $f$.

An example of such a resource inequality arises from measurement-based quantum
computing (MBQC). In MBQC a classical control computer
that can only perform mod-2 linear computations interacts with
an empirical model. Raussendorf \cite{raussendorf_contextuality_2013}
building on Anders and Browne \cite{anders_computational_2009}
showed that any MBQC that can compute a non mod-2 linear function
requires a strongly contextual empirical model.
This was later refined into a resource inequality
relating the contextual fraction to the likelihood
of an MBQC computing a non-mod 2 linear function.

In this section, we first define the contextual fraction
and then show that non-local games give another example of a resource inequality.
The contextual fraction is a measure of contextuality
that can be seen as the fraction
of an empirical model that cannot be explained by a non-contextual
model.
\begin{definition} \label{def:contextual_fraction}
    Let $e$ be an empirical model. 
    The \emph{non-contextual fraction}
    of $e$, denoted by $\text{NCF}(e)$, is the greatest $\epsilon$
    such that $e$ is a convex combination of a non-contextual
    empirical model $e'$ and another empirical model $e''$.
    \begin{equation}
        e = \epsilon \cdot e' + (1 - \epsilon) \cdot e''
    \end{equation}
    The \emph{contextual fraction}, denoted by $\text{CF}(e)$,
    is defined as $1 - \text{NCF}(e)$.
\end{definition}

Let $S$ be a measurement scenario
and $\Phi$ a game such that the success
probability of any non-contextual empirical model
is at most $\gamma$.
The violation of $\gamma$
by any empirical model $e:S$
is at most $\cf{e}$.
\begin{lemma}
    Let $(\Phi,e)$ be a non-local game with bound $\gamma$.
    For any empirical model $e'$
    the violation of $\gamma$ by $e'$ is bounded by
    the classical limit and the contextual fraction.
    \begin{equation}
        p_S(e, \Phi) \leq \gamma + \cf{e}
    \end{equation}
\end{lemma}
\begin{proof}
    Let $e$ be an empirical model.
    We can write $e$ as a convex combination
    \begin{align}
        e = \cf{e} \cdot e' + (1-\cf{e}) \cdot e_\text{NC}
    \end{align}
    where $e_\text{NC}$ is non-contextual.
    The success probability of $e$ is then
    \begin{align}
        p_s(e, \Phi) = \cf{e} p_S(e', \Phi) + (1-\cf{e}) p_S(e_\text{NC}, \Phi)
    \end{align}
    The success probability of $e'$ is at most one,
    and the success probability of $e_\text{NC}$ at most $\gamma$.
    Therefore
    \begin{align}
        p_S(e, \Phi) &\leq \cf{e} + (1-\cf{e}) \gamma\\
                     &\leq \gamma + \cf{e}
    \end{align}
\end{proof}
\chapter{Comparing two obstructions for contextuality}
Cohomological invariants can be a powerful mathematical tool. Abramsky
et al.\ \cite{abramsky_cohomology_2012, abramsky_contextuality_2015} 
showed that a cohomological invariant based on \v{C}ech cohomology 
can detect contextuality in a range of examples. However,
the \v{C}ech cohomology approach is generally not complete.
There are instances of contextuality, called ``false negatives'',
where the cohomological obstruction vanishes.
In this chapter, we compare the \v{C}ech cohomology approach
to a different cohomological approach for detecting contextuality.

The \emph{topological approach} of Okay, Bartlett, Roberts, 
and Raussendorf \cite{okay_topological_2017}
studies certain sets of quantum measurement operators.
Recall that for any dimension $d \geq 2$ the single-qudit
Weyl operators are a set of $d^2$ unitary operators
generalising the Pauli operators.
The generalised $n$-qudit Pauli group is the group
of operators generated by $n$-fold tensor products
of single-qudit Weyl operators.
\begin{definition}
    For any dimension $d \geq 2$, and $p_1, p_2 \in \mathbb{Z}_d^2$ 
    the \emph{single-qudit Weyl operator} $W(p_1, p_2)$
    is defined by
    \begin{align}
        W(p_1, p_2) := \ket{j} \mapsto \omega^{j p_2}\ket{j + p_1}
    \end{align}
    where $\omega = e^{2\pi i / d}$. The \emph{$n$-qudit generalised Pauli group
    $P_{n,d}$} is the group of operators
    on the form
    \begin{align}
        \omega^q W(p_{11}, p_{12}) \otimes \dots \otimes W(p_{n1}, p_{n2})
    \end{align}
    where $p_{11}, p_{12}, \dots, p_{n1}, p_{n2} \in \mathbb{Z}_d$.
\end{definition}

The topological approach studies sets of $n$-qudit Weyl operators
that contain the identity operator, is closed under
commuting products and $\omega^q$-phases.
\begin{definition}
    A set of $n$-qudit generalised Pauli operators
    $\mathcal{O} \subset P_{n,d}$ is \emph{closed} if
    \begin{enumerate}
        \item $\mathcal{O}$ contains the identity operator: $I \in \mathcal{O}$.
        \item $\mathcal{O}$ is closed under commuting products: 
            If $O_1, O_2 \in \mathcal{O}$ and 
                $O_1O_2 = O_2O_1$ then $O_1O_2 \in \mathcal{O}$.
        \item $\mathcal{O}$ is closed under $\{\omega^k\}$-phases:
            If $O \in \mathcal{O}$ and $k \in \mathbb{Z}_d$ then
            $\omega^k O \in \mathcal{O}$.
    \end{enumerate}
\end{definition}

\begin{figure}
    \centering
    \hfill
    \begin{subfigure}{0.3\textwidth}
        \includegraphics[width=1\textwidth]{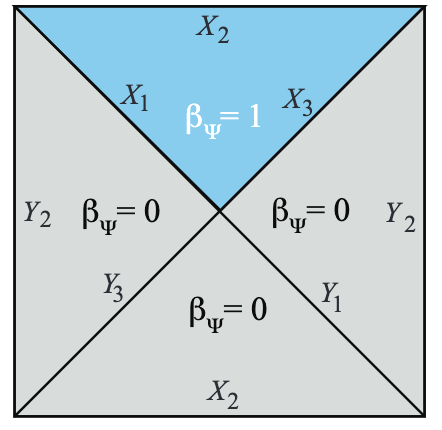}
        \caption{}
        \label{fig:state-dep-mermin-star-topo}
    \end{subfigure}
    \hfill
    \begin{subfigure}{0.3\textwidth}
        \includegraphics[width=1\textwidth]{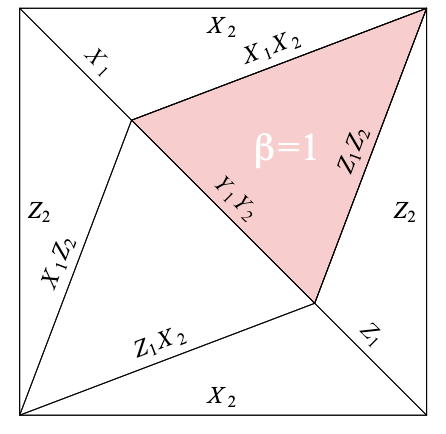}
        \caption{}
        \label{fig:mermin-square-topo}
    \end{subfigure}
    \caption{Examples of classifying spaces taken from Okay et al. \cite{okay_topological_2017}. (a) is the GHZ proof,
    (b) Mermin's square.
    }
    \label{fig:top_realisation_examples}
\end{figure}

For any closed set of Weyl operators Okay et al.\ defines a topological
space (Figure \ref{fig:top_realisation_examples}). 
They show that key properties of the set of operators
are reflected in the topology of this space. One of their
results is that both state-dependent and state-independent
contextuality can be detected by the non-vanishing of a cohomology
class. Recall that each $n$-qudit Weyl operator $W(p_1, p_2) \neq I$
has $d$ distinct eigenvalues $\omega^0, \dots, \omega^{d-1}$.
Under the identification $\omega^i \mapsto i$
each Weyl operator defines a projective measurement
with outcomes $\mathbb{Z}_d$.
A state-dependent or state-independent contextuality
proof is a proof that either the state-dependent
or state-independent empirical models
\begin{equation}
    \mathcal{S}_{\mathcal{O}}:(\mathcal{O}, \mathcal{M}, \mathbb{Z}_d),
    \quad 
    \mathcal{S}_{\mathcal{O}, \psi}:(\mathcal{O}, \mathcal{M}, \mathbb{Z}_d)
\end{equation}
are contextual.

In this chapter, we consider the following problem. What is the minimal structure
required to define the topological obstruction at the level of empirical models.
Secondly, assuming that the topological obstruction can be defined,
are there instances where the \v{C}ech cohomology obstruction
vanishes, but the topological obstruction does not?

\subsection{Structure of chapter} 
In Section \ref{section:bundles} we introduce bundles over commutative partial 
monoids and we prove the splitting lemma, relating left splittings, 
right splittings, and trivialisations.
In Section \ref{section:relative-cohomology} we define the cohomology
of a commutative partial monoid. We show that the problem
extending a local right splitting of a bundle is characterised by
a cohomological obstruction.
In Section \ref{section:bundle-scenarios} we introduce a class of measurement
scenarios and empirical models generalising closed sets of Weyl operators.
We show that for any such empirical model
a cohomological obstruction can be defined.
Finally, in Section \ref{section:comparison} we show that this
obstruction is not stronger than the \v{C}ech cohomology obstruction.

\newcommand{\gpauli}[2]{P(#1,#2)}
\section{Bundles over commutative partial monoids} \label{section:bundles}
In this chapter, we are working with commutative groups, monoids, and partial monoids.
We will therefore ommit the word commutative to avoid unnecessarily complicating
terminology.

Recall that if $G$ and $H$ are groups then a group extension of $H$ by $G$
is a sequence of groups and homomorphisms
\begin{equation}
    \begin{tikzcd}
        G \arrow[r, "i"] & H \arrow[r, "j"] & K
    \end{tikzcd} 
\end{equation}
such that $i$ is injective, $j$ is surjective, and $\text{im}(i) = \text{ker}(j)$.
The simplest example of a group extension of $H$ by $G$ is the product $G \times H$
along with the inclusion $\text{in}_1:G \to G \times H$ and the projection
$\pi_2:G \times H \to H$.
\begin{equation}
    \begin{tikzcd}
        G \arrow[r, "\text{in}_1"] & G \times H \arrow[r, "\pi_2"] & H
    \end{tikzcd} 
\end{equation}
As a group extension, the direct product is not interesting because its structure
is determined completely by $G$ and $H$. It is therefore called the \emph{trivial}
extension.
A homomorphism $h:G \to G \times K$ that is compatible
with both the inclusion and projection maps, that is the diagram
\begin{equation}
    \begin{tikzcd}
        G \arrow[dr, "\text{in}_1"] \arrow[r, "i"] & H \arrow[d, "h"] \arrow[r, "j"] & K \\
         & G \times K \arrow[ur, "\pi_2"] &
    \end{tikzcd}
\end{equation}
commutes is called a trivialisation.
It can be shown that any trivialisation is an isomorphism. 
A bundle that has a splitting
is said to split, and its structure is therefore also determined completely
by $G$ and $K$.
The \emph{splitting lemma} for groups gives a necessary and sufficient characterisation
of when a group extension has a splitting.

Partial monoids generalise groups by omitting the requirement that elements
have inverses, and the requirement that all products are defined.
\begin{definition}
    A \emph{(commutative) partial monoid} is a tuple $(M, +, 0)$
    where $M$ is a set, the \emph{product} $+:M^2 \to M$ is a partial function, and $0 \in M$ is the \emph{identity},
    such that the following conditions hold:
    \begin{itemize}
        \item Commutativity: $m + m'$ is defined if and only if $m' + m$ is defined and $m + m' = m' + m$, for all
        $m,m' \in M$.
        \item Identity: $0 + m$ is defined and $0 + m = m$ for all $m \in M$.
        \item Associativity: For all $m,m',m'' \in M$
            \begin{itemize}
                \item  If $(m + m') + m''$ and $m + (m' + m'')$ are both
                    defined then they are equal.
                \item  If $m + m', m + m'', m' + m''$ are all defined
                    then $(m + m') + m''$ and $m + (m' + m'')$ are both defined.
            \end{itemize}
    \end{itemize}
    If the product $+$ is a total function then $(M, +, 0)$ is a \emph{commutative monoid}.
\end{definition}

In this section, we introduce a generalisation of group extensions to partial
monoids and we show that the splitting lemma generalises, and the problem
of extending a local splitting to a global splitting is equivalent.

\subsection{Bundles}
To generalise the definition of a group extension to partial monoids
we first recast the definition to emphasise the role of a group action.

\begin{definition}
    Let $\theta:G \times X \to X$ be a group action.
    $\theta$ is \emph{free} if $\theta(\_, x): G \to X$ is injective for all $x \in X$.
    The \emph{orbit} of $x \in X$, is the set of elements
    that are equivalent to $x$ up to the action of $G$:
    \begin{equation}
        [x]_{\theta} := \{\theta(g,x) \mid g \in G\}
    \end{equation}
    We write $X / \theta$ for the set of orbits. When a particular group action
    is assumed we will simplify notation by defining $g \cdot x := \theta(g,x)$.
\end{definition}

Observe that for any group extension
\begin{tikzcd}
        G \arrow[r, "i"] & H \arrow[r, "j"] & K
    \end{tikzcd} 
there is an action of $G$ on $H$ defined by
\begin{align}
    \theta : G \times H \to H :: (g,h) \mapsto i(g) +_H h
\end{align}

This action is \emph{free} because $i$ is injective and $H$ has inverses.
It is also \emph{compatible} with the group structures of $G$ and $H$
in the sense that it is a homomorphism from $G \times H$ to $H$.
The requirement that $\text{im}(i) = \text{ker}(j)$
is equivalent to saying that the orbits of $\theta$
and the fibers of $j$ are the same:
\begin{equation}
    H / \theta = \{j^{-1}(k) \mid k \in K\}
\end{equation}

We can recast the definition of a group extension in terms of this action.
A group extension can be defined as a surjective homomorphism
$j:H \to K$ and a free, compatible group action $\theta$
such that the orbits of $\theta$ are the fibers of $j$.
We will use this view of group extensions to generalise them to partial monoids.

For a partial monoid the natural notion of homomorphism
is a function on the underlying set that preserves the identity and products whenever
they are defined.
\begin{definition}
    A \emph{homomorphism} of partial monoids $h:M \to M'$ is
    a function between the underlying sets, such that:
    \begin{itemize}
        \item $h$ preserves the identity element: $h(0_M) = 0_{M'}$.
        \item $h$ preserves products: $h(m_1) +_{M'} (m_2)$ is defined and 
            $h(m_1 +_M m_2) = h(m_1) +_{M'} h(m_2)$, for all $m_1,m_2$
            such that $m_1 +_M m_2$ is defined.
    \end{itemize}
\end{definition}

If $G$ is a group and $M$ is a partial monoid then the set product $G \times M$
is a partial monoid with identity and product defined component-wise.
\begin{align}
    0_{G \times M} &= (0_G, 0_M)\\
    (g,m) +_{G \times M} (g', m') &= (g +_A g', m +_M m')
\end{align}
for all $g,g' \in G$ and $m,m' \in M$ such that $m +_M m'$ is defined.
We define an action of a group $G$ on $M$ to be a group action,
in the usual sense, that is furthermore a homomorphism from $G \times M$ to $M$.
\begin{definition} \label{def:action}
    Let $G$ be a group and $M$ a partial monoid.
    An \emph{action of $G$ on $M$} is a \emph{homomorphism}
    $\theta:G \times M \to M$ such that the following
    conditions hold:
    \begin{align}
        \theta(0, \_) &= \text{id}_M\\
        \theta(g, \_) \circ \theta(g', \_) &= \theta(g + g', \_),
        \quad \text{for all $g,g' \in G$}
    \end{align}
\end{definition}

We define a bundle over a partial monoid to be
a partial monoid equipped with a compatible group action
and a surjective homomorphism such that the fibers of the homomorphism
and the orbits of the action are the same.
\begin{definition}
Let $G$ be a group and $M$ a partial monoid.
A \emph{$G$-bundle} over $M$ is a tuple $(N, j, \theta)$, where
\begin{itemize}
    \item $N$ is a partial monoid,
    \item $\theta:G \times N \to N$ is a free action,
    \item $j:N \to M$ is a surjective homomorphism,
    \end{itemize}
        such that the orbits of $\theta$ are the fibers of $j$:
    \begin{equation}
        N / \theta = \{j^{-1}(m) \mid m \in M\}
    \end{equation}
    $\theta$ is called the \emph{bundle action}
    and $j$ the \emph{bundle map}.
\end{definition}

The simplest example of a $G$-bundle over $M$ is given by the product $G \times M$.
Write $\theta_{G \times M}$ for the action of $G$ on $G \times M$ 
applying the group operation of $G$ on the first component,
and $\pi_2:G \times M \to M$ for the projection onto the second component.
\begin{align}
    \theta(g, (g', m)) := (g + g', m)
\end{align}
The triple $(G \times M, \theta_{G \times M}, \pi_2)$ is called the 
\emph{trivial bundle}. As a bundle it has no interesting structure
because it is completely determined by $G$ and $M$ alone.

\subsection{The splitting lemma}
Let
\begin{tikzcd}
        G \arrow[r, "i"] & H \arrow[r, "j"] & K
\end{tikzcd} 
be a group extension. The splitting lemma for groups
gives the following characterisation of trivialisations,
that is homomorphisms $h:H \to G \times K$
such that the following diagram commutes:
\begin{equation}
    \begin{tikzcd}
        G \arrow[dr, "\text{in}_1"] \arrow[r, "i"] & H \arrow[d, "h"] \arrow[r, "j"] & K \\
         & G \times K \arrow[ur, "\pi_2"] &
    \end{tikzcd}
\end{equation}
in other words, $\text{in}_1 = h \circ i$ and $j = \pi_2 \circ h$.
Trivialisations are necessarily isomorphisms. Any group extension
that has a trivialisation is therefore isomorphic to the product group extension.

A \emph{left splitting} is a homomorphism $s:H \to G$
such that $i \circ l = \text{id}_H$.
A \emph{right splitting} a homomorphism $r:K \to H$
such that $r \circ j = \text{id}_K$.
The \emph{splitting lemma} for groups states that the three are equivalent: A group
extension has a left splitting if and only if it has a right splitting,
if and only if it has a trivialisation.

To generalise left splittings and trivialisations
we observe that their definitions can be recast in terms of the group
action of $G$ on $H$.
\begin{definition} \label{def:action-homomorphism}
    Let $G$ be a group, $M,M'$ partial monoids, and $\theta:G \times M \to M$, $\theta':G \times M' \to M'$
    group actions. An \emph{action homomorphism} $h:\theta \to \theta'$
    is a partial monoid homomorphism $h:M \to M'$
    such that $\theta'(g, f(x)) = f(\theta(g,x))$ for all $g \in G, x \in X$.
\end{definition}

For any group $G$ write 
write $\theta_G$ for the group action of $G$ on itself: $\theta_G(g,g') := g +_G g'$.
A left splitting of a group extension is then equivalently
an action homomorphism from the bundle action $\theta$ to $\theta_G$.
The requirement that a trivialisation is compatible
with the inclusion maps, that is $\text{in}_1 = h \circ i$
is equivalent to $h$ being an action homomorphism
from the bundle action $\theta$ to the bundle action
on the product bundle $\theta_{G \times M}$.
\begin{definition}
    Let $G$ be a group, $M$ a partial monoid, and $B = (N, j, \theta)$ a $G$-bundle over $M$.
    \begin{enumerate}
        \item  A \emph{left splitting} is an action homomorphism $l:\theta \to \theta_G$.
        \item  A \emph{right splitting} is a partial monoid homomorphism $r:M \to N$ such that $j \circ r = \text{id}_M$.
        \item A \emph{trivialisation} is an action homomorphism $h:\theta \to \theta_{G \times M}$
        such that $j = \pi_2 \circ h$.
    \end{enumerate}
\end{definition}

For example, the trivial bundle $(G \times M, \theta_{G \times M}, \pi_2)$
has a left splitting $\pi_1$ and a right splitting
$\text{in}_2$:
\begin{align}
    \pi_1:G \times M \to G &:: (g,m) \mapsto b\\
    \text{in}_2 : M \to G \times M &::= m \mapsto (0_B, m)
\end{align}

Let $B = (N, j ,\theta)$ be a $G$-bundle over a partial monoid $M$
and let $l:N \to G$ be a left splitting.
There is then a natural map from $N$ to $G \times M$
given by
\begin{align}
    \langle l, j \rangle : N \to G \times M ::= n \mapsto (l(n), j(n))
\end{align}
Because $l$ is a left splitting and therefore an action
homomorphism from $\theta$ to $\theta_G$ we have that
$\langle l, j \rangle$ is an action homomorphism from $\theta$
to the bundle action $\theta_{G \times M}$ of the trivial bundle.
$\langle l, j \rangle$ is a trivialisation.

We then clearly have $\pi_2 \circ \langle l, j \rangle = j$.
Because $l$ is an action homomorphism from $\theta$ to $\theta_G$
we have that $\langle l, j \rangle$ is an action homomorphism
from $\theta$ to $\theta_{G \times M}$. Conversely if $h:N \to G \times M$
is a trivialisation then we can define a left splitting
by projecting onto the first component: $\pi_1 \circ h$.

Because the bundle action $\theta$ is free something similar is true
for right splittings.
For any left splitting $l$ let $\mathcal{R}(l)$
be the function
\begin{align}
    \mathcal{R}(l) : M \to N ::= m \mapsto -l(\eta(m)) \cdot \eta(m)
\end{align}
where $\eta:M \to N$ is any function such that $j \circ \eta = \text{id}_M$.
Observe that the definition is independent of the choice of $\eta$
because
\begin{align}
    -l(g \cdot \eta(m)) \cdot (g \cdot \eta(m)) &= (-l(\eta(m)) -g + g) \cdot \eta(m) \\
    &= -l(\eta(m)) \cdot \eta(m)
\end{align}
for any $m \in M$ and $g \in G$.

\begin{lemma}[Splitting lemma]
    Let $B = (N, \theta, j)$ be a $G$-bundle over a partial monoid $M$.
    \begin{enumerate}
        \item The map $l \mapsto <l,j>$ is a bijection between left splittings and trivialisations.
        \item The map $l \mapsto \mathcal{R}(l)$ is a bijection between left and right splittings.
    \end{enumerate}
    \end{lemma}
\begin{proof}
    1.\ $h \mapsto \pi_1 \circ h$ is an inverse to $l \mapsto <l,j>$.
    $\pi_1 \circ h$ is an action homomorphism from $\theta$ to $\theta_G$
    if and only if $h$ is an action homomorphism from $\theta$ to $\theta_{G \times M}$.
    
    For 2.\ we first check that $\mathcal{R}(l)$ is a homomorphism.
    It preserves the identity.
    We have $\eta(0) = a \cdot 0$ for some unique $a$.
    Hence $\mathcal{R}(l)(0) = (-l(\eta(0))) \cdot \eta(0) =
        -l(a \cdot 0) \cdot (a \cdot 0) = (0 - a + a) \cdot 0$ as required.
    To see that it preserves products, take $m, m'$ such that $m +_M m'$ is defined.
    There is a unique $g$ such that $\eta(m + m') = g \cdot (\eta(m) + \eta(m'))$.
    Therefore
    \begin{align}
        \mathcal{R}(l)(m + m') &= (-l(\eta(m+m'))) \cdot \eta(m + m')\\
                               &= (-l(g \cdot (\eta(m) + \eta(m'))) \cdot (g \cdot (\eta(m) + \eta(m')))\\
                               &= (-l(\eta(m)) -l(\eta(m')) - g + g) \cdot (\eta(m) + \eta(m'))\\
                               &= \mathcal{R})(l)(m) + \mathcal{R}(l)(m')
    \end{align}
    
    To see that $\mathcal{R}$ is a bijection we can define the inverse
    directly as follows. For any right splitting $r:M \to N$
    there is a unique function $\mathcal{R}^{-1}(r):M \to N$ such that
    \begin{align}
        \mathcal{R}^{-1}(r)(n) \cdot n = r(j(n))
    \end{align}
    for all $n \in N$.
    
    We first check that it is a homomorphism.
    For the identity we have $h(j(0)) = 0$ and $h(j(0)) = s(0) \cdot 0$,
    hence $s(0) = 0$ as required.
    That it preserves products we have both
    \begin{align}
        h(j(n + n')) = s(n + n') \cdot (n + n')
    \end{align}
    and
    \begin{align}
        h(j(n + n')) &= h(j(n)) + h(j(n')) \\
                    &= s(n) \cdot n + s(n') \cdot n' \\
                    &= (s(n) + s(n')) \cdot (n + n')
    \end{align}
    Therefore, by uniqueness of $s$ we have $s(n + n') = s(n) + s(n')$.
    Finally to see that it preserves the action, we have
    $r(j(g \cdot n)) = r(j(n))$. Hence
    \begin{align}
        \mathcal{R}^{-1}(r)(g \cdot n) \cdot (g \cdot n) = 
        \mathcal{R}^{-1}(r)(n) \cdot n
    \end{align}
    
    $\mathcal{R}^{-1}(r)(g \cdot n) = g + \mathcal{R}^{-1}(r)(n)$
    
    Finally we check that $\mathcal{R}^{-1}$ in fact is an inverse
    to $\mathcal{R}$. 
    \begin{align}
        R_B(R_B^{-1}(r)) &= m \mapsto (- s_r(\eta(m))) \cdot \eta(n)\\
                        &= m \mapsto r(j(\eta(m))) = m \mapsto r(m)
    \end{align}
    and 
    \begin{align}
        (h^{-1} \circ \text{in}_M)(j(n)) &= h^{-1}(0, j(n))\\
        &= (0 - h_1(n)) \cdot n
    \end{align}
    Hence by uniqueness $(h^{-1} \circ \text{in}_M) \mapsto h$
    and so $\mathcal{R}$ is a left inverse to $\mathcal{R}^{-1}$.
\end{proof}

Similarly, as for groups, trivialisations of bundles are necessarily isomorphisms.
The splitting lemma, therefore, gives a characterisation of when
a bundle is isomorphic to the trivial bundle.
A natural candidate for the inverse of a trivialisation
$h$ is the map given by first
taking the left splitting $\pi_1 \circ h:N \to G$, then composing the right splitting
associated with $\pi_1 \circ h$ with the projection $\pi_2$:
$\mathcal{R}(\pi_1 \circ h) \circ \pi_2$.
It can be verified that this map in fact is an inverse.

\begin{lemma}
    Trivialisations are isomorphisms. Let $B = (N, j, \theta)$ be an $A$-bundle
    over a commutative partial monoid $M$ and $h:N \to A \times M$ a trivialisation.
    The inverse of $h$ is
    \begin{equation}
        h^{-1}:A \times M \to N ::= (a,m) \mapsto (a - h_1(\eta(m))) \cdot \eta(m)
    \end{equation}
    where $\eta : \prod_{m \in M} j^{-1}(m)$ is an arbitrary section
    and $h_1 = \text{proj}_1 \circ h: N \to A$ is the first component of $h$.
\end{lemma}
\begin{proof}
    We first show that $h^{-1}$ is independent of the choice of the section $\eta$.
    Any other section $\eta' = m \mapsto \gamma(m) \cdot \eta(m)$ differs from 
    $\eta$ by  some $\gamma:M \to A$.
    If we expand the definition of $h^{-1}$ using the section $\eta'$
    in terms of $\gamma$ and $\eta$ we see that the terms involving $\gamma$ cancels
    out. For any $m \in M$ we have
    \begin{align}
        (a - h_1(\eta'(m))) \cdot \eta'(m) &=
            (a - h_1(\gamma(m) \cdot \eta(m))) \cdot (\gamma(m) \cdot \eta(m))\\
            &= (a - \gamma(m) - h_1(\eta(m)) + \gamma(m)) \cdot \eta(m)\\
            &= (a - h_1(\eta(m))) \cdot \eta(m)
    \end{align}
    Therefore $h^{-1}$ is independent of the choice of $\eta$.
    
    Next, we check that $h^{-1}$ is a left inverse to $h$. Let $n \in N$.
    To see that $h^{-1}(h(n)) = n$ we first expand 
    $h(n) = (h_1(n), h_2(n))$
    into the two components of the product.
    Using the section $\eta$ we can write $n$ uniquely on the form
    \begin{equation}
        n = a \cdot \eta(m)
    \end{equation}
    where $a \in A$ and $m = j(n)$. Because $h$ is a trivialisation
    $j(n) = h_2(n)$.
    \begin{align}
        h^{-1}(h(n)) &= h^{-1}(h_1(n), h_2(n))\\
            &= h^{-1}(h_1(n), m)
    \end{align}
    Because both $h$ and $\text{proj}_1$ preserve the action of
    $A$ we have $h_1(a \cdot \eta(m)) = a + h_1(\eta(m))$
    therefore if we plug $(h_1(n), m)$ into $h^{-1}$
    the terms involving $h_1(\eta(m))$ cancels out
    \begin{align}
        h^{-1}(h_1(n), m) &= (h_1(n) - h_1(\eta(m))) \cdot \eta(m)\\
            &= (h_1(a \cdot \eta(m)) - h_1(\eta(m))) \cdot \eta(m)\\
            &= (a + h_1(\eta(m)) - h_1(\eta(m))) \cdot \eta(m) \\
            &= a \cdot \eta(m) = n
    \end{align}
    as required. 
    
    Finally, we check that $h^{-1}$ is a right inverse to $h$. Let $(a,m) \in A \times M$.
    We first expand the definition of $h^{-1}(a,m)$
    \begin{align}
        h(h^{-1}(a,m)) &= h((a - h_1(\eta(m))) \cdot \eta(m))
    \end{align}
    and then separately check that $h_1(h^{-1}(a, m)) = a$ and 
    $h_2(h^{-1}(a, m))_2 = m$. For the
    first part we use the fact that $h_1$ is an $A$-action homomorphism
    \begin{align}
        h_1((a - h_1(\eta(m))) \cdot \eta(m)) = (a - h_1(\eta(m))) + h_1(\eta(m)) = a
    \end{align}
    The second part follows because $h$  maps the fiber $j^{-1}(m)$
    to $A \times \{m\}$
    \begin{align}
        h_2((a - h_1(\eta(m))) \cdot \eta(m)) = m
    \end{align}
    as required.
\end{proof}

\subsection{Extending local splittings} \label{section:extending-splittings}

The splitting lemma gives a correspondence between left splittings,
right splittings, and trivialisations. We now show that this
correspondence is compatible with restrictions. This means that the problem
of extending a right splitting, left splitting, or trivialisations
defined on a sub-bundle are all equivalent.

Suppose that $B = (N, j, \theta)$ is a $G$-bundle over a partial monoid $M$,
and that $M' \subset M$ is a sub partial monoid.
We first explain that $B$ restricts to a sub bundle over $M'$.

The pre-image $j^{-1}(M') \subset N$ is a sub-partial
monoid of $N$, and it is closed under the action $\theta$.
We can therefore restrict $B$ to a $G$-bundle over $M'$
by restricting both the bundle map $j$ and action $\theta$.
\begin{definition}
    Let $B = (N, j, \theta)$ be a $G$-bundle over a partial monoid
    $M$ and $M' \subset M$ a sub partial monoid. The \emph{restriction
    of $B$ to $M'$}, denoted by $\res{B}{M'}$, is the $G$-bundle over $M'$
    \begin{equation}
        \res{B}{M'}  :=
            \big(N',
            j',
            \theta'
            \big)
    \end{equation}
    where  $N' := j^{-1}(M')$ and
    $j':N' \to M', \theta':G \times N' \to N'$
    are the restrictions of $j$ and $\theta$ to $N'$.
\end{definition}

Because the maps in the splitting lemma are defined
pointwise they are natural with respect to restrictions.
The problems of extending a left splitting, 
right splitting, or trivialisation of the restricted
bundle $\res{B}{M'}$ to $B$ are therefore equivalent.

\begin{lemma}
    Let $B = (N, j, \theta)$ be a $G$-bundle over a partial
    monoid $M$, $M' \subset M$ a sub partial monoid,
    and $l':M' \to G$ a left splitting
    of the restricted bundle $\res{B}{M'}$.
    The following conditions are equivalent:
    \begin{itemize}
        \item There exists a left splitting $l:N \to G$
            such that $\res{l}{M'} = l'$.
        \item There exists a trivialisation $h:N \to G \times M$
            such that $\res{h}{N'} = <l', j>$.
        \item There exists a right splitting $r:M \to N$ such
        that $\res{r}{M'} = \mathcal{R}(l')$.
    \end{itemize}
\end{lemma}

\section{Cohomology of commutative partial monoids} \label{section:relative-cohomology}

We concluded the previous section by explaining that for
a $G$-bundle $B$ over a partial monoid $M$
the problems of extending either a left splitting, right splitting,
or trivialisation, defined on a sub bundle are equivalent.
In this section we show that this problem can be given a cohomological
characterisation. The construction can be seen as a generalisation
of \emph{group cohomology}.

Let $G$ and $K$ be groups. A well-known problem in group
theory is to classify the possible group extensions
of $K$ by $G$. An elegant solution to this problem
is given by \emph{group cohomology} \cite{brown_cohomology_2012}.
Two group extensions are \emph{equivalent}
if they are related by an isomorphism:
\begin{equation}
    \begin{tikzcd}
        G \arrow[dr, "i'"] \arrow[r, "i"] & H \arrow[d, "h"] \arrow[r, "j"] & K \\
         & H' \arrow[ur, "j'"] &
    \end{tikzcd}
\end{equation}
So in particular an extension splits if it is equivalent to the trivial extension.

For any group $K$ there is a topological
space $X_K$ called the \emph{classifying space} of $K$.
There is a bijection between the second cohomology group $H^2(X_K, G)$
of $X_K$ with coefficients in $G$ ,
and equivalence classes of group extensions.
In particular, the equivalence class of the trivial extension
correspond to the zero class $0 \in H^2(X_K;G)$.

In this section, we first generalise group cohomology. 
In Section \ref{section:grp-coho}
we define the \emph{relative cohomology groups} $H^n(M,M';G)$
of a partial monoid $M$ with respect to a sub partial monoid $M \subset M'$
with coefficients in a group $G$.
In \ref{section:obstruction-splitting} we define for any
local right splitting $r$ of a sub-bundle
a cohomological obstruction $\mu(r) \in H^2(M,M';G)$
and we show that $\mu(r) = 0$ if and only
if $r$ can be extended to a global splitting.

\subsection{The cohomology groups of a partial monoid} \label{section:grp-coho}
Let $G$ be a group, $M$ a partial monoid, and $M' \subset M$ a sub partial monoid.
In this section, we define the relative cohomology groups $H^n(M,M';G)$.

We begin by defining a family of sets $\{M_n\}_{n \in \nats}$
and \emph{boundary maps} $\delta_{n,i}:M_n \to M_{n-1}$, where $i=1,\dots,n$ encoding
the structure of a partial monoid $M$.
\begin{definition}
    Let $K$ be a commutative partial monoid.
    $\{K_n\}_{n \geq 0}$
    and $\delta_{n,i}:K_n \to K_{n-1}$,
    where $n \geq 1, i=0,\dots,n$
    are defined by $K_0 := \{()\}$ and when $n \geq 1$
    \begin{align}
    	K_n &:= \{(k_1, k_2, \cdots, k_n) \in K^{n} \mid k_1 + k_2 + \cdots + k_n \text{ is defined} \}\\
        \delta_{n,i} &::= (k_1, \dots, k_n) \mapsto
            (k_1, \dots, k_{i-1}, k_i + k_{i+1}, k_{i+2}, \dots, k_n)
    \end{align}
\end{definition}

By considering the set of functions
$f:M_n \to G$ that vanish on $M'_n \subset M_n$
we define the relative co-chain complex
\begin{equation}
    \vcenter{
        \hbox{
            \begin{tikzcd}
                0 \arrow[r, "d^{-1} := 0" above] &
                C^0(M,M';G) \arrow[r, "d^0" above] &
                C^1(M,M';G) \arrow[r, "d^1" above] &
                C^2(M,M';G) \arrow[r, "d^2" above] &
                \cdots
            \end{tikzcd}
        }
    }
\end{equation}
\begin{definition}
    Let $G$ be a group, $M$ a partial monoid, and $M' \subset M$ a sub partial monoid.
    \begin{enumerate}
        \item The \emph{relative $n$-cochains}, denoted by $C^n(M,N;G)$,
            is the commutative group of assignments $f:M_n \to G$ that vanish 
            on $M'_n$.
            \begin{align}
                    C^n(M',M;G) &:= \{f:M_n \to G \mid \res{f}{{M'}_n} = 0\}
            \end{align}
        \item The $n$'th coboundary map, denoted by $d^n$
            is the following homomorphism from the relative $n$-cochains
            to relative $(n-1)$-cochains
            \begin{align}
                d^n &:C^n(M', M;G) \to C^{n+1}(M', M;G)\\
                d^n(f) &:= 
                    (m_1, \dots, m_n) \mapsto \sum_{i=0}^{n} (-1)^i f(\delta_{n,i}(m_1, \dots, m_n))
            \end{align}
    \end{enumerate}
\end{definition}

The relative cohomology groups $H^n(M,M';G)$ are the cohomology groups
of this co-chain complex. To verify that this in fact
defines a co-chain complex we need to verify
that $d^{n+1} \circ d^n = 0$. This can be done by a straightforward computation.
In our case, it is only necessary to verify this for the maps
\begin{align}
    \label{defd2}
	d^2(f)(m_1,m_2,m_3) &= 
	    f(m_2,m_3) - f(m_1 + m_2,m_3) + f(m_1, m_2 + m_3) - f(m_1,m_2)\\
	d^1(f)(m_1,m_2) &= f(m_2) - f(m_1+m_2) + f(m_1)\\
	d^0 &= 0
\end{align}
which is easily done.
\begin{definition}
    Let $G$ be a commutative group $M$ a commutative partial monoid
    and $M' \subset M$ a sub partial monoid.
    \begin{enumerate}
          \item  The \emph{relative $n$-cocycles} 
                    $Z^n(M,N;G) := \ker{d^n}$ is the kernel
                    of $d^n:C^n(K,M;G) \to C^{n-1}(K,M;G)$.
            \item The \emph{relative $n$-coboundaries}
                is the image $B^n(M,N;G) := \im{d^{n-1}}$
                of $d^n:C^n(K,M;G) \to C^{n-1}(K,M;G)$.
        \item The \emph{relative cohomology group}
            $H^n(K,M;G) := Z^n(K,M;G) / B^n(K,M;G)$ is the quotient of the
                relative $n$-cocycles over the relative $n$-coboundaries.
    \end{enumerate}
\end{definition}
    
\subsection{The obstruction to extending a local splitting}
\label{section:obstruction-splitting}

Let $B = (N, \theta, j)$ be a $G$-bundle over a partial monoid $M$,
$M' \subset M$ a sub partial monoid, and $r':M' \to N'$
a right splitting of the restriction $\res{B}{M'}$ of $B$ to $M'$.

To define the cohomological obstruction $\mu(r') \in H^2(M,M';G)$
we first choose a function $\eta:M \to N$, not necessarily a homomorphism,
such that $j \circ \eta = \text{id}_M$ and $\res{\eta}{M'} = r'$.
$\eta$ is not necessarily a homomorphism but because
$j \circ \eta = \text{id}_M$ there is for every $m_1, m_2 \in M$
some $g \in G$ such that 
$\eta(m_1 +_M m_2) = g \cdot (\eta(m_1) +_M \eta(m_2))$. 
Because $\theta$ is free this $g$ is unique.

\begin{definition}
    Let $B = (N, j, \theta)$ be an $G$-bundle over a 
    partial monoid $M$, $M' \subset M$ a sub partial monoid,
    and $\eta:M \to N$ a function
    such that $\eta \circ j = \text{id}_M$.
    Write  $\sExtFailNoArgs :M_2 \to G$ for the unique
    function satisfying
    \begin{align}
        \eta(m_1 + m_2) = \sExtFail{m_1}{m_2} \cdot (\eta(m_1) + \eta(m_2))
    \end{align}
    for all $(m_1, m_2) \in M_2$.
\end{definition}

$\sExtFailNoArgs$ can be thought of as measuring the failure
of $\sExtNoArgs$ to be a splitting
because $\sExtNoArgs$ is a homomorphism 
if and only if $\sExtFailNoArgs = 0$.
We define $\mu(r')$ to be the cohomology class of $\Delta \eta$.
\begin{definition}
    Let $B = (N, j, \theta)$ be an $G$-bundle over $M$,
    $M' \subset M$ a sub-partial monoid,
    and $r':M' \to N'$ a right splitting of $\res{B}{M'}$.
    The \emph{obstruction to $r'$},
    is the cohomology class
    \begin{align}
        \mu(r') := [\Delta \eta] \in H^2(M',M';A)
    \end{align}
    where $\eta:M \to N$ is any function such that $\res{\eta}{M'} = r'$
    and $j \circ \eta = \text{id}_M$.
\end{definition}

For this to be well defined we need to check that $\Delta \eta$
is a relative co-cycle and that the cohomology class
$[\sExtFailNoArgs]$ is independent of the choice of $\eta$.
\begin{lemma}
    Let $B = (N, j, \theta)$ be an $A$-bundle over a commutative partial monoid
    $M$, $M' \subset M$ a sub partial monoid, and
    $r$ a right splitting of the restricted bundle $\res{B}{M'}$.
    \begin{enumerate}
        \item $\Delta \eta \in Z^2(M',M;A)$ for any $\eta:M \to N$ such
        that $\res{\eta}{M'} = r'$ and $j \circ \eta = \text{id}_M$.
        \item $(\Delta \eta - \Delta \eta') \in B^2(M',M;A)$ for any
            two $\eta, \eta':M' \to N'$ such that $\res{\eta}{M'} = \res{\eta'}{M'} = r'$ and $j \circ \eta = j \circ \eta' = \text{id}_M$.
    \end{enumerate}
\end{lemma}
\begin{proof}
    For 1.\ we first have to show that $\sExtFailNoArgs$ is a \emph{relative} cochain, that is, that $\sExtFailNoArgs$ vanishes
    on $(M_2')$, and secondly that
    \begin{align}
            \sExtFail{m_2}{m_3}
            - \sExtFail{m_1 + m_2}{m_3}
            + \sExtFail{m_1}{m_2 + m_3}
    	    - \sExtFail{m_1}{m_2}
    	    = 0
	\end{align}
	for all $(m_1, m_2, m_3) \in M_3$. That $\sExtFailNoArgs$
	vanishes on $M_2'$ is clear because its restriction
	is a homomorphism. For the second part we use that
	$m_1 + m_2 + m_3$ can be written as
	both $m_1 + (m_2 + m_3)$ and $(m_1 + m_2) + m_3$.
	\begin{align*}
	    \sExt{m_1 + (m_2 + m_3)}
	        &= 
	            \sExtFail{m_1}{m_2 + m_3} 
	                \cdot (\sExt{m_1} + \sExt{m_2 + m_3})\\
            &= (\sExtFail{m_1}{m_2 + m_3} + \sExtFail{m_2}{m_3})
	            \cdot (\sExt{m_1} + \sExt{m_2} + \sExt{m_3})
	\end{align*}
	and similarly
	\begin{align*}
	    \sExt{(m_1 + m_2) + m_2} &=
	        (\sExtFail{m_1 + m_2}{m_3} + \sExtFail{m_1}{m_2})
	            \cdot (\sExt{m_1} + \sExt{m_2} + \sExt{m_3})
	\end{align*}
	Because the two terms are equal and the action is free
	\begin{align}
	\sExtFail{m_1}{m_2 + m_3} + \sExtFail{m_2}{m_3}
	 =
	 \sExtFail{m_1 + m_2}{m_3} + \sExtFail{m_1}{m_2}
	\end{align}
	as required.
    
    For 2.\ suppose that $\sExtNoArgs, \sExtNoArgs'$ are two sections
    that extend $r$. We have to show that there is
    some $\gamma : C^1(M', M ; A)$ such that
    \begin{align}
        \Delta \eta\, (m_1, m_2) - \Delta \eta'\, (m_1, m_2) =
            \gamma(m_1) - \gamma(m_1 + m_2) + \gamma(m_2)
    \end{align}
    for all $(m_1, m_2) \in M_2$. Let $\gamma : M \to A$
    be the unique function such that 
    \begin{equation}
        \eta = m \mapsto s(m) \cdot \eta'(m)
    \end{equation}
    Because $\eta,\eta'$ both extend $r$ we have $\res{s}{M'} = 0$
    and so $\gamma$ is a relative cochain, 
    $\gamma \in C^1(M',M;A)$.
    Expanding $\eta(m_1 + m_2), \eta(m_1), \eta(m_2)$ in terms 
    of $\gamma$ and $\eta'$
    gives
    \begin{align}
	    \eta(m_1 + m_2) &= 
	            \Delta \eta\, (m_1,m_2) \cdot (\eta(m_1) + \eta(m_2))\\
	                &= (\Delta \eta\, (m_1, m_2) + s(m_1) + s(m_2)) \cdot (\eta'(m_1) + \eta'(m_2))
    \end{align}
    and
    \begin{align}
        \eta(m_1 + m_2) &= s(m_1 + m_2) \eta'(m_1 + m_2)\\
            &= (s(m_1 + m_2) + \Delta \eta'\, (m_1, m_2)) \cdot (\eta'(m_1) + \eta'(m_2))
    \end{align}
    hence
    \begin{align}
        \Delta \eta\, (m_1, m_2) + s(m_1) + s(m_2)
        =
        s(m_1 + m_2) + \Delta \eta'\, (m_1, m_2)
    \end{align}
    as required.
\end{proof}

Observe that the obstruction is \emph{sound}
in the sense that if $r'$ can be extended to a right
splitting $r:M \to N$ then $[\Delta \eta] = 0$.
This is true because if such an $r$ exists
then the cohomology class $[\Delta \eta]$ is equal
to the cohomology class $[\Delta r]$ had we instead
chosen $r$. Because $r$ is a homomormphism $\Delta r = 0$
and so $[\Delta r] = 0$ as required.

We now show that the obstruction in fact is complete
in the sense that $\mu(r') = 0$ if and only if $r'$ can be extended globally.
\begin{theorem}
    Let $B = (N, j, \theta)$ be a $G$-bundle over
    a partial monoid $M$, $M' \subset M$ a sub partial monoid,
    and $r$ a right splitting of $\res{B}{M'}$.
    There exists a right splitting $r:M \to N$
    such that $\res{r}{M'} = r'$
    if and only if $\mu(r') = 0$.
\end{theorem}
\begin{proof}
    We have already explained that $\mu(r') = 0$ if $r'$ can
    be extended globally. For the converse
    suppose that $\mu(r') = 0$.
    
    We extend $r'$ to a global right splitting $r:M \to N$
    by first choosing a function $\eta:M \to N$ such that
    $\eta \circ j = \text{id}_M$
    and $\res{\eta}{M'} = r'$.
    Because $[\Delta \eta] = 0$ there is a unique $\gamma \in C^1(M',M;A)$
    such that $\Delta \eta = d^1(\gamma)$. We define $r$ by
    \begin{align}
        r(m) ::= -\gamma(m) \cdot \eta(m)
    \end{align}
    for all $m \in M$.
    To see that $r$ is a homomorphism take $m_1, m_2 \in M$
    such that $m_1 +_M m_2$ is defined.
    We have
    \begin{align}
        r(m_1 + m_2) &= -\gamma(m_1 + m_2) \cdot \eta(m_1 + m_2)\\
        \eta(m_1 + m_2) &= \Delta \eta\, (m_1, m_2) \cdot (\eta(m_1) + \eta(m_2))\\
        \Delta \eta\, (m_1, m_2) &= d^1(\gamma)(m_1,m_2) = \gamma(m_1) - \gamma(m_1 + m_2) + \gamma(m_2)
    \end{align}
    Hence
    \begin{align}
        r(m_1 + m_2) &= -\gamma(m_1 + m_2) \cdot \eta(m_1 + m_2)\\
                    &= (-\gamma(m_1 + m_2) + \Delta \eta\, (m_1,m_2)) \cdot (\eta(m_1) + \eta(m_2))\\
                    &= (-\gamma(m_1 + m_2) + \gamma(m_1) - \gamma(m_1 + m_2) + \gamma(m_2)) \cdot (\eta(m_1) + \eta(m_2))\\
                    &= r(m_1) + r(m_2)
    \end{align}
    as required.
\end{proof}

\section{Measurement scenarios with bundle structure}
\label{section:bundle-scenarios}

We now introduce a class of measurement scenarios and empirical models
generalising the state-dependent and state independent
empirical models $\pmodel_{\mathcal{O}}, \pmodel_{\mathcal{O}, \psi} : (\mathcal{O}, \mathcal{M}, \mathbb{Z}_d)$ associated with a closed set of Weyl operators
$\mathcal{O} \subset P_{n,d}$ and a state $\psi$.
We first give an abstract description of the structure of these
models. In Section \ref{section:detecting-homomorphisms} we explain
that for models of this type a test for (non) contextuality
is to extend a local \emph{homomorphism} to a global homomorphism.
We give two examples, based on GHZ and Mermin's square.
In Section \ref{section:topological-obstruction} we show
that these scenarios can be given a bundle structure,
and that the cohomological obstruction to extending
local splittings can be used to detect contextuality.

Let $\mathcal{O} \subset P_{n,d}$ be a closed set of Weyl operators.
$\mathcal{O}$ contains the identity
and is closed under products of commuting operators.
Restricting the group product of $P_{n,d}$ to pairs of commuting operators
, therefore, gives $\mathcal{O}$ the structure of a partial monoid.
The set of operators $\mathcal{O}$ is also closed under 
the $\mathbb{Z}_d$-action
\begin{align}
    \Omega : \mathbb{Z}_d \times \mathcal{O} \to \mathcal{O} := (p, O) \mapsto \omega^p O
\end{align}
where $\omega := e^{2 \pi i / d}$. Note that $\Omega$
is compatible (Definition \ref{def:action}) 
with the partial monoid product because
\begin{align}
    \omega^p O \omega^{p'} O' = \omega^{p+p'}OO'
\end{align}
for all $p \in \mathbb{Z}_d$ and commuting $O,O' \in \mathcal{O}$.

Let $C \subset \mathcal{O}$ be a maximal context of pairwise
commuting operators. $C$ is then a \emph{submonoid}
of $\mathcal{O}$, i.e. all products are defined.
Because two operators $\omega^q O$ and $O$ that differ
by some $\omega^q$ commute we also have that $C$ is closed
under the action $\Omega$. Write $\Omega_C$ for the restriction
of $\Omega$ to $C$. We now observe that
any value assignment $s:C \to \mathbb{Z}_d$ that is consistent
with quantum mechanics preserves the action $\Omega_C$.
\begin{lemma}
    Let $\mathcal{O} \subset P_{n,d}$ be a closed set of Weyl operators
    and $C \subset \mathcal{O}$ a maximal context.
    A joint outcome assignment $s:C \to \mathbb{Z}_d$
    that is consistent with quantum mechanics
    is an \emph{action homomorphism} (Definition \ref{def:action-homomorphism}) 
    from $\Omega_C$ to 
    the action $\theta_{\mathbb{Z}_d}$ of $\mathbb{Z}_d$ on itself.
\end{lemma}
\begin{proof}
    A measurement of a Weyl operator $M$
    with outcome $q$ correspond to the $\omega^q$
    eigenvalue.
    Let $s:C \to \mathbb{Z}_d$ be an outcome assignment and suppose that
    $s$ is consistent with quantum mechanics.
    Let $M_1, M_2 \in C$. There is then some state
    $\psi$ such that $\psi$ is an eigenvector
    of $M_1,M_2,M_1M_2$ with
    eigenvalues $s(M_1), s(M_2), s(M_1M_2)$.
    \begin{align}
        M_1 \ket{\psi} &= \omega^{s(M_1)} \ket{\psi}\\
        M_2 \ket{\psi} &= \omega^{s(M_2)} \ket{\psi}\\
        M_1M_2 \ket{\psi} &= \omega^{s(M_1M_2)} \ket{\psi}\\
    \end{align}
    Using the first two equations
    we have
    \begin{align}
        M_1M_2\ket{\psi} = \omega^{s(M_1)}\omega^{s(M_2)}\ket{\psi}
    \end{align}
    Comparing this to the third equations we have
    $s(M_1M_2) = s(M_1) + s(M_2)$.
    That $s(I) = 0$ is clear because
    $I$ only has one eigenvalue $1$ which
    is identified with $0$.
    Finally,
    if we multiply a Weyl operator $M$ with
    a scalar $\omega^q$, then the effect is to
    permute the eigenvalues, hence
    $s(\omega^q M) = q + s(M)$.
\end{proof}

From this it follows that both the state-dependent and state independent
empirical models $\pmodel_{\cWeylOps}, 
\pmodel_{\cWeylOps, \psi}:\weylScenario$
are instances of the following definition.
\begin{definition} \label{def:bundle-scenario}
    Let $(X, \mathcal{M}, G)$ be a measurement scenario
    equipped with the following additional structure:
    \begin{enumerate}
        \item The outcomes $G$ is a commutative group.
        \item Each maximal context $C \in \mathcal{M}$
            is a commutative monoid
            with a compatible action $\theta_C:G \times C \to C$,
            such that for all maximal contexts $C,C' \in \mathcal{M}$,
            $g \in G$, $x,x' \in C \cap C'$:
            \begin{align}
                0_C &= 0_{C'}\\
                x +_C x' &= x +_{C'} x'\\
                \theta_C(g,x) &= \theta_{C'}(g,x)
            \end{align}
    \end{enumerate}
    An \emph{empirical model}
    $\mathcal{S}:(X, \mathcal{M}, G)$
    is an empirical model (in the usual sense),
    such that every local section $s \in \mathcal{S}(C)$
    is an action homomorphism from $\theta_C$
    to the action $\theta_G$ of $G$ on itself.
\end{definition}

\subsection{Detecting contextuality with homomorphisms}
\label{section:detecting-homomorphisms}

Suppose that $\mathcal{S} : (X, \mathcal{M}, G)$ is an empirical model
and measurement scenario with the additional structure of 
Definition \ref{def:bundle-scenario}.
We first observe that the monoid structures
on the contexts $C \in \mathcal{M}$ and the actions
$\theta_C$ ``glue together'' to define a \emph{partial} monoid
and compatible action on $X$.
\begin{align}
    \theta(g, x) &:= \theta_C(g,x)\\
    0 &:= 0_C\\
    x + x' &:= x +_C x'
\end{align}
for any $C \in \mathcal{M}$, $x,x' \in C$, and $g \in G$. That this is well
defined follows from the compatibility conditions in 
Definition \ref{def:bundle-scenario}).

Note that in the case of a closed set of Weyl operators
$\mathcal{O} \subset P_{n,d}$ the partial monoid structure
on $\mathcal{O}$ is the structure given by gluing together
the monoid structure on each maximal context.
This is the case because two operators $O,O' \in \mathcal{O}$ commute if and only
if they are both contained in a maximal context $C \subset \mathcal{O}$.
As a partial monoid $\mathcal{O}$ is therefore completely defined
by its restriction to the maximal submonoids $C \subset \mathcal{O}$.

Because the action $\theta$ and partial monoid structure
on $X$ are completely determined by the monoids and actions
on the maximal contexts, it follows that $s:X \to G$
is an action homomorphism from $\theta$ to $\theta_G$
if and only if $\res{s}{C}$ is an action homomorphism
from $\theta_C$ to $\theta_G$ for all maximal contexts $C \in \mathcal{M}$.
We can therefore consider the problem of extending a local
action homomorphism $s' \in \mathcal{S}(C)$ defined on a maximal
context $C \in \mathcal{M}$ as a test for contextuality.

This test is sound, but not necessarily complete. There
can be action homomorphisms $s:X \to G$ that are not global
sections of $\mathcal{S}$.

We now give two examples, showing that in the case of GHZ and
Mermin's square the homomorphism condition detects contextuality.

\begin{exmp}[Mermin's square]
Let $X \subset P_2$ be any set of Pauli measurements that is closed under products
of commuting measurements
and contains the measurements displayed in Mermin's square.
We consider the state independent model 
$\pmodel_X:(X,\mcvx, \zn{2})$ which in this case satisfies Definition \ref{def:bundle-scenario}.

Observe that equations (1)-(6) induced by Mermin's square all can be rearranged
to be on the form
\[
    M_1 \oplus M_2 =  M_1 M_2 \]
for $M_1,M_2 \in X$ with $M_1M_2=M_2M_1$. That the equations
are mutually inconsistent therefore literally says that there is no
homomorphism from $X$ to $\mathbb{Z}_2$. 
\end{exmp}

GHZ is an example of state-dependent contextuality. We first show that the set of operators whose outcome when measuring a given state is deterministic
is a submonoid of the total partial monoid.
\begin{lemma}
    Let $\cWeylOps \subset P_{n,d}$ be a closed set of $n$-qudit
    Weyl operators. Let $\psi$ be a state and
    $\cWeylOps_\psi \subset \cWeylOps$ the subset of operators
    whose outcome on $\psi$ is deterministic
    \begin{align}
        \cWeylOps_\psi := \{M \in \cWeylOps \mid M\ket{\psi} = \omega^q \ket{\psi},
        \text{ for some $q \in \wPhaseGrp$}\}
    \end{align}
    $\cWeylOps_\psi$ is a sub monoid of $\cWeylOps$.
\end{lemma}
\begin{proof}
    That the identity operator $I$ is contained in $\cWeylOps_\psi$
    is clear. If $M,M' \in \cWeylOps$ with
    outcomes $q,q'$ then $MM'$
    has outcome $q+q'$.
\end{proof}

In the case of state-dependent contextuality, we try to extend the unique
value assignment consistent with $\psi$.
\begin{exmp}[GHZ]
\label{example2}
Let $X := \bigotimes_{i=1}^3 \pm \{\pauli{X}{}, \pauli{Y}{}, \pauli{Z}{}, I\}$. First note
that the state-dependent model
$\pmodel_{X,\text{GHZ}}:(X, \mcvx, \zn{2})$ is an instance of Definition \ref{def:bundle-scenario} because $X$
is closed under commuting products and contains $\pm I$.
Next, consider the set $X_\text{GHZ}$ of measurements whose outcome is uniquely determined
by $\ghz$ and observe that
the equations in the GHZ example \ref{example:contextuality-ghz}
are all of the form
\[
    M_1 \oplus M_2 \oplus M_3 = s_{\text{GHZ}}(M_1M_2M_3) \]
where 
$M_1,M_2,M_3 \in X$ are compatible, $M_1M_2M_3 \in X_{\text{GHZ}}$,
and $s_{\text{GHZ}}(M_1M_2M_3)$ is the unique outcome that is consistent with $\ghz$.
That the equations are mutually inconsistent therefore ensures that there is no
global action homomorphism $g:X \to \zn{2}$ whose restriction to $X_\text{GHZ}$ is
$s_\text{GHZ}$.
It follows that if $C \in \mcvx$ is any context 
that contains $X_{\text{GHZ}}$ then there is no
$s \in \pmodel_{X, \text{GHZ}}(C)$
that can be extended to a global action homomorphism.
Note that such a context exists because the maximal submonoids of $X$ are the contexts
and by Lemma 4.1 $X_\text{GHZ}$ is a monoid.
\end{exmp}

\subsection{The cohomological obstruction}
\label{section:topological-obstruction}
We now generalise the cohomological obstruction of Okay et al.\ to any
empirical model with the structure of Definition \ref{def:bundle-scenario}.
Let $\mathcal{S} : (X, \mathcal{M}, G)$ be an empirical model
and measurement scenario according to Definition \ref{def:bundle-scenario}.
We first show that the scenario comes with the structure
of a bundle over a partial monoid.

At $X$, and at each maximal context $C \in \mathcal{M}$
we take the quotient partial monoid with respect to the actions
$\theta$, $\theta_C$.
\begin{definition}
    Let $G$ be a group, $M$ a partial monoid, and $\theta:G \times M \to M$ an action of 
    $G$ on $M$. The \emph{quotient partial monoid} is the set of orbits 
    $M / \theta$ with identity and product defined by
    \begin{align}
        0_{M / G} &:= [0_M]_{\theta}\\
        [m_1]_{\theta} +_{M / G} [m_2] &:= [m_1 +_M m_2]_{\theta}
    \end{align}
    for all $m_1, m_2$ such that $m_1 +_M m_2$ are defined.
\end{definition}

Observe that the product operation of $M / \theta$ is well defined because $\theta:G \times M \to M$ is a homomorphism. For all $a_1, a_2 \in A$ and
$m_1,m_2 \in M$ such that $m_1 +_M m_2$ is defined
we have $(a_1 \cdot m_1 +_M a_2 \cdot m_2 = (a_1 +_A a_2) \cdot (m_1 +_M m_2)$.
Therefore
\begin{align}
    [a_1 \cdot m_1] +_{M / G} [a_2 \cdot m_2] = [(a_1 +_A a_2) \cdot (m_1 +_M m_2)]
\end{align}
as required.

$X$ is then a bundle over the quotient partial monoid
$X / \theta$, with action $\theta$ and projection map $[\_]_\theta:X \to X / \theta$,
and for each maximal context $C \in \mathcal{M}$
the monoid $C$ is a bundle over $C / \theta_C$
with action $\theta_C$ and projection map $[\_]_{\theta_C} : C \to C / \theta_C$.
That each local section $s \in \mathcal{S}(C)$
is an action homomorphism from $\theta_C$ to $\theta_G$
means that it is a left splitting of the bundle.
The cohomological obstruction for $s$
is the obstruction to extending the right splitting
$\mathcal{R}(s)$ to a global right splitting
of the bundle $X$.

\begin{definition} \label{def:generalised-obstruction}
    Let $\mathcal{S}:(X, \mathcal{M}, G)$ be a measurement scenario
    and empirical model satisfying Definition \ref{def:bundle-scenario}.
    For any maximal context $C \in \mathcal{M}$
    and $s \in \mathcal{S}(C)$
    the cohomological obstruction
    is the obstruction $\mu(\mathcal{R}) \in H^2(X/\theta,C/\theta_C;G)$
    for the right splitting $\mathcal{R}$
    to extend to a global splitting.
\end{definition}

Because the homomorphism test for contextuality is not necessarily
complete the cohomological obstruction is also not necessarily complete.
However, from the examples in the previous section
we have that it detects contextuality in the case of GHZ
and Mermin's square.

\section{Comparing two obstructions} \label{section:comparison}

We conclude this chapter by comparing the topological obstruction
of Okay et al.\ \ref{def:bundle-scenario} to the \v{C}ech cohomology obstruction.
The \v{C}ech cohomology obstruction and the topological obstruction
are different in the type of algebraic structure they are defined with.
The topological approach relies upon a pre-existing
structure in the measurement scenario and empirical model.
The \v{C}ech cohomology approach does not require any pre-existing structure,
instead, it uses a free construction to give the required structure to any 
empirical model. 

At first glance, it might be surprising that
any interesting structure is left behind by this free construction.
An explanation for why the \v{C}ech cohomology obstruction detects contextuality is that many 
examples of contextuality in quantum
mechanics are of the AvN type. For the examples, GHZ and magic square, where we have shown
that the topological approach also detects contextuality. 
The question is then if there are instances of contextuality
that are detected by the topological approach, but not the \v{C}ech cohomology
approach. We now show that this is not the case.

Recall that a false negative of the \v{C}ech cohomology approach
occurs when a local section $s$ can be extended to a compatible
family of the pre-sheaf $F_\mathbb{Z} \mathcal{S}$.
If an empirical model $\mathcal{S}$ satisfies Definition \ref{def:bundle-scenario}
then each local section in this formal linear combination
is an action homomorphism. 
Because homomorphisms are closed under affine combinations we can collapse
the formal affine combination to a global splitting of the bundle.
\begin{theorem}
    Let $\mathcal{S}:(X, \mathcal{M}, G)$ be an empirical
    model of Definition X.
    Let $l \in \mathcal{S}(C)$ be a local section.
    If the \v{C}ech cohomology obstruction
    $\gamma(l)$ vanishes, then
    the topological obstruction $\mu(l)$ vanishes:
    \begin{align}
        \gamma(s_0) = 0 \implies \mu(s_0) = 0
    \end{align}
\end{theorem}
\begin{proof}
    Suppose that $s_0 \in \pmodel(C_0)$ is a local section
    such that the \v{C}ech cohomology obstruction vanishes, $\gamma(s_0) = 0$. We need to show that $s_0$ extends to a global right splitting.
    
    Recall that if the measurement cover is connected
    and $\gamma(s_0) = 0$ then there is some
    compatible family
    $\{r_C \in F_\ints \pmodel(C)\}_{C \in \mcvx}$ such that
    $r_{C_0} = 1 \cdot s_0$. 
    Now, any measurement cover by commutative monoids is connected
    because the identity element is contained in all contexts.
    We can therefore take such a family $\{r_C\}_{C \in \mcvx}$.
    Observe now that any such family in fact is a
    compatible family of
    formal affine combinations: For any $C \in \mcvx$
    \[
    \sum_{s \in \pmodel(C)} 
            r_C(s) \cdot \res{s}{C \cap C_0}
            =
        \res{r_C}{C \cap C_0} = 
        \res{r_{C_0}}{C \cap C_0} = 1 \cdot \res{s_0}{C \cap C_0}
        \]
    hence $\sum_{s \in \pmodel(C)} r_C(s) = 1$.
    
    We now use the unique module
    action\footnote{
    i.e. $0 \cdot a = 0$, and for $n \geq 1$: 
    $n \cdot a := a + a + \cdots + a$ ($n$ times) and $-n \cdot a = -(n \cdot a)$.
    }of $\ints$ on $A$ to collapse this
    compatible family to a function $g:X \to A$.
    \[
        g(x) := \sum_{s \in \pmodel(C)} r_C(s) \cdot s(x), \quad \text{where
        $C \in \mcvx$ is any context with $x \in C$}
        \]
    Because the set of splittings is closed under affine combinations
    this function is in fact a splitting that furthermore extends $s_0$.
\end{proof}

Another point is that the topological approach uses the equivalence between
right and left splittings. A question we can ask is what the cohomology classes actually
mean.

\chapter{From contextuality to shallow circuits: a general construction of quantum advantage}

In this chapter, we present a generalised version of Bravyi, Gosset,
and K\"{o}nig's quantum advantage result with shallow circuits.
The quantum circuit $\{Q_n\}_{n \in \nats}$
and the computational problem $\{\text{2D-GHZ(n)}\}_{n \in \nats}$
introduced by BGK are relatively simple to define.
Their main technical contribution is the technique used
to prove the classical bound. We start by explaining
how this technique can be recast in the sheaf theoretic framework.

The quantum circuit $Q_n$ defines a mapping from inputs
to distributions over outputs 
$\tilde{Q}_n :: x \mapsto \sum_{y} \tilde{Q}(x,y) \cdot y$. 
We can think of this mapping as an empirical model
for a multipartite scenario.
$\tilde{Q}_n$ is related to the strategy
$e_\text{GHZ}$ for the GHZ game by a simulation
$s_n$.
\begin{align}
    (s_n)_*(\tilde{Q}_n) = e_\text{GHZ}
\end{align}
The 2D-GHZ problem is the pullback $(s_n)^*(\text{GHZ})$
of the GHZ-game across $s_n$.
Because the strategy $e_\text{GHZ}$ solves the GHZ-game perfectly
it, therefore, follows that $Q_n$ also solves the GHZ game perfectly.
\begin{align}
    p_S(Q_n, \text{2D-GHZ}) = p_S((s_n)_*(\tilde{Q}_n), \text{GHZ}) = 1
\end{align}

We can similarly think of a classical shallow circuit $\{C_n\}_{n \in \nats}$
as defining a family of empirical models $\{\tilde{C}_n\}_{n \in \nats}$.
This is strictly speaking not true because
$\tilde{C}_n :: x \mapsto \sum_{y} \tilde{C}_n(x,y) \cdot y$
does not necessarily satisfy the no-signaling assumption.

Recall the resource inequality
\begin{align}
    p_S(e, \Phi) \leq \gamma + \cf{e}
\end{align}
relating the success probability of an empirical model $e$
on a non-local game $\Phi$ to the classical bound $\gamma$
and the contextual fraction $\cf{e}$.

Using this inequality we can bound the success probability
of $C_n$ on $\text{2D-GHZ}(n)$ in terms of the contextual fraction
of the pushforward $(s_n)_*(\tilde{C}_n)$.
\begin{align}
    p_S(C_n, \text{2D-GHZ}(n)) = p_S((s_n)_*(\tilde{C}_n), \text{GHZ}) \leq 
    3/4 + \cf{(s_n)_*(\tilde{C}_n)}
\end{align}

The technically most involved part of their result
is to establish a bound on $\cf{(s_n)_*(\tilde{C}_n)}$.
To do this they combine two results.
The simulation $s_n$ is a probability distribution
over a deterministic simulation
$t$ for each choice of players 
$v_A, v_B, v_C$ and paths $u_\text{AB}, u_\text{BC}, u_\text{CA}$.
BGK first gives a combinatorial condition involving the paths
and the circuit $C$ ensuring that the pushforward
$t_*(\tilde{C}_n)$ is non-contextual.
They then prove that when the paths are chosen sufficiently
uniformly then the probability of this condition being satisfied
is high.

\subsection{Structure of chapter}
In Section \ref{section:circuits-as-ontological-models}
we give some elementary background on circuits,
and we make the idea that circuits can implement
empirical models and strategies for non-local games precise.
In Section \ref{section:distributing-line} we present a protocol
based on teleportation that allows a number of agents
to implement measurements on a single-qudit state at arbitrarily long distances
along a line.
In Section \ref{section:distributing-graph} we generalise
the construction from Section \ref{section:distributing-line}
to a protocol that allows us to simulate measurements
in a distributed way.
In Section \ref{section:distributing-games} we show
that the distributed simulation protocol can be used
to construct non-local games that are solved by quantum
circuit of small depth and fan-in.
In Section \ref{section:bounding-cf} we restate BGK's technique
for proving their classical bound in the sheaf theoretic framework,
and we use it to derive a bound for the games introduced in Section
\ref{section:distributing-games}.
In Section \ref{section:qa-result} we put everything
together and show that the construction can be used
to derive unconditional quantum advantage results
with shallow circuits from any non-local game.

\section{Circuits} \label{section:circuits-as-ontological-models}
\begin{figure}
    \centering
    \begin{subfigure}{0.25\textwidth}
        \centering
        \includegraphics[width=\textwidth]{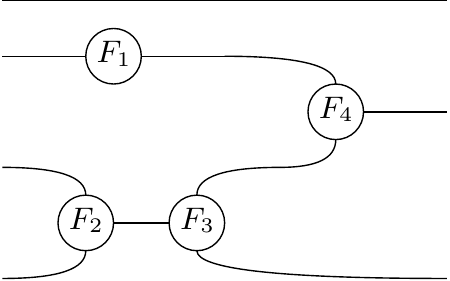}
        \caption{}
    \end{subfigure}
    \hfill
    \begin{subfigure}{0.25\textwidth}
        \centering
        \includegraphics[width=\textwidth]{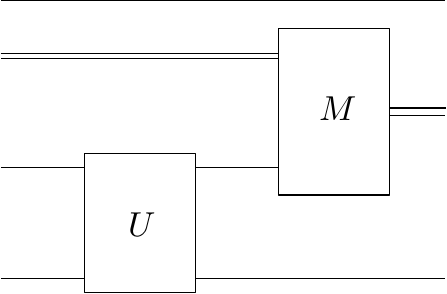}
        \caption{}
    \end{subfigure}
    \hfill
    \begin{subfigure}{0.25\textwidth}
        \centering
        \includegraphics[width=\textwidth]{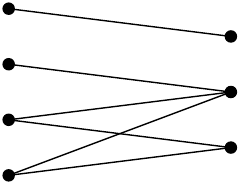}
        \caption{}
    \end{subfigure}
    \caption{(a) A classical circuit with gates $F_1, F_2, F_3, F_4$. (b) A quantum circuit with a unitary gate $U$ and a classically controlled measurement gate $M$. Classical wires are drawn as double lines and quantum wires as single lines. (c) The lightcone relationship between inputs and outputs in both circuit (a) and(b).}
    \label{fig:my_label}
\end{figure}

A \emph{circuit} can formally be defined as a directed acyclic graph whose 
nodes are either input wires, output wires, or gates. The graph structure
defines the order of evaluation for the gates. To evaluate
a circuit we first fix an input value for each of the input wires,
then evaluate the gates in the order given by the graph,
and finally return the values of the output wires.

Both the quantum and classical circuits that we work with
have only classical input and output wires. A quantum
circuit additionally uses some number of qudits
initially prepared in the computational 
basis state. In a classical circuit,
gates are probabilistic. In a quantum circuit gates
are classically controlled unitaries and measurement gates.

The \emph{depth} of a circuit is the length of the longest path
from an input to an output.
The fan-in of a gate is its number of inputs,
and the \emph{maximal fan-in} of a circuit is the maximal fan-in over all of the gates.
We say that a family of circuits is shallow if it has both bounded depth
and maximal fan-in.
\begin{definition}
    A \emph{shallow circuit} is a family
    of circuits $F = \{F_n\}_{n \in \nats}$
    for which there exists $K,D \in \nats$ such
    that $F_n$ has depth at most $D$ and maximal fan-in
    at most $K$ for all $n \in \nats$.
\end{definition}

The graph structure of a circuit restricts the possible dependencies
between input wires and output wires. These dependencies are captured
by the \emph{lightcones} of the circuit.
\begin{definition}
    Let $F$ be a circuit with input wires 
    $\{\text{in}_i\}_{i \in I}$ labelled by a set $I$
    and output wires  $\{\text{out}_j\}_{j \in J}$ labelled by a set $J$.
    \begin{itemize}
        \item The \emph{forward lightcone} of $i \in I$,
        denoted by $\lcf{F}(i)$,
        is the set of output wires $j \in J$
            such that there is a path in $F$ from $\text{in}_i$
            to $\text{out}_j$.
        \item The \emph{backward lightcone} of $j \in J$,
        denoted by $\lcb{F}(j)$,
        is the set of all input wires $i \in I$
        such that there is path in $F$ from $\text{in}_i$
        to $\text{out}_j$.
    \end{itemize}
\end{definition}

In a circuit $F$ of depth $D$ and maximal fan-in $K$
each output wire is reachable from at most $K^D$
input wires, $\supp{\lcb{F}(\text{out}_i)} \leq K^D$
for all output wires $\text{out}_i$.
If a shallow circuit $\{F_n\}_{n \in \nats}$
has maximal depth $D$ and fan-in $K$ and $F_n$
has $n$ inputs, then the fraction of inputs
that can reach a given output tends to zero as $n$
increases.

\subsection{Circuit strategies}
BGK observed that the quantum strategy for the GHZ game can be seen 
as a simple circuit. 
Figure \ref{fig:ms-circuit}
shows a similar example from \cite{bravyi_quantum_2020}
based on the Magic Square game.
In this section, we define \emph{circuit strategies}
for cooperative games with one or more rounds.
We then define the \emph{behaviour} of a circuit,
and we show that ideas like simulations and the contextual
fraction can be applied to these objects.

\begin{figure}
    \centering
    \begin{subfigure}{0.45\textwidth}
        \includegraphics[width=1\textwidth]{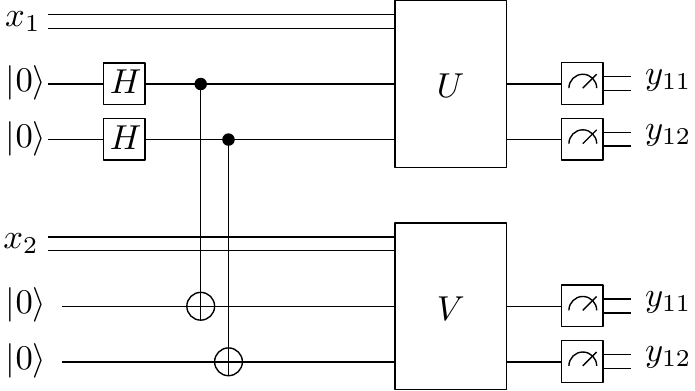}
        \caption{}
    \end{subfigure}
    \hfill
    \begin{subfigure}{0.45\textwidth}
        \includegraphics[width=1\textwidth]{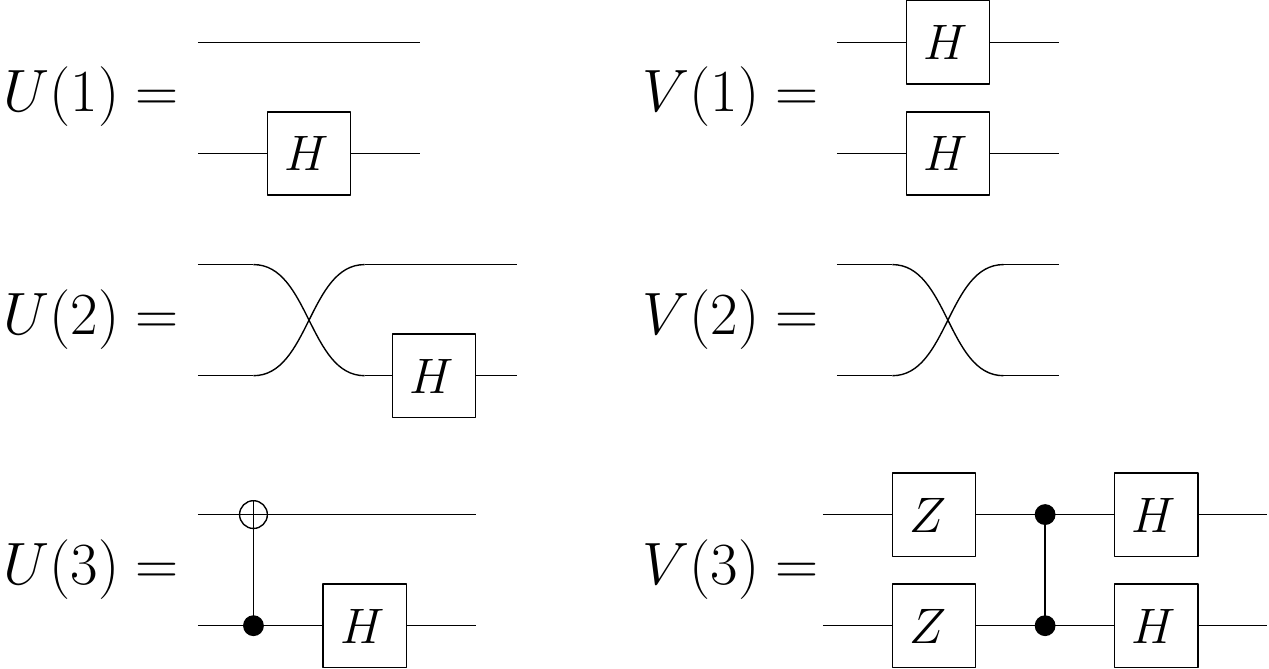}
        \caption{}
    \end{subfigure}
    \caption{Circuit version of the quantum strategy for the Magic Square game.
        First, two maximally entangled states are prepared using
        Hadamard gates and controlled not gates.
        The classical inputs $x_1, x_2 \in \{1,2,3\}$ are used to perform
        a basis change that is followed by computational basis measurements.
        We observe that under the controlled basis changes
        the computational basis
        measurements are equivalent to performing the measurements
        used in the Magic Square game.
    }
    \label{fig:ms-circuit}
\end{figure}

We first consider the case of a single-round game.
Let $S$ be a multipartite scenario.
A circuit strategy is a classical or quantum
circuit with a classical input wire and a classical output wire
for each measurement site. In a multipartite
scenario, joint measurements specify at most one measurement
setting for each measurement site. The input
to a circuit strategy consists of a measurement setting
for each wire, or a symbol ``$\bullet$'' denoting no measurement.
The output of the circuit consists of an output
for each measurement site that has been measured, or ``$\bullet$''
for the measurement sites that have not been measured.
\begin{definition}
    Let $S = (I, X, Y)$ be a multipartite scenario. A (single-round) \emph{circuit
    model} is a classical or quantum circuit $F$
    with an input wire $\text{in}_i$ and an output wire $\text{out}_i$
    for each measurement site $i \in I$, such that:
    \begin{itemize}
        \item $\text{in}_i$ has type $X_i \sqcup \{\bullet\}$
        \item When $F$ is evaluated and the input to $\text{in}_i$ is $\bullet$
        then the output of $\text{out}_i$ is also $\bullet$,
        otherwise if the input to $\text{in}_i$ is $x \in X_i$
        then the output of $\text{out}_i$ has type $Y_{i,x}$.
    \end{itemize}
\end{definition}

We additionally allow circuit strategies to sample a random seed.
Let $F$ be a circuit strategy for a multipartite scenario
$S = (I, X, Y)$. Given a joint measurement $C$
we evaluate $F$ by setting each input wire $\text{in}_i$
where $C$ specifies a measurement setting to this value,
and otherwise to $\bullet$.
We then read off the values of $\text{out}_i$
and return the joint outcome $s \in \mathcal{E}_S(C)$
by setting $s(i,x)$ to the value of output wire $i$.
This defines a family of probability distributions
$\tilde{F} = \{\tilde{F}(C) \in \mathcal{E}_S(C)\}$
which we call the \emph{behaviour} of $F$.

\begin{definition}
    Let $S$ be a multipartite scenario and $F$ a circuit strategy.
    For any context $C$ of $S$ and local section
    $s \in \mathcal{E}_S(C)$
    the probability
    $\tilde{F}(C)(s)$ is the probability
    that the output of $\text{out}_i$
    is $y_i$ when the input to $\text{in}_i$
    is $x_i$ for all $i \in I$, where
    \begin{figure}[h]
        \centering
        \begin{subfigure}{0.45\textwidth}
            \centering
            \includegraphics{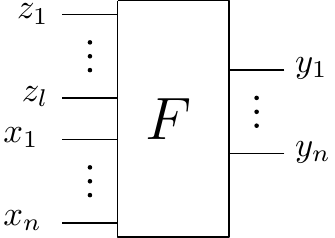}
        \end{subfigure}
        \begin{subfigure}{0.45\textwidth}
            \centering
            \begin{align}
                x_i &:= \begin{cases}
                    x, &\text{if $(i,x) \in C$}\\
                    \bullet, &\text{otherwise}
                \end{cases}\\
                y_i &:= \begin{cases}
                    s(i,x), &\text{if $(i,x) \in C$}\\
                    \bullet, &\text{otherwise}
                \end{cases}
            \end{align}
        \end{subfigure}
        \caption{}
        \label{fig:1r-strategy}
    \end{figure}
    and the random seed $z$ is sampled randomly.
\end{definition}

Any quantum strategy for a non-local game
can be seen as a circuit strategy. Let $S = (I, X, Y)$
be a multipartite scenario, $e$ a quantum strategy for a non-local
game, using a multi-qudit state $\psi$
and a single-qudit measurement $M_{i,x}$
for each measurement $(i,x)$ of $S$.
We define the quantum circuit strategy
$Q_{\psi,M}$ to be the circuit
with a single unitary gate $U_\psi$ that prepares
the state $\psi$, and performs a classically controlled
$M_{i,x}$ measurement for each qudit
and returns the outcomes.
By definition the behaviour of $Q_{\psi,M}$ is equivalent
to the empirical model $e$
\begin{align}
    \tilde{Q}_{\psi,M}(U) = e_{\psi,M}(U)
\end{align}
for each measurement context $U$ of $S$.

We will now define circuit strategies for games with more than one round.
An $n$-round circuit strategy is a classical or quantum circuit $F$
with a classical input wire and a classical output wire
for each measurement site and each of the $n$ rounds (Figure \ref{fig:generator-two-round}).
We require that the output wires for round $j$ are not reachable
from the input wires in round $j' > j$.
\begin{definition}
    Let $S = (I, X, Y)$ be a multipartite scenario and $n \geq 1$. 
    An $n$-round \emph{circuit strategy}
    is a circuit $F$ with input wires
    $\{\text{in}_{i,j}\}_{i \in I, j = 1,\dots, n}$
    and output wires $\{\text{out}_{i,j}\}_{i \in I, j = 1,\dots,n}$, such that:
    \begin{itemize}
        \item Input wire $\text{in}_{i,j}$ has type $X_i$.
        \item When $F$ is evaluated and the input to $\text{in}_{ij}$ is $\bullet$
            then the output of $\text{out}_{ij}$ is also $\bullet$,
            otherwise if the input to $\text{in}_{ij}$ is $x \in X_i$
            then the output of $\text{out}_{ij}$ has type $Y_{i,x}$.
        \item $\text{in}_{ij}$ is not reachable from $\text{out}_{i'j'}$
            when $j' \leq j$: $\text{in}_{i,j} \notin \lcf{F}(\text{out}_{i', j'})$.
        \end{itemize}
\end{definition}

Interactive circuits can be thought of in two ways.
If $i < i'$ then $\text{out}_{i,j}$ is not reachable
from $\text{in}_{i',j'}$ for any $j,j'$.
It is therefore possible to partially evaluate the circuit
on the inputs in round $i$ without
fixing the values for the inputs in round $i'$.
Alternatively, we can think of the evaluation
as being performed by ``plugging in'' a classical circuit
$C$ (Figure \ref{fig:two-round-evaluation}).

\begin{figure}
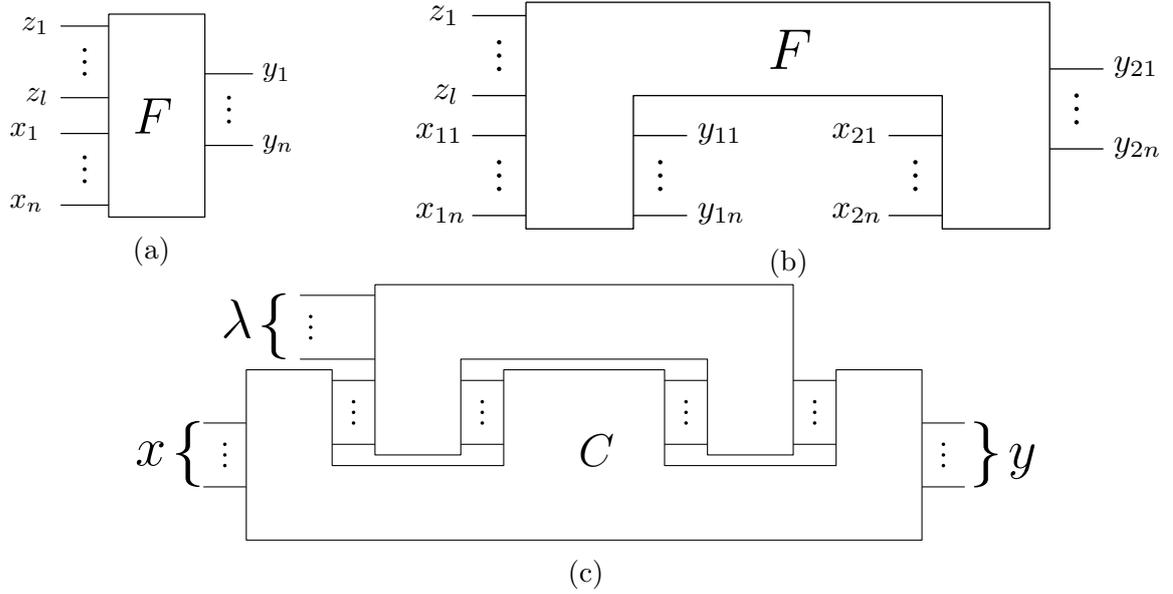

    \centering
    \begin{subfigure}{1\textwidth}
        \centering
        \begin{subfigure}{0.25\textwidth}
            \centering
            \includegraphics[width=\textwidth]{single-round-strategy.pdf}
            \caption{}
        \end{subfigure}
        \hfill
        \begin{subfigure}{0.65\textwidth}
            \centering
            \includegraphics[width=\textwidth]{classical-generator-two-round.pdf}
            \caption{} \label{fig:generator-two-round}
        \end{subfigure}
    \end{subfigure}
    \begin{subfigure}{1\textwidth}
        \centering
        \includegraphics{two-round-evaluation.pdf}
        \caption{} \label{fig:two-round-evaluation}
    \end{subfigure}
    \caption{(a) and (b) shows the inputs and outputs from respectively a 
    single and two-round circuit strategy. 
    Here $x_1, \dots, x_n$ and $y_1, \dots, y_n$
    are the measurement settings and outcomes, while $z_1, \dots, z_l$
    is a random seed. Note that $F$ can be either classical or quantum.
    (c) We can think of the evaluation of (b) as being performed by
    ``plugging in'' a classical circuit $C$.
   }
\end{figure}

The behaviour of an $n$-round circuit strategy
$F$ is a family of probability distributions
$\tilde{F} = \{\tilde{F}(m) \in \dist(\mathcal{E}_S(m))\}_{m \in \text{MP}_n(S)}$
over the runs of each measurement protocol.

\begin{definition}
    Let $S$ be a multipartite scenario, $n \geq 1$,
    and $F$ an $n$-round circuit strategy.
    For each $n$-round measurement protocol $m$ and run
    $r = (U_1, s_1), \dots, (U_n,s_n)$ of $m$
    write $\tilde{F}(m)(r)$ for the probability that
    when $\{\text{in}_{ij}\}$ are set to $\{x_{ij}\}$
    and $F$ is evaluated then
    the return value of output wires $\{\text{out}_{ij}\}$ are $\{y_{ij}\}$,
    where
    \begin{align}
        x_{ij} &:= \begin{cases}
            x, &\text{if $(i,x) \in U_j$}\\
            \bullet, &\text{otherwise}
        \end{cases}\\
        y_{ij} &:= \begin{cases}
            s(i,x), &\text{if $(i,x) \in U_j$}\\
            \bullet, &\text{otherwise}
        \end{cases}
    \end{align}
\end{definition}

The behaviour of a circuit strategy $F$ is not an empirical model.
$\tilde{F}$ does not generally satisfy the no-signaling principle,
and only specifies what happens for a given number of rounds of measurements.
However, the definition of the contextual fraction,
the success probability on a game, and the pushforward can
be generalised directly.

\begin{definition}
    Let $S = (X, \mcvx, O)$ be a measurement scenario and 
    $n \geq 1$.
    An $n$-round behaviour is a family
    of probability distributions $B = \{B(m) \in \dist(\mathcal{E}_S(m)\}_{m \in \text{MP}_n(S)}$
    \begin{enumerate}
        \item  The success probability of $B$ on an $n$-round game
            $\Phi = \sum_{C, A} \Phi_{C,A} \cdot (C,A)$
            is 
            \begin{align}
                p_S(B, \Phi) := \sum_{C,A} \Phi_{C,A} \tilde{B}(C)(A)
            \end{align}
            where $C$ is an $n$-round measurement protocol
            and $A$ is a constraint on the runs of $C$.
        \item  The contextual fraction of $B$ is the least $\epsilon$ such that
        there exists a non-contextual empirical model $e$
        and another behaviour $B'$
        such that for all $n$-round measurement protocols $C$
            \begin{align}
                B(C) = \epsilon \cdot B'(C) + (1 - \epsilon) \cdot e_C
            \end{align}
        \item Let $T$ be another measurement scenario,
        $t:S \to T$ a deterministic $n$-round simulation,
        and $s:S \to T$ a probabilistic $n$-round simulation.
        The pushforward $t_*(B)$
        is the single-round behaviour for $T$
        defined by
            \begin{align}
                t_*(B)(U) = \sum_{r \in \mathcal{E}_S(f(U))} B(f(U))(r) \cdot g_U(r)
            \end{align}
        for each context $C$ of $T$.
        The pushforward $s_*(B)$ is a convex combination
        of the pushforward $t_*(B)$ for each $t$,
        with weight the probability $s(t)$ of $t$ occuring.
    \end{enumerate}
\end{definition}

\section{Performing measurements far away with teleportation} \label{section:distributing-line}
\begin{figure}
    \centering
    \includegraphics{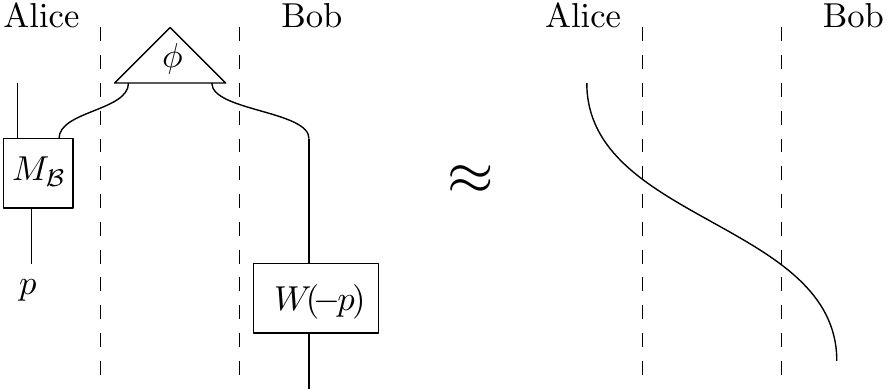}
    \caption{
        In the teleportation protocol for qudits, Alice and Bob each hold one qudit
        of a pair in the maximally entangled state $\phi$.
        Alice first measures her maximally entangled qudit and
        another qudit in the Bell basis $\mathcal{B}$ giving an outcome $p \in \mathbb{Z}_d^2$.
        If Alice's qudit is initially in the state $\psi$ then
        the post-measurement state of Bob's qudit
        is $W(p)\ket{\psi}$. It follows that if
        Bob performs the correction $W(-p)$ then
        up to an unobservable phase his qudit is in the state $\ket{\psi}$.
        In diagrammatic notation, the protocol is equivalent to the identity wire
        from Alice to Bob.
    }
    \label{fig:tp_simple}
\end{figure}

Quantum teleportation (Figure \ref{fig:tp_simple}) 
is a way of transferring a single-qudit state between two agents that share a \emph{maximally entangled state} $\phi$.
\begin{align}
    \ket{\phi} &:= \frac{1}{\sqrt{d}} \sum_{j} \ket{j} \ket{j}
\end{align}
The protocol involves a measurement in the \emph{Bell basis} $\mathcal{B}$
\begin{align}
    \mathcal{B} &:= 
        \{\ket{\phi_p} := (I \otimes W(p))\ket{\phi} \}_{p \in \mathbb{Z}_d^2}
\end{align}
performed by one of the agents, classical communication
of the measurement outcome $p \in \mathbb{Z}_d^2$
to the other agent, and finally a Weyl operator correction $W(-p)$.

\begin{figure}
    \centering
    \includegraphics[width=0.7\textwidth]{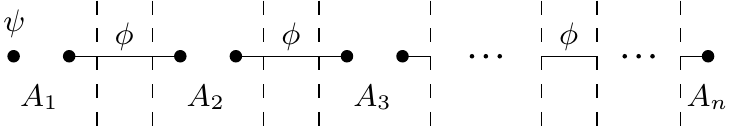}
    \caption{
        Any number of agents $A_1, \dots, A_n$ are arranged on a line, such
        that $A_1$ has a qudit $\psi$, and $A_i, A_{i+1}$, where $i=1,\dots,n-1$
        share a maximally entangled pair of qudits. We send $A_n$ a measurement setting.
        The goal is for the agents to implement a measurement on $\psi$, without
        communicating the measurement setting to any of the other agents $A_1, \dots, A_{n-1}$.
    }
    \label{fig:agents-line}
\end{figure}

In this section we extend the usual teleportation protocol
to any number of agents $A_1, \dots, A_n$ arranged on a line.
The first agent has a state $\psi$ and each consecutive pair $A_{i}, A_{i+1}$
have a maximally entangled pair of qudits. We show that $A_n$ can implement
a measurement on $\psi$ without communicating the measurement setting
to any of the other agents, such that only constantly many rounds of quantum
measurements are performed. We present two versions of the protocol.
In Section \ref{section:2r-line} we present a version for an arbitrary
quantum measurement that uses two rounds of parallel measurements.
In Section \ref{section:1r-line} we restrict to Weyl operator
measurements and we present a protocol using only a single round of parallel measurements.

\subsection{Teleportation on a line followed by a measurement} \label{section:2r-line}
\begin{figure}
    \centering
    \includegraphics{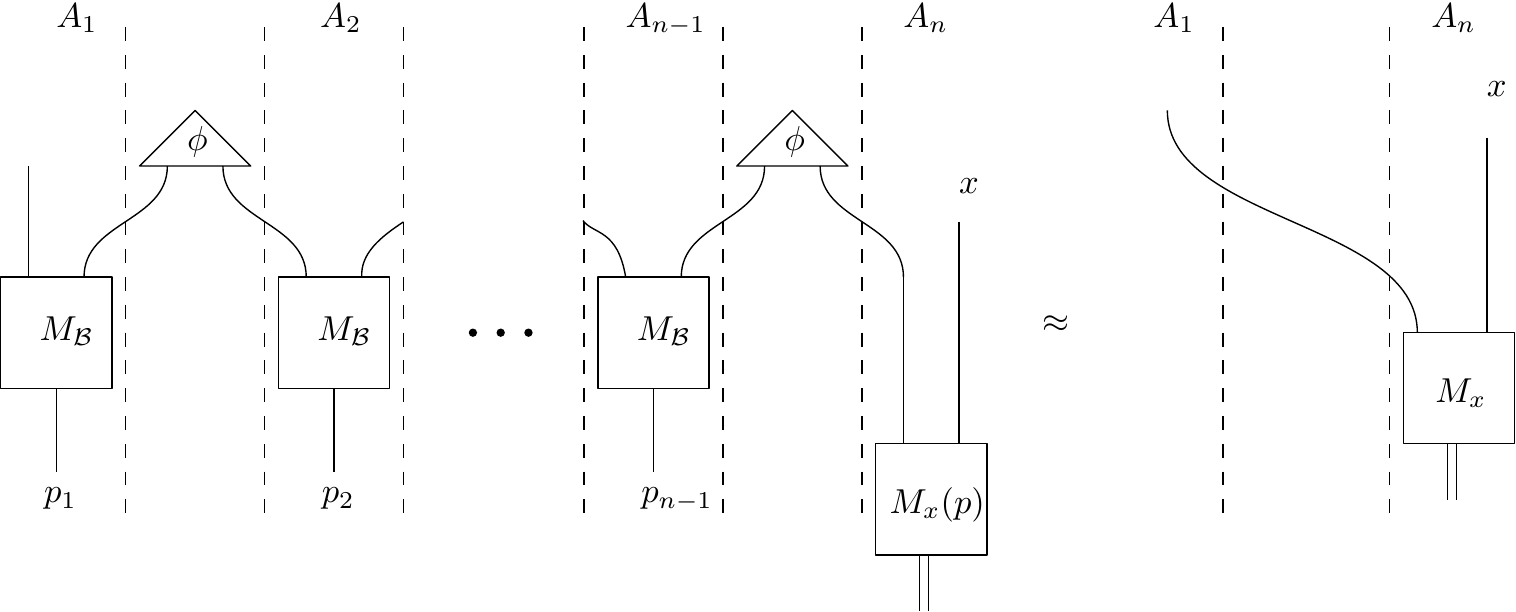}
    \caption{
        Teleportation along a line followed by a measurement,
        using two rounds of measurements.
        Let $\{M_x\}_{x \in X}$ be a family of quantum measurements.
        We send a measurement setting $x$ to $A_n$
        and $A_1, \dots, A_{n-1}$ perform Bell basis measurements.
        The respective outcomes $p_1, \dots, p_{n-1} \in \mathbb{Z}_d^2$
        are sent to $A_n$.
        In the second round $A_n$ performs the conjugated measurement
        $A_x(p) := W(p) M_x W(p)^\dagger$
        of $M_x$ with the Weyl operator $W(p)$.
        The effect of the protocol is up to an unobervable phase equivalent
        to $A_n$ performing the measurement $M_x$ on $A_1$'s qudit and returning the outcome.
    }
    \label{fig:two_round_teleportation_on_line}
\end{figure}

Let $\{M_x\}_{x \in X}$ be a family of single-qudit measurements
and $\psi$ a single-qudit state.
Consider the setup in Figure (\ref{fig:agents-line}).
Suppose that we select a measurement setting $x \in X$
and send this to $A_n$. 

A simple way for the agents to implement
the measurement $M_x$ is to first teleport $\psi$ from $A_1$ to $A_n$
then perform $M_x$. The naive way of doing this uses $n-1$ rounds.
In round $i=1,\dots,n-1$ 
a Bell basis measurement teleports $\psi$ from $A_i$ to $A_{i+1}$
up to a Weyl operator phase $W(p_i)$
which is corrected with the operator $W(-p_i)$.

Weyl operators compose up to an unobservable phase:
\begin{equation}
    W(p)W(p') \approx W(p+p')
\end{equation}

Using the composition law we can reduce the number of rounds
from $n-1$ to two (Figure \ref{fig:two_round_teleportation_on_line}). 
In the first round
agents $1,\dots,n-1$ perform Bell basis measurements
in parallel. After this the state of $A_n$'s qudit is
$W(p_1 + \dots p_{n-1})\ket{\psi}$. If $A_n$ performs the correction
$W(-p)$, where $p := p_1 + \dots + p_{n-1}$, then the effect is to teleport
$\psi$ to $A_n$, using only a single round of parallel measurements.

Instead of performing the correction $W(-p)$ and then the measurement $M_x$
the agent $A_n$ can equivalently perform a single measurement
$W(p) M_x W(p)^\dagger$.

Overall we have a protocol for implementing $M_x$ using only a single round
of parallel Bell basis measurements, followed by a single measurement
performed by $A_n$.

\subsection{Teleportation on a line followed by a Weyl measurement} \label{section:1r-line}
\begin{figure}
    \centering
    \includegraphics{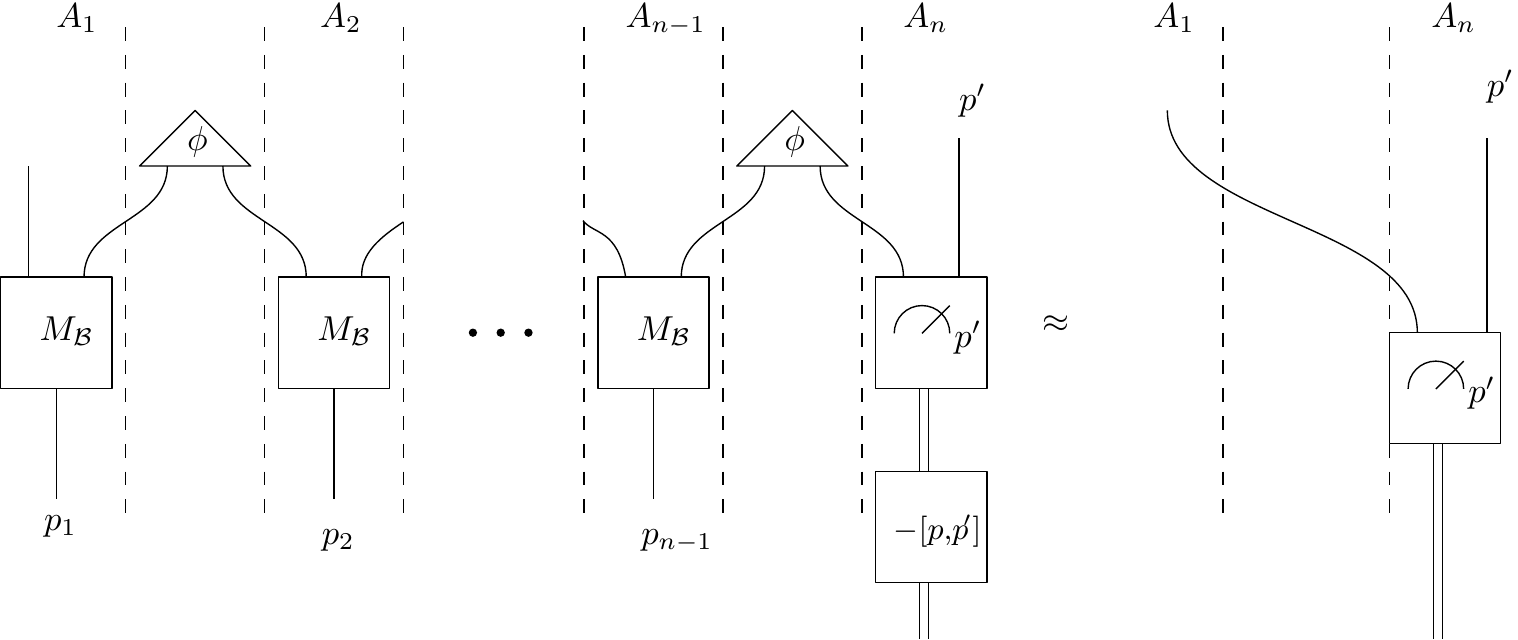}
    \caption{Teleportation on a line with a simultaneous Weyl operator measurement.
    We send a measurement setting $p' \in \mathbb{Z}_d^2$ to $A_n$.
    $A_1, \dots, A_{n-1}$ perform Bell basis measurements with
    outcomes $p_1,\dots,p_{n-1} \in \mathbb{Z}_d^2$, and $A_n$ performs a Weyl measurement
    $W(p')$ with outcome $q \in \mathbb{Z}_d$. The final outcome is $q - [p,p']$,
    where $p := p_1 + \dots + p_{n-1}$.
    The protocol is equivalent to $A_n$ measuring $W(p')$ and returning the outcome.}
    \label{fig:iterated-tp}
\end{figure}

Recall that two Weyl operators $W(p), W(p')$, where $p,p' \in \mathbb{Z}_d^2$ commute according to
\begin{align}
    W(p)W(p') = \omega^{[p,p']} W(p')W(p)
\end{align}
where $\omega = e^{2 \pi i / d}$ and
\begin{align}
    [p,p'] := p_1p_1' + p_2p_2'
\end{align}

We now consider the case of the teleportation protocol on a line
when the set of measurements we want to perform are given by Weyl operators.
Suppose that we send some Weyl measurement setting
$p \in \mathbb{Z}_d^2$ to $A_n$.

If $A_1, \dots, A_{n-1}$ perform Bell basis measurements with outcomes
$p_1', \dots, p_{n-1}'$ then $\psi$ is teleported to $A_n$ up to a phase
given by the Weyl operator $W(p')$, where $p' := p_1' + \dots + p_{n-1}'$.
We now want to perform the measurement given by $W(p)$.
Because the operators $W(p)$ and $W(p')$ commute up
to a factor $\omega^{[p,p]}$ it can be shown
that the adaptive measurement $W(p) W(p') W(p)^\dagger$
can be replaced by a measurement of $W(p')$
followed by a classical correction $-[p,p']$.

\begin{lemma} \label{lemma:permutation-lemma}
    For any $p,p' \in \mathbb{Z}_d^2$ the following
    are equivalent, up to an unobservable phase:
    A Weyl operator measurement $W(p)$
    followed by a classical correction $-[p,p']$,
    a Weyl operator correction $W(-p')$ followed
    by a Weyl operator measurement $W(p)$.
    \begin{equation}
        \vcenter{\hbox{\includegraphics[scale=1.4]{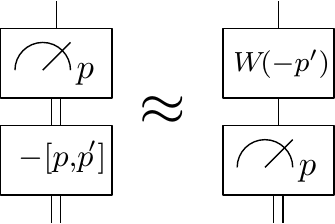}}}
    \end{equation}
\end{lemma}
\begin{proof}
    Let $\ket{p,q}$ be an $\omega^q$-eigenvector of the Weyl operator
    $W(p)$, where $p \in \mathbb{Z}_d^2$ and $q \in \mathbb{Z}_d$. 
    Let $W(p), W(p')$ be Weyl operators, where $p,p' \in \mathbb{Z}_d^2$.
    The claim is equivalent to saying that
    $W(p')$ permutes the eigenvectors of $W(p)$ by sending
    $\ket{p, q}$ to $\ket{p, q + [p,p']}$, up to an unobservable phase:
    \begin{equation}
        W(p')\ket{p,q} \approx \ket{p, q + [p,p']}
    \end{equation}
    where $q \in \mathbb{Z}_d$.
    By the commutation law of Weyl operators we have
    \begin{align}
        W(p)(W(p')\ket{p,q}) &= \omega^{[p,p']}W(p')W(p)\ket{p,q}\\
        &= \omega^{[p,p']}W(p')\omega^q\ket{p,q}\\
        &= \omega^{[p,p'] + q}(W(p')\ket{p,q})
    \end{align}
    Hence $W(p')\ket{p,q}$ is an eigenvector of $W(p)$
    with eigenvalue $q + [p,p']$, as required.
\end{proof}

From the Lemma, it is clear that the protocol in Figure \ref{fig:iterated-tp}
is equivalent to performing a Weyl measurement and returning the outcome.

\section{Distributing measurements on graphs}
\label{section:distributing-graph}
\begin{figure}
    \centering
    \begin{subfigure}{0.48\textwidth}
    \centering
        \includegraphics[width=1\textwidth]{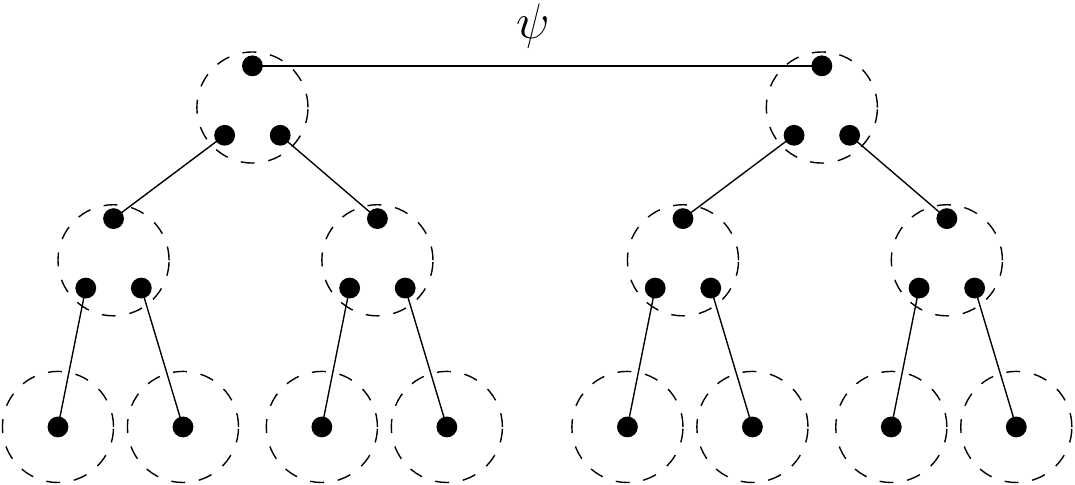}
        \caption{
        }
        \label{fig:gate-structure}
    \end{subfigure}
    \hfill
    \begin{subfigure}{0.48\textwidth}
    \centering
        \includegraphics[width=1\textwidth]{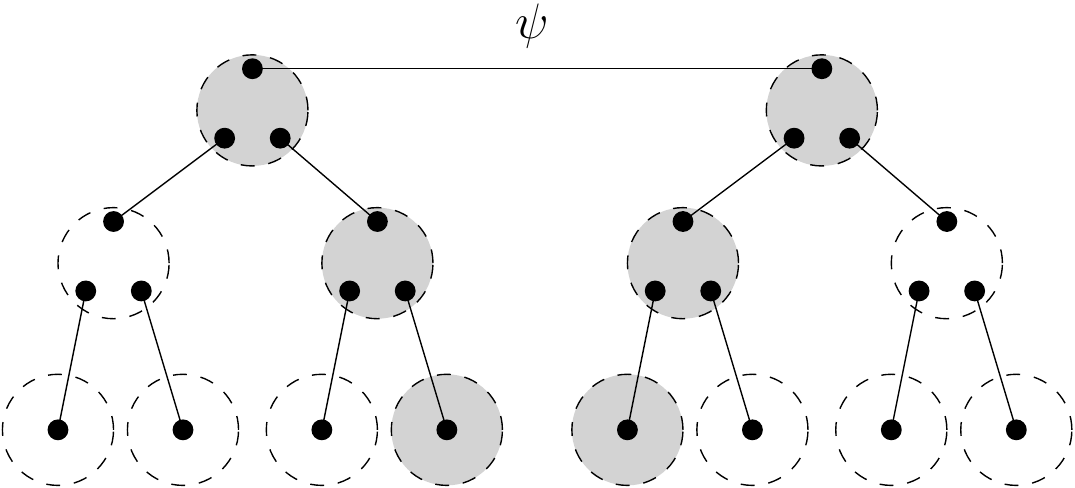}
        \caption{}
        \label{fig:step-one}
    \end{subfigure}
    \caption{
        Given a rooted graph $G$ and a multi-qudit state $\psi$ we consider a scenario
        with an agent for each node of $G$ and qudit of $\psi$. The agents corresponding
        to the roots have each qudit of $\psi$, and agents given by adjacent nodes
        have shared entanglement. (a) $\bullet$ denotes a qudit, and each line connecting
        two dots either the state $\psi$ or a maximally entangled state $\phi$.
        The set of qudits held by each agent is circled. (b) A path in $G$
        for each qudit of $\psi$ defines a sequence of qudits that can be used
        to implement measurements on $\psi$ by agents that are far away.
    }
    \label{fig:graph_scenario_example}
\end{figure}

In Section \ref{section:distributing-line} we showed that measurements on a single-qudit
state can be performed at long distances along a line, using only local entanglement
and constantly many rounds of local measurements. 

We now extend the scenario
from a line to a graph and from a single-qudit state to multiple qudits.
The purpose of this construction is to show
that measurements on a quantum state can be performed
in a distributed way, using a simulation with
only a constant number of rounds.
In Section \ref{section:distribute-2r}
we define a simulation for arbitrary measurements
using two rounds,
in Section \ref{section:distribute-1r}
we restrict to Weyl measurements and present a simulation
with a single round.

The information specifying the setup is conveniently represented as a \emph{rooted graph}.
We first define rooted graphs and some of their basic properties.
\begin{definition}
    A \emph{rooted graph} $G = (V, E, r)$ is an undirected
    and connected graph with a distinguished node $r$ called the \emph{root}.
    \begin{enumerate}
        \item  A \emph{path} is a non-repetitive list of nodes
                $v_1 = r, v_2, \dots, v_l$, starting with the root,
                such that $\{v_i, v_{i+1}\} \in E$ for all $i=1,\dots,l-1$.
               Write $\text{Paths}(G)$ for the set of paths in $G$.
        \item The \emph{neighbourhood} of a node $v \in v$
            is the set of nodes $N_G(v) = \{w \in V \mid \{v,w\} \in E\}$ 
            that are adjacent to $v$.
        \item The \emph{degree} of $G$ is the size of the largest neighbourhood:
                $\text{deg}(G) = \text{max}_{v \in V} \supp{N_G(v)}$.
        \item The \emph{radius} of $G$, denoted by $\text{rad}(G)$,
                is the least $K \geq 1$ such that every $v \in V$
                is reachable by a path of length at most $K$.
    \end{enumerate}
\end{definition}

Let $\psi$ be a multi-qudit state with qudits labelled by a set $I$
and $G = (V, E, r)$ a rooted graph. We consider
a scenario with agents labelled by $I \times V$ (Figure \ref{fig:graph_scenario_example}).
The agents share a single instance of the state $\psi$ and a number of maximally
entangled two-qudit states $\phi$. 
Each of the agents $(i,r) \in I \times V$
has qudit $i$ of $\psi$. Additionally, each pair of agents $(i,v), (i,w) \in I \times V$ such that $\{v,w\} \in E$,
has one qudit each out of a maximally entangled state.
Denote the total state by
\begin{align}
    \ket{\psi,G} &:= 
        \ket{\psi}  \otimes 
        \bigotimes_{i \in I, \{v,w\} \in E} \ket{\phi}_{(i,v,w), (i,w,v)}
\end{align}
The qudit held by agent $(i,v)$ of the maximally entangled 
state $\ket{\phi}_{(i,v,w), (i,w,v)}$ is labelled by $(i,v,w)$.
Agent $(i,r)$ therefore has the following set of qudits
\begin{align}
    \text{Qudits}(i,r) &:= \{i\} \cup \{(i,v,w) \mid w \in N_G(r)\}
\end{align}
where $N_G(r)$ is the neighbourhood of the root, and when $v \neq r$
agent $(i,v)$ has qudits
\begin{align}
    \text{Qudits}(i,v) &:= \{(i,v,w) \mid w \in N_G(v)\}
\end{align}

We now consider the following problem. Suppose that we want to perform
a measurement on each qudit of $\psi$. How can this be done
in such a way that 1) we minimise the probability
that any single agent knows the measurement setting,
2) we minimise the number of agents involved in the protocol.
The solution is to use the teleportation protocols from the previous section.
We first randomly select a path for each
qudit $i \in I$.
Let $r = v_1, \dots, v_j, \dots, v_l = v_i$ be a path in $G$.
We teleport qudit $i$ to agent $(i,v_i)$
using the sequence of qudits (Figure \ref{fig:two_round_examples})
\begin{align}
    i, (i, v_1, v_2), \dots, (i, v_j, v_{j-1}), (i, v_j, v_{j+1}),
    \dots, (i, v_l, v_{l-1})
\end{align}
Here $(i, v_j, v_{j+1}), (i, v_{j+1}, v_j)$ are maximally entangled,
$i, (i,v_1, v_2) \in \text{Qudits}(i,r)$,
and $(i, v_j, v_{j-1}), (i,v_j, v_{j+1}) \in \text{Qudits}(i,v_j)$
for each $j=2,\dots,l-1$. Hence the protocol uses only local measurements at each agent.

\newcommand{\pathdist}{u_\text{paths}}
\begin{definition}
    Given a rooted graph $G$ let 
    $\pathdist \in \mathcal{D}(\text{Paths}(G))$ be any
    distribution such that
    for any $v \in V$
    \begin{align}
        \pathdist(v_1, \dots, v_l) > 0 &\Rightarrow l \leq \text{rad}(G)\\
        \pathdist(v_1, \dots, v_l \text{ such that $v_l = v$})
        &= 1 / \supp{V}
    \end{align}
\end{definition}

If the paths are chosen independently for each qudit $i$
from the distribution $\pathdist$ then the probability
that any given player knows the measurement setting is
at most $1 / \supp{G}$, and the number of agents
involved in simulating the measurement on $i$
is at most $\text{rad}(G)$.

\subsection{Distributing measurements in two rounds} \label{section:distribute-2r}
\begin{figure}
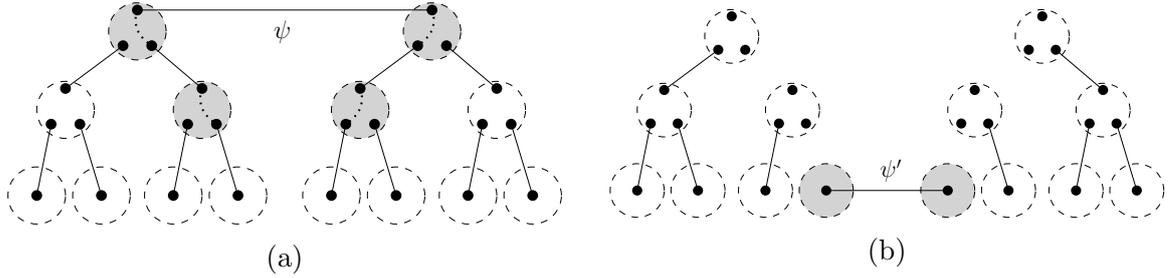

    \centering
    \begin{subfigure}{0.48\textwidth}
    \centering
        \includegraphics[width=1\textwidth]{resource-state-step-1_2.pdf}
        \caption{
        }
        \label{fig:measurement-setting}
    \end{subfigure}
    \hfill
    \begin{subfigure}{0.48\textwidth}
    \centering
        \includegraphics[width=1\textwidth]{resource-state-step-2_2.pdf}
        \caption{}
        \label{fig:step-two}
    \end{subfigure}
    \caption{Simulation protocol. (a) In the first round we perform Bell 
        basis measurements along each path. (b) The first step teleports
        each qudit up to a Weyl operator phase. In the second round we measure
        with the corrected measurement.
    }
    \label{fig:two_round_examples}
\end{figure}
We now suppose that $S = (I, X, Y)$
is a multipartite scenario,
$e$ is a quantum realised empirical model
with quantum realisation $(\psi,M)$ in qudit dimension $d$,
and $G = (V, E, r)$ is a rooted graph.

We first define a multipartite scenario $T(S, G, d)$.
The measurement sites of $T(S, G, d)$ are $I \times V$,
and at each measurement site $(i,v)  \in I \times V$
the measurement settings and outcomes are as follows.
When $v = r$ the measurement settings indicate a measurement
setting at measurement site $i$ in the scenario $S$,
or one of the neighbors of $r$ in the graph $G$.
Otherwise, when $v \neq r$, the
measurement settings indicate either one of the measurement settings
at measurement site $i$ as well as a neighbour of $v$
in $G$ and one of the Weyl measurement settings 
$\mathbb{Z}_d^2$, or two distinct neighbours of $v$ in $G$.
For the measurements involving one of the measurements from $S$
the outcomes are the outcomes given by $S$,
otherwise the outcomes are $\mathbb{Z}_d^2$.
\begin{definition}
    Let $S = (I, X, Y)$ be a multipartite scenario, $d \geq 2$ a dimension,
    and $G = (V, E, r)$ a rooted graph.
    $T(S,G,d)$ is the multipartite scenario with measurement sites $I \times V$
    and the following measurement settings and outcomes:
    \begin{align}
        \begin{tabular}{c | c | c}
             Measurement site & Measurement settings & Outcomes \\
             \hline
             $(i,r)$ & $x \in X_i$ & $Y_{i,x}$\\
                     & $w \in N_G(r)$ & $\mathbb{Z}_d^2$\\
            \hline
            $(i,v)$ & $(x,w,p) \in X_i \times N_G(v) \times \mathbb{Z}_d^2$ 
                & $Y_{i,x}$\\
                    & $w,w' \in N_G(v).\ w \neq w'$ & $\mathbb{Z}_d^2$\\
        \end{tabular}
    \end{align}
    for all $i \in I$ and $v \neq r$.
\end{definition}

Next, we define a quantum realised empirical model $e'$
for the scenario $T(S,G,d)$.
Recall that the quantum realisation $(\psi,M)$
consists of a state $\psi$ with qudits labelled by the measurement sites $I$,
and a single-qudit measurement $M_{i,x}$
on qudit $i$, for each measurement $(i,x)$ of the scenario $S$.

To define the empirical model $e'$ we interpret the measurement settings 
of the scenario $T(S,G,d)$
as quantum measurements on the state $\ket{\psi,G}$ in the following way.
At each measurement site $(i,r) \in I \times V$
the measurement setting $x$
is the measurement $M_{i,x}$
on qudit $i$, and the measurement setting
$w$ is a Bell basis measurement on qudits $i,(i,r,w)$.
Otherwise, when $v \neq r$,
the measurement setting
$(x,w,p)$ is the conjugated measurement $W(p) M_{i,x} W(p)^\dagger$
on qudit $(i,v,w)$,
and $(w,w')$ is a Bell basis measurement on qudits
$(i,v,w), (i,v,w')$.
\begin{definition}
    Let $G = (V, E, r)$ be a rooted graph,
    $S = (I, X, Y)$ a multipartite scenario,
    $\psi$ an $I$-qudit state,
    and $\pi(i,x)$ a single-qudit measurement with
    outcomes $Y_{i,x}$ for each $i \in I, x \in X_i$.
    $e_{G,\psi, \pi} : S_{G,d}$ is the empirical model
    realised by the following measurements on $\ket{G, \psi}$.
    \begin{align}
        \begin{tabular}{c | c | c}
             Measurement site & Measurement setting & Quantum measurement\\
             \hline
             $(i,r)$ & $x \in X_i$ & $\pi(i,x)$ on qudit $i$\\
                     & $w \in N_G(r)$ & Bell basis on qudits $i, (i,v,w)$\\
            \hline
            $(i,v)$ & $(x,w, p) \in X_i \times N_G(v) \times \mathbb{Z}_d^2$ 
                & $W(p) \pi(i,x) W(p)^\dagger$ on qudit $(i,v,w)$\\
                    & $w,w' \in N_G(v).\ w \neq w'$ & 
                        Bell basis on qudits $(i,v,w), (i,v,w')$\\
        \end{tabular}
    \end{align}
\end{definition}

Using the empirical model $e'$ we can simulate the empirical model $e$
in the following way. Suppose that $v$
is a path in $G$. Given any measurement $(i,x)$
on the scenario $S$ we subject the empirical model $e'$
to the following measurements.
In the case that $v = r$ is the path of length one
we perform measurement $x$ on measurement site $(i,r)$.
Otherwise, if $v = v_1, \dots, v_l$ then
we perform measurement $v_2$ on $(i,r)$,
and  measurement $(v_{j-1}, v_{j+1})$
on $(i,v_j)$ for each $j = 2,\dots,l-1$.
If the outcomes of these measurements are
$p_1, \dots, p_{l-1}$ then we perform
measurement $(v_{l-1}, x, \sum_{j=1}^{l-1} q_j)$
on $(i, v_l)$.
This defines a measurement protocol, which we denote by $C_{v, i, x}$.
\begin{align}
    C_{r, i, x} &:= (i,r) \mapsto x
\end{align}
and $C_{(v_1, \dots, v_l, i, x} = C_{(v_1, \dots, v_l, i, x}^1,
C_{(v_1, \dots, v_l, i, x}^2(s_1)$, where
\begin{align}
    C_{(v_1, \dots, v_l), i, x}^1 &:=
    \begin{cases}
        (i,r) &\mapsto v_2\\
        (i, v_j) &\mapsto (v_{j-1}, v_{j+1})
    \end{cases}\\
    C_{(v_1, \dots, v_l), i, x}^2(s_1) &:= 
        (i, v_l) \mapsto (v_{l-1}, x, \sum_{j=1}^{l-1} s_1(i, v_j))
\end{align}

After performing these measurements we return the outcome of the measurement
performed at measurement site $(i,v_l)$.
Write $g_{v,i,x}:\mathcal{E}_S(C_{v,i,x}) \to Y_{i,x}$
for the function from runs of $C_{v,i,x}$ to outcomes of $(i,x)$.
\begin{align}
    g_{r, i, x} &:= s \mapsto s(i,r)\\
    g_{(v_1, \dots, v_l), i, x} &= (s_1, s_2) \mapsto s_2(i, v_l)
\end{align}

\begin{definition}
    Let $G = (V, E, r)$ be a rooted graph,
    $S = (I, X, Y)$ a multipartite scenario, and $d \geq 2$.
    $s(S,G,d) : T(S,G,d) \to S$ is the simulation 
    \begin{align}
        \sum_{v \in \text{Paths}(G)^I} \big[ \prod_{i \in I} \pathdist(v_i) \big] \cdot t_v
    \end{align}
    where for each $v = \{v_i \in \text{Paths}(G)\}_{i \in I}$
    the deterministic simulation
    $t_v$ is defined by 
    $t_v := (\{C_{v_i, i, x} \}_{i \in I, x \in X_i}, \{g_{v_i, i, x}\}_{i \in I, x \in X_i})$.
\end{definition}

When one of the deterministic simulations
$t_v$ is applied to the empirical model $e'$
the effect is to perform the two-round teleportation
protocol along a choice of path $v_i$ for each measurement site $i \in I$.
It is therefore clear that $e'$ in fact simulates
the empirical model $e$.
\begin{lemma}
    Let $e_{\psi, M} : S$ be a quantum realised empirical mode,
    $G = (V, E ,r)$ a rooted graph.
    $s(S,G,d)$ simulates $e_{\psi,M}$ using $e_{G, \psi, M}$ as a resource.
    \begin{align}
        s(S,G,d)_*(e_{G, \psi, M}) = e_{\psi, M}
    \end{align}
\end{lemma}

\subsection{Distributing Weyl measurements in a single round} \label{section:distribute-1r}

Consider a multipartite scenario $(I, \mathbb{Z}_d^2, \mathbb{Z}_d)$,
a quantum realised empirical model $e$
along with a quantum realisation $(\psi, W)$
where $W(i,p)$ is the Weyl measurement $W(p)$
on qudit $i$, and $G = (V, E, r)$ a rooted graph.

\begin{figure}
    \centering
    \begin{subfigure}{0.48\textwidth}
    \centering
        \includegraphics[width=1\textwidth]{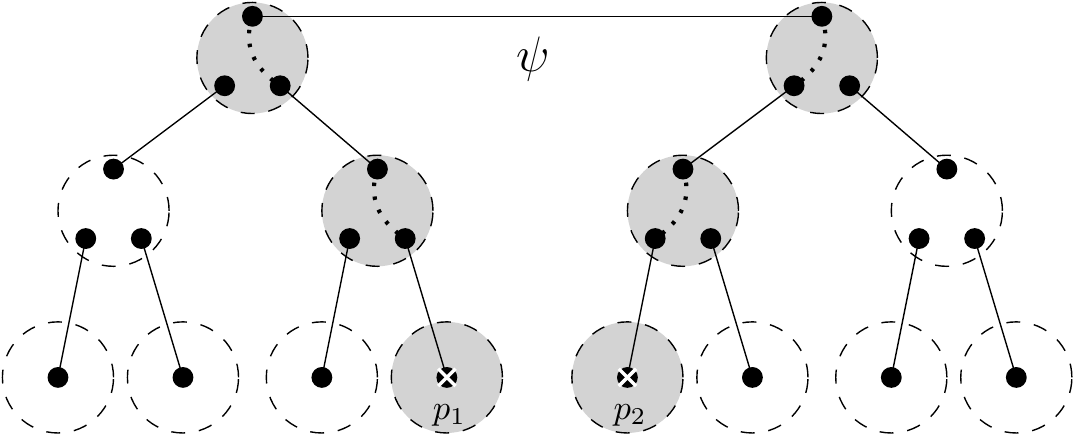}
    \end{subfigure}
    \caption{
    The state $\ket{\psi,G}$ for a two-qudit state $\psi$
    and a depth two binary tree. ``$\bullet$'' denotes a qudit, two qudits connected 
    by a line a maximally entangled state, and each group of qudits is surrounded by a dotted circle.
    }
    \label{fig:teleportation_examples_one_round}
\end{figure}

We first define a multipartite scenario $T(I,G,d)$
with measurement sites $I \times V$.
At each measurement site $(i,r) \in I \times V$
the measurement settings are either 
one of the Weyl measurement settings $\mathbb{Z}_d^2$
or a neighbour of $r$.
Otherwise, when $v \neq r$,
the measurement settings at $(i,v)$ are either a distinct pair of neighbours of 
$v$ or a Weyl measurement setting as well as a neighbour of $v$.
The outcomes are either $\mathbb{Z}_d^2$ or $\mathbb{Z}_d$.
\begin{definition}
    Let $I$ be a set, $d \geq 2$, and $G = (V, E, r)$ a rooted graph.
    $T(I,G,d)$ is the multipartite scenario with measurement sites $I \times V$
    and the following measurement settings and outcomes:
    \begin{align}
        \begin{tabular}{c| c | c}
             Measurement site &  Measurement settings & Outcomes \\
             \hline
             $(i,r)$ & $w \in N_G(r)$ & $\mathbb{Z}_d^2 $\\
                     & $p \in \mathbb{Z}_d^2$ & $\mathbb{Z}_d$\\
             \hline
             $(i,v)$ & $w \neq w' \in N_G(v)$ & $\mathbb{Z}_d^2$\\
                     & $(w,p) \in N_G(v) \times \mathbb{Z}_d^2$ & $\mathbb{Z}_d$
        \end{tabular}
        \label{table:measurement-settings}
    \end{align}
    for all $i \in I$ and $v \neq r$.
\end{definition}

We interpret the measurement settings of the scenario
$T(I,G,d)$ as quantum measurements on the state
$\ket{\psi,G}$.
At each measurement site $(i,r) \in I \times V$
measurement setting $p$ is a Weyl measurement $W(p)$
on qudit $i$, and measurement setting $w$
is a Bell basis measurement on qudits $i, (i,r,w)$.
Otherwise, when $v \neq r$, measurement
setting $(p,w)$ is a Weyl measurement on qudit
$(i,v,w)$ and $(w,w')$ a Bell basis measurement
on qudits $(i,v,w), (i,v,w')$.
\begin{definition}
    Let $I$ be a set, $d \geq 2$, $\psi$ an $I$-qudit
    state, and $G = (V, E, r)$ a rooted graph.
    The empirical model $e_{G,\psi}:S_{G,I,d}$
    is the empirical model realised by the state $\ket{G,\psi}$
    and measurements:
    \begin{align}
        \begin{tabular}{c|c|c}
             Measurement site & Measurement setting &  Quantum measurement \\
             \hline
             $(i,r)$ & $p \in \mathbb{Z}_d^2$ & $W(p)$ on qudit $i$\\
                     & $w \in N_G(r)$ & Bell basis on qudits $i, (i,r,w)$\\
            \hline
             $(i,v)$ & $(w,p) \in N_G \times \mathbb{Z}_d^2$ 
                        & $W(p)$ on qudit $(i,v,w)$\\
                     & $w \neq w' \in N_G(v)$ & 
                        Bell basis on qudits $(i,v,w), (i,v,w')$
        \end{tabular}
    \end{align}
    for all $i \in I$ and $v \neq r$.
\end{definition}

We define a simulation from $T(I, G,d)$ to $(I, \mathbb{Z}_d^2, \mathbb{Z}_2)$
in the following way (Figure \ref{fig:teleportation_examples_one_round}). 
Given a path $v$ in $G$ an a measurement
setting $(i,p)$ for the scenario $(I, \mathbb{Z}_d^2, \mathbb{Z}_2)$
perform the measurement $p$ on measurement site $(i,r)$
in the case that $v = r$,
otherwise if $v = v_1, \dots, v_l$
perform measurement $v_2$ on $(i,r)$,
$(v_{j-1}, v_{j+1})$ on $(i, v_j)$ where $j=2,\dots,l-1$, 
and measurement setting $(v_{l-1}, p)$ on $(i, v_l)$. 
If the outcomes of these measurements are $p_1, \dots, p_{l-1}, q$
then we return the outcome $q - [p, p_1 + \dots p_{l-1}]$.

Write $f_{v, i, p}$ for the measurement context,
and $g_{v,i,p}:\mathcal{E}_{T(I,G,d)}(f_{v,i,p}) \to \mathbb{Z}_d$
for the outcome map:
\begin{align}
    f_{r,i,p} &:= (i,r) \mapsto p\\
    g_{r,i,p} &:= s \mapsto s(i,r)\\
    f_{(v_1, \dots, v_l, i,p} &:= \begin{cases}
        (i,v_1) &\mapsto v_2\\
        (i, v_j) &\mapsto (v_{j-1}, v_{j+1})\\
        (i, v_l) &\mapsto (v_{l-1}, p)
    \end{cases}\\
    g_{v_1, \dots, v_l, i,p} &:= s \mapsto s(i, v_l) + [p, \sum_{j=1}^{l-1}s(i, v_j)]
\end{align}

\begin{definition}
    Let $I$ be a set, $d \geq 2$, and $G = (V, E, r)$ 
    a rooted graph.
    $s(I,G,d): T(I, G, d) \to (I, \mathbb{Z}_d^2, \mathbb{Z}_d)$
    is the simulation 
    \begin{align}
        \sum_{v \in \text{Paths}(G)^I} \big[ \prod_{i \in I} \pathdist(v_i) \big] \cdot t_v
    \end{align}
    where $t_v$ is the deterministic simulation
    $(\{f(v_i)_{i,p}\}_{i \in I, p \in \mathbb{Z}_d^2}, \{g(v_i)_{i, p}\}_{i \in I, p \in \mathbb{Z}_d^2})$.
\end{definition}

When the simulation is applied to the empirical model $e$
the effect is to perform the single-round teleportation protocol on a line.
It is therefore clear that the pushforward of
$e_{G,\psi,M}$ is $e$.
\begin{lemma}
    Let $e_\psi:(I, \{\mathbb{Z}_d^2\}, \{\mathbb{Z}_d\})$
    be a Pauli measurement model
    and $G$ a rooted graph.
    $s_G: S_{G,I,d} \to (I, \mathbb{Z}_d^2, \mathbb{Z}_d)$
    simulates $e_\psi$ using $e_{\psi,G}$ as a resource.
    \begin{equation}
        (s_G)_*(e_{\psi,G}) = e_\psi
    \end{equation}
\end{lemma}

\section{Distributing non-local games on graphs} \label{section:distributing-games}
In Section \ref{section:distributing-graph} we showed that
measurements on a multi-qudit state can be simulated in a distributed
way, using a graph as a template. We now use this construction
to define distributed versions of non-local games and show
that they are solved by quantum circuits of low depth and fan-in.
We present two versions of this. In Section 
\ref{section:arbitrary-nlg} we use the two-round teleportation
protocol, and we work with general non-local games.
In Section \ref{section:weyl-nlg} we use the single-round protocol
and we restrict attention to non-local games using Weyl measurements.

\subsection{Two-round distributed non-local games} \label{section:arbitrary-nlg}

Let $S = (I, X, Y)$ be a multipartite scenario,
$(e, \Phi)$ a non-local game along with
a quantum realisation $(\psi,M)$ of $e$
in qudit dimension $d \geq 2$, and $G = (V, E, r)$ 
a rooted graph.

In Section \ref{section:distribute-2r} we defined
a measurement scenario $T(S,G,d)$, a quantum realised empirical model $e'$,
and a simulation $s(S,G,d):T(S,G,d) \to S$ such
that ${s(S,G,d)}_*(e') = e$. We now consider the pullback
of the cooperative game $\Phi$ across the simulation $s(S,G,d)$.
Because $e'$ simulates $e$ we have that the success probability
of $e'$ on the pullback problem is equal to the success
probability of $e$. $e'$ therefore violates
the classical bound for the non-local game $(e,\Phi)$.
\begin{align}
    p_S(e', {s(S,G,d)}^*(\Phi)) = p_S(e, \Phi) > \gamma
\end{align}
where $\gamma$ is the classical bound.

We can implement $e'$ as a two-round quantum circuit strategy (Figure \ref{fig:two_round_circuit_strategy}).
We first have to prepare $\ket{\psi,G}$.
This can be done with a single unitary gate of fan-in $\supp{I}$
and a number of unitary two-qudit gates.
We then have to implement measurements. The measurements
corresponding to each measurement site $(i,v)$
act on $\text{Qudits}(i,v)$. The fan-in of these gates, 
therefore, depends only on the \emph{degree} of $G$.

\begin{figure}
    \centering
    \begin{subfigure}{\textwidth}
        \includegraphics[width=\textwidth]{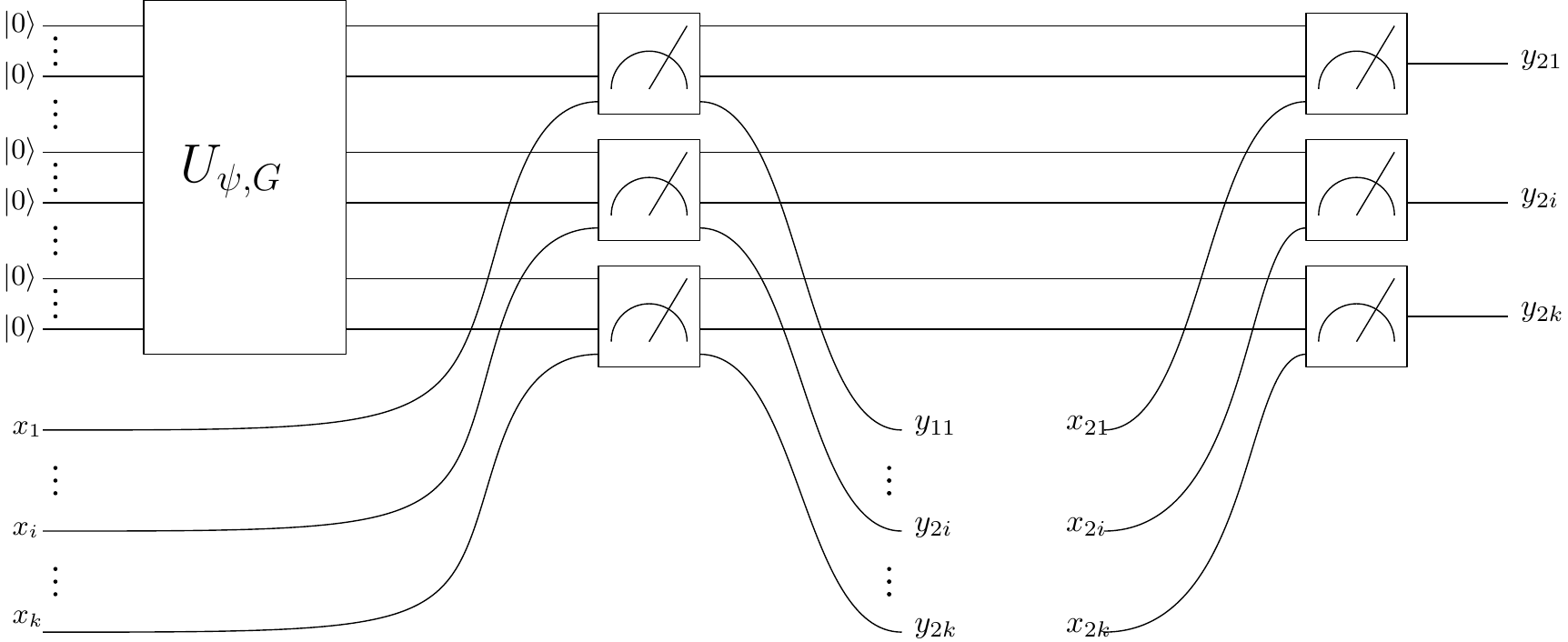}
        \caption{}
        \label{fig:my-circ-2}
    \end{subfigure}
    \begin{subfigure}{\textwidth}
        \centering
        \begin{align*}
            U_{\psi,G} &= 
                U_{\psi} \otimes 
                    \bigotimes_{i \in I, \{v,w\} \in E} U_{(i, v, w), (i,w,v)}\\
                    U_\psi \ket{0 \dots 0} &= \ket{\psi}\\
                    U_{(i,v,w), (i,w,v)}\ket{00} &= \ket{\phi}
        \end{align*}
        \caption{}
        \label{fig:my-circ-state}
    \end{subfigure}
    \begin{subfigure}{\textwidth}
        \centering
        \begin{tabular}{c | c|c}
            Measurement site &  Input value & Measurement setting \\
             \hline
             $(i,r)$ & $w \in N_G(r)$ & Bell basis on qudits $i, (i,r,w)$\\
                     & $x \in X_i$ & $M_{i,x}$ on qudit $i$\\
                     & $\bullet$ & Identity\\
            \hline
            $(i,v)$ & $w, w' \in N_G(v).\ w \neq w'$ & Bell basis on qudits $(i,v,w), (i,v,w')$\\
                    & $(x, w, p) \in X_i \times N_G(v) \times \mathbb{Z}_d^2$ 
                        & $W(p) M_{i,x} W(-p)$ on qudit $(i, v, w)$\\
                        & $\bullet$ & Identity
        \end{tabular}
        \caption{}
    \end{subfigure}
    \caption{
        The two-round quantum circuit strategy $Q_{G,\psi,M}$
        uses a multi-qudit register initially set to the computational basis state. The state $\ket{\psi,G}$ is prepared by 
        a single $\supp{I}$-qudit gate $U_\psi$ and a two-qudit gate
        $U_{(i,v,w), (i,w,v)}$ for each $i \in I, \{v,w\} \in E$.
        The first round of inputs $x_{i,v}$
        is used to control a non-destructive measurement
        gate $M_{i,v}$ with measurement settings
        given by (b),
        and the second round of inputs $x_{(i,v)}'$
        controls a destructive measurement gate $M_{i,v}'$
        also with measurement settings from (b).
    }
    \label{fig:two_round_circuit_strategy}
\end{figure}

\begin{lemma}
    Let $S = (I, X, Y)$ be a multipartite scenario and
    $(e, \Phi)$ a non-local game with classical bound $\gamma$.
    Suppose that $e$ has a quantum realisation in qudit dimension $d \geq 2$.
    There exists a two-round cooperative game $\Phi'$ 
    with classical bound $\gamma$ and quantum circuit strategy $Q$
    such that
    \begin{enumerate}
        \item The success probability of $Q$ exceeds $\gamma$: $p_S(Q, \Phi') > \gamma$.
        \item The depth and maximal fan-in of $Q$ depends only on the size of $I$ and
            the degree of $G$.
    \end{enumerate}
\end{lemma}

\begin{figure}
    \centering
    \begin{subfigure}{0.48\textwidth}
    \centering
        \includegraphics[width=1\textwidth]{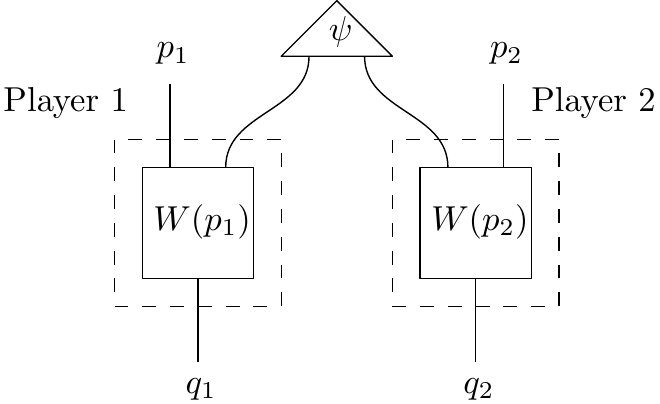}
        \caption{
        }
        \label{fig:measurement-setting}
    \end{subfigure}
    \hfill
    \begin{subfigure}{0.48\textwidth}
    \centering
        \includegraphics[width=1\textwidth]{resource-state-step-1.pdf}
        \caption{}
        \label{fig:step-two}
    \end{subfigure}
    \caption{(a) 
        A generic two-player Weyl measurement game. Alice and Bob
        share a two-qudit state $\psi$. Verifier randomly selects
        $q_1, q_2 \in \mathbb{Z}_2^2$ and an accepting condition
        $A \subset \mathbb{Z}_d^2$ according to a probability distribution
        $d(q_1,q_2,A)$. Alice and Bob measure the Weyl 
        measurements $W(p_1), W(p_2)$ respectively. 
        Their success probability is the likelihood
        that $(q_1,q_2) \in A$.
        (b) Graph version of the game (a) played on a tree.
        randomly selects paths $(r, v_{11}, \dots, v_{1k_1}), (r, v_{21}, \dots, v_{2k_2})$ according to a path distribution $d_\text{paths}$
        and an instance $(p_1, p_2, A)$ of (a) with probability $d(p_1, p_2, A)$.
        Verifier sends players $(1, r), (1, v_{11}), \dots, (1, v_{1k_1})$
        They win if $(q_1' + [p_1, p_1'], q_2' + [p_2, p_2']) \in A$,
        where $p_i' = p_{i1}' + \dots + p_{i2}'$.
    }
    \label{fig:teleportation}
\end{figure}

\subsection{Single-round distributed Weyl-measurement games} \label{section:weyl-nlg}

Let $(e,\Phi)$ be a non-local game for a multipartite scenario
on the form $(I, \mathbb{Z}_d^2, \mathbb{Z}_d)$,
such that $e$ has a quantum realisation on an $I$-qudit
state $\psi$, where the measurement $(i,p)$
is the Weyl measurement $W(p)$ on qudit $i$.
Let $G = (V, E, r)$ be a rooted graph.

In Section \ref{section:distribute-1r} we defined a quantum realised
empirical model $e'$ and a single-round simulation $s(I,G,d)$
such that ${s(I,G,d)}_*(e') = e$.
Taking the pullback $s(I,G,d)^*(\Phi)$ we then
have
\begin{align}
    p_S(e', s(I,G,d)^*(\Phi)) > \gamma
\end{align}
where $\gamma$ is the classical bound of $(e,\Phi$.
We can implement $e'$ as the quantum circuit in Figure (\ref{fig:circuit-definition}).

\begin{lemma}
    Let $(e, \Phi)$ be a Weyl measurement game with
    classical bound $\gamma$ and $G$ a rooted graph.
    Consider the pullback of $\Phi$ across the single-round simulation.
    This game has a quantum circuit strategy $Q$
    such that 
    that
    \begin{enumerate}
        \item The success probability of $Q$ at the pullback of $\Phi$
            is the success probability of $e$ at $\Phi$, which exceeds $\gamma$:
                $p_S(Q, s(I,G,d)^*(\Phi)) > \gamma$.
        \item  The depth and maximal fan-in of $Q$ is only
            dependent on the size of $I$ and the degree of $G$.
    \end{enumerate}
\end{lemma}

We can describe the pullback game more directly as follows.
Recall $\Phi$ is defined as a convex combination 
$\sum_{U, A} \Phi_{U,A} \cdot (U,A)$
where $U$ is a joint measurement for the scenario
$(I, \mathbb{Z}_d^2, \mathbb{Z}_d)$ and 
$A \subset \mathcal{E}_{(I, \mathbb{Z}_d^2, \mathbb{Z}_d)}(U)$
is an accepting condition.

In the pullback game Verifier randomly selects $U,A$ with
probability $\Phi_{U,A}$. 
For each joint input $(i,p_i) \in U$ Verifier then selects
a path $v_i = (v_{i1}, \dots, v_{il_i})$ in $G$
with probability $\pathdist(v_i)$.
If $v_i$ is the trivial path then Verifier gives
input $p_i$ to $(i,r)$, otherwise
if $l > 1$, Verifier gives input $v_{i2}$
to $(i,r)$, input $(v_{i(j-1)}, v_{i(j+1)})$ to $(i,v_j)$,
and finally input $(v_{i(l_i-1)}, p_i)$
to and $(i,v_{l_i})$.
The total joint input is then
\begin{align}
    U_v := \begin{cases}
        (i, r) \mapsto p_i, &\text{if $(i,p_i) \in U$ and $v_i=r$}\\
        (i, r) \mapsto v_2, &\text{if $(i,p_i) \in U$ and $l_i > 1$}\\ (i, v_{ij} \mapsto (v_{i(j-1)}, v_{i(j+1)}), &\text{if $(i,p_i) \in U$ and $j=2, \dots, l_i-1$}\\ (i, v_{il_i}) \mapsto (v_{i(l_i-1)}, p_i), &\text{if $(i, p_i) \in U$ and $l_i > 2$} \end{cases}
\end{align}
The players then respond with a joint output
$s \in \mathcal{E}_{S_{G,I,d}}(U_{v,p})$.
The output is accepted if
\begin{align}
    A_{v,U}(s) :\iff \big( (i, p_i) \mapsto s(i, v_{il_i}) + [
    p_i, \sum_{j=1}^{l_i-1} s(i, v_{ij})] \big)_{(i,p_i) \in U} \in A
\end{align}
As a convex combination the pullback is then the game
\begin{align}
    \sum_{v, U, A} \big[ \prod_{i \in I} \pathdist(v_i) \big] \Phi_{U,A} \cdot (U_v, A_{v,U})
\end{align}

\begin{figure}
    \centering
    \begin{subfigure}{1\textwidth}
    \centering
        \includegraphics[width=0.48\textwidth]{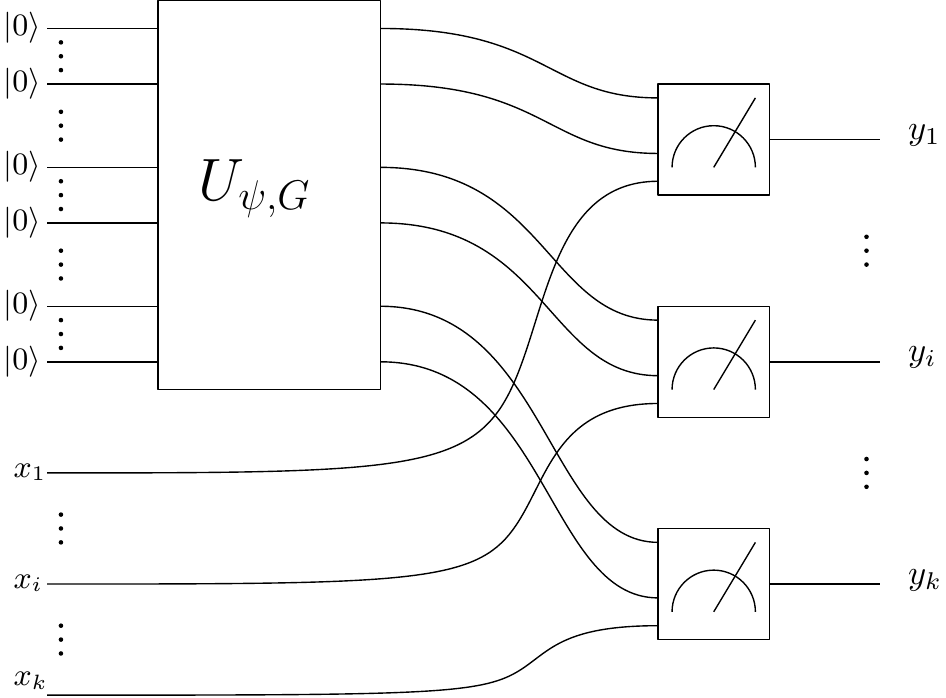}
        \caption{}
        \label{fig:step-two}
    \end{subfigure}
    \hfill
    \begin{subfigure}{\textwidth}
        \centering
        \begin{align*}
            U_{\psi,G} &= 
                U_{\psi} \otimes 
                    \bigotimes_{i \in I, \{v,w\} \in E} U_{(i, v, w), (i,w,v)}\\
                    U_\psi \ket{0 \dots 0} &= \ket{\psi}\\
                    U_{(i,v,w), (i,w,v)}\ket{00} &= \ket{\phi}
        \end{align*}
        \caption{}
        \label{fig:my-circ-state}
    \end{subfigure}
    \begin{subfigure}{1\textwidth}
    \centering
        \begin{tabular}{c | c|c}
             Input wire & Value & Measurement setting \\
             \hline
             $(i,r)$ & $p \in \mathbb{Z}_d^2$ & $W(p)$ on qudit $i$\\
                     & $w \in N_G(r)$ & Bell basis on qudits $i$ $(i,r,w)$\\
                     & $\bullet$ & Identity measurement\\
            \hline
             $(i,v)$ & $(w, p) \in N_G(v) \mathbb{Z}_d^2$ & $W(p)$ on qudit $(i, v, w)$\\
                     & $(w,w') \in N_G(v)^2.\ w \neq w'$ 
                     & Bell basis on qudits $(i,v,w)$ $(i,v,w')$\\
                     & $\bullet$ & Identity measurement\\
        \end{tabular}
        \caption{}
    \end{subfigure}
    \caption{(a) Circuit strategy $Q_{\psi,G}$, where
    $\psi$ is an $n$-qudit state and $G = (V, E, r)$ is a rooted
    graph.
    to entangle each qudit of $\psi$
    with a qudit that is far away in the circuit and measures this qudit in the Weyl basis.
    (b) The effect of the Bell basis measurements is to teleport $\psi$
    up to a random phase.
    }
    \label{fig:circuit-definition}
\end{figure}

\section{Separating quantum and classical circuits of low depth and fan-in}
\label{section:separating-quantum-and-classical}

In Section \ref{section:distributing-graph} we defined two simulations
$s(S,G,d)$ and $s(I,G,d)$ with respectively one and two rounds, that perform
measurements on a quantum state in a distributed way.
The simulations where then used in Section \ref{section:distributing-games}
to define distributed versions of non-local games, such that
the quantum strategies can be recast as circuits
of small depth and fan-in.
The purpose of this section is to bound the success probability
of classical circuits on these problems.
Recall the resource inequality
\begin{align}
    p_S(e', \Phi) \leq \gamma + \cf{e'}
\end{align}
relating the success probability of an arbitrary empirical model
$e'$ to the classical bound $\gamma$ for a non-local game and its contextual
fraction $\cf{e'}$.
We observe that the derivation of this bound does not rely on the no-disturbance
condition. In Section \ref{section:circuits-as-ontological-models}
we explained that circuit strategies give rise to empirical models
that don't satisfy the no-disturbance condition, we called these
objects ``behaviours'', and explained that constructions
like simulations and the contextual fraction can be generalised.
Because the bound $\gamma + \cf{e'}$ does not rely on no-disturbance
we have a bound
\begin{align}
    p_S(B, \Phi) \leq \gamma + \cf{B}
    \label{eq:resourcebound-circui}
\end{align}
for any behaviour $B$.

We want to bound the contextual fraction of a classical circuit
$C$ on the pullback $s^*(\Phi)$ of a game $\Phi$ across some simulation
$s$. To do this we bound the contextual fraction of the pushforward
$s_*(\tilde{C})$ and rely on the inequality \ref{eq:resourcebound-circui}.

Although we use a different terminology this is precisely
what Bravyi, Gosset, and K\"{o}nig did. We will work
at a general level, first stating the result
as a general property of simulations, and then
restrict to the two simulations.

The idea is to consider a general simulation $s:S \to T$. Recall that $s$ is defined as a probability
distribution over deterministic simulations $t:S \to T$. In Section \ref{section:lightcones}
we consider the case of a deterministic simulation. We derive a combinatorial condition involving the 
lightcones of the circuit $C$ that ensures that the pushforward $t_*(\tilde{C})$ is non-contextual.
In Section \ref{section:bounding-cf} we first present a lemma due to BGK,
and we restate this as a bound on the probability that this condition
holds, when $t$ is selected randomly from a simulation $s$.
We finally apply this to the simulations $s(G,I,d)$ and $s(S,G,d)$.

\subsection{Lightcones of simulations} \label{section:lightcones}

The simulations $s(S,G,d)$ and $s(I, G, d)$ are defined as probability distributions
over deterministic simulations $t_v$
corresponding to each choice of paths $v \in  \text{Paths}(G)^I$.
Consider these simulations for a fixed choice of paths $v$.
The measurements performed in these simulations
are independent of the input for all but a small number of measurement sites,
and the outcome we return only depends on the outcomes
of the measurements on a small number of measurement sites.
In the two-round case, given a measurement $(i,x)$ 
for the scenario $S$ we first perform a measurement 
$C_{v_i, i, x}^1$ that is independent of $x$
and in the second round we perform a measurement $C_{v_i, i, x}$ that is only defined
on $(i, v_{il_i})$.
The final outcome only depends on the outcome at $(i,v_{il_i})$.
In the single-round case, we perform a single measurement
where only the setting at $(i,v_{il_i})$ depends on $x$.
And the outcome returned depends only on the subset of measurement sites
$\{(i, v_{ij})\}$.

For any deterministic simulation $t:S \to T$ we can identify
the unique minimal subset of measurement sites such that the measurement
setting in round $k$ depends on the input, 
and the outcome uses the measurement outcome.
For each round $k \leq n$ and $j \in J$ we define the input and output
dependencies of the simulation as follows.
\begin{definition}
    Let $S = (I, X, Y), T = (J, Z, W)$
    be two multipartite scenarios and $t = 
    (\{C_{j,z}\}_{j \in J, z \in Z_j}, \{g_{j,z}\}_{j \in J, z \in Z_j}):S \to T$ an $n$-round deterministic 
    simulation, where $C_{j,z}$ is the measurement protocol
    \begin{align}
        C_{i,z} = C_{i,z}^1, \dots, C_{i,z}^k(s_1, \dots, s_{k-1}), \dots, C_{i,z}^n(s_1, \dots, s_{n-1})
    \end{align}
    and $g_{j,z}:\mathcal{E}_S(C_{j,z}) \to W_{j,z}$.
    For any measurement site $j \in J$ and round $k = 1,\dots, n$
    let $\text{In}_k(t)(j), \text{Out}_k(t)(j) \subset I$ as follows
    \begin{itemize}
        \item $\text{In}_k(t)(j)$ is the least $I' \subset I$
            such that for each $z \in Z_j$ and run $(s_1, \dots, s_n) \in \mathcal{E}_S(C_{j,z})$
            the joint measurement $C_{j,z}^k(s_1, \dots, s_{k-1})$
            can be written as a union
            \begin{align}
                C_{j,z}^k(s_1, \dots, s_{k-1}) = U(s_1, \dots, s_{k-1}) 
                    \cup U_z(s_1, \dots, s_{k-1})
            \end{align}
            where $U(s_1, \dots, s_{k-1})$ is independent of $z$
            and $U_z(s_1, \dots, s_{k-1})$ is a joint measurement for the measurement
            sites $I'$.
        \item $\text{Out}_k(t)(j) \subset I$ is the least $I' \subset I$ such
        that for each run $(s_1, \dots, s_n) \in \mathcal{E}_S(C_{j,z})$
        the outcome $g_{j,z}(s_1, \dots, s_n) \in W_{j,z}$
        is independent of the value of $s_k$ on measurement sites $I \setminus I'$:
        \begin{align}
            g_{j,z}(s_1, \dots, s_n) = g_j(s_1, \dots, \res{s_k}{I'}, \dots, s_n)(x)
        \end{align}
        where $g_j$ is some function.
    \end{itemize}
\end{definition}

For example, for the two-round simulation $t := t(S, G, d, v)$ we have
\begin{align}
    \text{In}_1(t)(i) &= \emptyset\\
    \text{In}_2(t)(i) &= \{(i, v_{il_i})\}\\
    \text{Out}_2(t)(i) &= \{(i, v_{il_i})\}
\end{align}
and for the single-round simulation $t := t(I, G, d, v)$
\begin{align}
    \text{In}_1(t)(i) &= \{(i, v_{il_i})\}\\
    \text{Out}_1(t)(i) &= \{(i, v_{ij})\}_{j=1}^{l_i}
\end{align}

If $t:S \to T$ is an $n$-round simulation such that $\text{In}_k(t)(j)$ is empty for all $k < n$
and $C$ is a classical strategy such that $t_*(\tilde{C})$ is contextual.
Then there must be some $j \neq j'$ and $i \in \text{In}_n(t)(j), i' \in \text{Out}_n(t)(j')$
such that there is path through $C$ from $\text{in}_{i n}$ to $\text{out}_{i' n}$.
Let the relation $\text{In}_n(t) ; \lcb{C} ; \text{Out}_n(t) \subset J \times J$
be defined by
\begin{align}
    \exists i,i' \in I:\ i \in \text{In}_n(t)(j) \land \text{in}_{n i} \in \lcb{C}(\text{out}_{n i'})
        \land i' \in \text{Out}_n(j')
\end{align}
for all $(j,j') \in J \times J$ (Figure \ref{fig:lightcone_composition}). 
In other words, if 
$\text{In}_n(t) ; \lcb{C} ; \text{Out}_n(t) \subset \text{id}_J$,
where $\text{id}_J \subset J \times J$ is the identity relation,
then $t_*(\tilde{C})$ is non-contextual.

\begin{figure}
    \centering
    \begin{subfigure}{0.3\textwidth}
        \centering
        \includegraphics[width=0.7\textwidth]{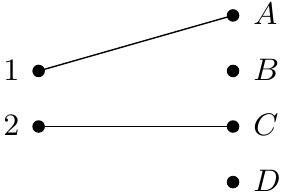}
        \caption{$\text{In}_1(t)$}
    \end{subfigure}
    \hfill
    \begin{subfigure}{0.3\textwidth}
        \centering
        \includegraphics[width=0.9\textwidth]{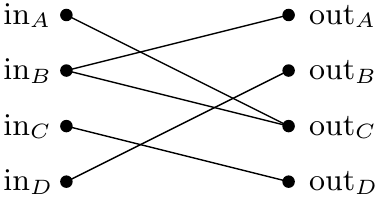}
        \caption{$\lcb{C}$}
    \end{subfigure}
    \hfill
    \begin{subfigure}{0.3\textwidth}
        \centering
        \includegraphics[width=0.7\textwidth]{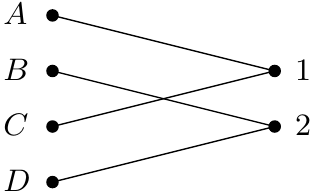}
        \caption{$\text{Out}_1(t)$}
    \end{subfigure}
    \begin{subfigure}{0.3\textwidth}
        \centering
        \includegraphics[width=0.7\textwidth]{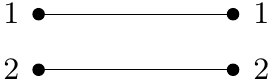}
        \caption{$\text{In}_1(t) ; \lcb{C} ; \text{Out}_1(t)$}
    \end{subfigure}
    \caption{ Consider a simulation $t$ from a scenario with measurement sites $\{1,2\}$
    to a scenario with measurement sites $\{A,B,C,D\}$ with the input
    dependencies in (a) and the output dependencies in (c)
    and a classical circuit $C$ with lightcones in (b).
    }
    \label{fig:lightcone_composition}
\end{figure}

\begin{lemma} \label{lemma:lightcones}
    Let $S$ and $T = (J, Z, W)$ be two multipartite scenarios, $t:S \to T$ an $n$-round deterministic simulation,
    and $C$ an $n$-round classical circuit strategy for $S$. If for all $k < n$ and $j \in J$ we have
    $\text{In}_k(t)(j) = \emptyset$, and $\text{In}_n(t) ; \lcb{C} ; \text{Out}_n(t) \subset \text{id}_J$
    then $t_*(\tilde{C})$ is non-contextual.
\end{lemma}
\begin{proof}
    Let $t : S \to T$ be an $n$-round deterministic simulation
    and $C$ an $n$-round classical circuit strategy.
    By definition $t_*(\tilde{C})$ is defined for any context $U$
    of $T$ by
    \begin{align}
        t_*(\tilde{C})(U) = \sum_{(s_1, \dots, s_n) \in \mathcal{E}_S(f(U))} \tilde{C}(f(U))(s_1, \dots, s_n) \cdot g_U(s_1, \dots, s_n)
    \end{align}
    where $f(U) \in \text{MP}_n(S)$ and $g_U:\mathcal{E}_S(f(U)) \to \mathcal{E}_T(U)$.
    By assumption the first $n-1$ measurements are independent of $U$, so that
    $f(U)$ is of the form:
    \begin{align}
        f(U) = V_1, V_2(s_1), \dots, V_{n-1}(s_1, \dots, s_{n-2}), V_n(U)(s_1, \dots, s_{n-1})
    \end{align}
    
    Because $C$ is a classical circuit we can write the probability distribution
    $\tilde{C}(F(U))$ as a convex combination of hidden variables
    $\lambda$, such that
    \begin{align}
        \tilde{C}(f(U)) = \sum_{\lambda} 
            p_\lambda \cdot (\lambda_1, \dots, \lambda_{n-1}, \lambda_n(V_n(U)(\lambda_1, \dots, \lambda_{n-1})))
    \end{align}
    and the value of $\lambda_n(U)(\lambda_1, \dots, \lambda_{n-1})$ at any measurement
    site $i \in I$ depends only on the measurement setting of $V_n(U)$ on the subset
    of measurement sites
    \begin{align}
        \overleftarrow{i} = \{i' \in I \mid \text{in}_{i'} \in \lcb{C}(\text{out}_i) \}
    \end{align}
    
    Plugging the hidden variable expression for $\tilde{C}(f(U))$ into the first equation
    we get a hidden variable expression for the pushforward $t_*(\tilde{C})(U)$
    \begin{align}
        t_*(\tilde{C})(U) = \sum_{\lambda} p_\lambda \cdot g_U(\lambda_1, \dots, \lambda_{n-1}, \lambda_n(V_n(U)(\lambda_1, \dots, \lambda_{n-1})))
    \end{align}
    
    To see that $t_*(\tilde{C})$ it is sufficient to show that 
    the value of
    \begin{align}
    g_U(\lambda_1, \dots, \lambda_{n-1}, \lambda_n(V_n(U)(\lambda_1, \dots, \lambda_{n-1})))
    \end{align}
    on a measurement $(j,z)$ is independent of $U$.
    
    Suppose therefore that $(j,z)$ is any measurement for $T$
    and for any $U$ consider the outcome assigned to $(j,z)$. That is, the value
    \begin{align}
        g_U(\lambda_1, \dots, \lambda_{n-1}, \lambda_n(V_n(U)(\lambda_1, \dots, \lambda_{n-1})))(j,z)
    \end{align}
    First we have that the outcome of $(j,z)$ only depends on the value of $\lambda_n(V_n(U)(\lambda_1, \dots, \lambda_{n-1})))$
    on the subset of measurement sites $\text{Out}_n(t)(j) \subset I$.
    Furthermore, the value of $\lambda_n(V_n(U)(\lambda_1, \dots, \lambda_{n-1})))$
    on measurement sites $\text{Out}_n(t)(j)$ depends only on the value
    of $V_n(U)(\lambda_1, \dots, \lambda_{n-1}))$ on
    \begin{align}
        \overleftarrow{\text{Out}_n(t)(j)} := \bigcup_{i \in \text{Out}_n(t)(j)} \overleftarrow{i}
    \end{align}
    
    Hence the outcome of $(j,z)$ depends only on the value of
    $V_n(U)(\lambda_1, \dots, \lambda_{n-1})$ on each $i'$ such that there exists
    an $i$ such that $\text{in}_{i'n} \in \lcb{C}(\text{out}_{in}$ and $i \in \text{Out}_n(t)(j)$.
    Finally, we have that the value of $V_n(U)(\lambda_1, \dots, \lambda_{n-1})$ on $i$
    only depends on the value of $U$ on $j$ such that $i \in \text{In}_n(t)(j)$.
    Therefore, if the condition holds then it is non-contextual.
\end{proof}

\subsection{Classical bound} \label{section:bounding-cf}

Suppose that $s:S \to T$ is a probabilistic simulation from a multipartite scenario $S$ to another
multipartite scenario $T$ and $C$ a classical strategy for $S$.
Suppose that we randomly select a deterministic simulation $t:S \to T$ with probability
given by $s$. Lemma \ref{lemma:lightcones} gives a condition ensuring
that the pushforward $t_*(\tilde{C})$ is non-contextual, involving the relation
$\text{In}_n(t) ; \lcb{C} ; \text{Out}_n(t)$.
We now consider the probability that this condition is satisfied
when $t$ is chosen randomly from the simulation $s$.

The following lemma from BGK shows that $C$ has low depth and fan-in, $\text{In}_n(t)(j)$ is small,
and $\text{Out}_n(t)(j)$ is distributed in a uniform way, then the probability
that the condition holds is high.
\begin{lemma}(BGK) \label{lemma:bgk}
    Let $C$ be a circuit with inputs $\{\text{in}_i\}_{i \in I}$
    and outputs $\{\text{out}_i\}_{i \in I}$
    of depth $D$ and fan-in at most $K$.
    Suppose that we randomly select a family of sets $\{I_\text{in}(j), I_\text{out}(j) \subset I\}_{j \in J}$
    and consider the relation $I_\text{in} ; \lcb{C} ; I_\text{out} \subset J \times J$
    given by
    \begin{align}
        \exists i,i' \in I.\ i \in I_\text{in}(j) \land \text{in}_i \in \lcb{C}(\text{out}_{i'}) \land i' \in I_\text{out}(j')
    \end{align}
    for all $(j,j') \in J \times J$.
    Suppose that the following conditions hold:
    \begin{enumerate}
        \item The size of $I_\text{out}(j)$ is at most $A$, for all $j \in J$.
        \item For all $j \in J$ and $i \in I$. If we randomly select
        $\{I_\text{in}(j), I_\text{out}(j)\}_{j \in J}$ from the marginal distribution
        fixing $I_\text{out}(j)$, then for each $j' \in J$, such that $j \neq j'$,
        the probability that $i \in I_\text{in}(j)$ 
        is at most $\epsilon$:
        \begin{equation}
            \text{Prob}(i \in I_\text{in}(j) \mid I_\text{out}(j')) \leq \epsilon
        \end{equation}
    \end{enumerate}
    The probability that the condition
    $I_\text{in} ; \lcb{C} ; I_\text{out} \subset \text{id}_J$ fails is at most
    $K^D \supp{J}^2 A \epsilon$.
\end{lemma}
\begin{proof}
    For some $j \neq j' \in J$ suppose that $I_\text{out}(j')$ is
    fixed and that $I_\text{in}(j)$ is chosen randomly.
    \begin{align}
        \text{Prob}(\text{LC}_C(I_\text{in}(j), I_\text{out}(j')))
        &\leq
        \sum_{i \in I_\text{out}(j')} \text{Prob}(\text{LC}_C(I_\text{in}(j), i)\\
        &\leq
        \sum_{i \in \lcb{C}(I_\text{out}(j'))} 
            \text{Prob}(i \in I_\text{in}(j))\\ 
        &\leq \sum_{i \in \lcb{C}(I_\text{out}(j'))} 
            \epsilon\\ 
        &\leq K^D A \epsilon
    \end{align}
    If
    $\{I_\text{in}(j), I_\text{out}(j) \subset I\}_{j \in J}$
    is chosen randomly we therefore
    have by the union bound that
  \begin{align}
    \text{Prob}(\exists j \neq j' \in J.\ \text{LC}_C(I_\text{in}(j), I_\text{out}(j')))
        &\leq
        \sum_{j \neq j' \in J} 
            \text{Prob}(\text{LC}_C(I_\text{in}(j), I_\text{out}(j')))\\
        &\leq \frac{\supp{J}(\supp{J} -1 )}{2} K^D A \epsilon\\
        &\leq \supp{J}^2 K^D A \epsilon
  \end{align}
\end{proof}

We next restate this as a bound on the contextual fraction.
\begin{lemma}  \label{lemma:bounding-cf}
    Let $S$ and $T$ be two multipartite scenarios,
    $s:S \to T$ an $n$-round simulation,
    and $C$ a classical $n$-round circuit strategy
    of depth $D$ and maximal fan-in $K$.
    Suppose that the following conditions hold
    when a deterministic simulation $t$ is chosen randomly with
    probability $s(t)$:
    \begin{itemize}
        \item For each $k < n$ and $j \in J$ we have $\text{In}_k(t)(j) = \emptyset$.
        \item For each $j \in J$ we have $\supp{\text{In}_n(t)(j)} \leq A$.
        \item For all $j,j' \in J$ and $i \in I$ such that $j \neq j'$:
        \begin{equation}
            s(i \in \text{In}_n(t)(j) \mid \text{Out}_n(t)(j')) \leq \epsilon
        \end{equation}
    \end{itemize}
    The contextual fraction of the pushforward
    $s_*(\tilde{C})$ is at most $\supp{J}^2 K^D A \epsilon$
    \begin{align}
        \cf{s_*(\tilde{C})} \leq \supp{J}^2 K^D A \epsilon
    \end{align}
\end{lemma}
\begin{proof}
    The pushforward $s_*(\tilde{C})$ is defined for each context $U$ of $T$ as the convex combination 
      \begin{align} 
        s_*(\tilde{C})(U) = \sum_{t} s(t) \cdot t_*(\tilde{C})(U) \label{eq:classical-bound-1}
      \end{align}
      The non-contextual fraction of $s_*(\tilde{C})$ (Definition \ref{def:contextual_fraction}) 
      is the greatest weight assigned 
      to the non-contextual part of any convex decomposition
      of $s_*(\tilde{C})$ into a non-contextual model and another
      empirical model.
      The non-contextual fraction of $s_*(\tilde{C})$
      is therefore bounded from below
      by the probability that $t_*(\tilde{C})$ is non-contextual 
      when $t$ is chosen randomly according to $s$.
        \begin{align}
                s(t_*(\tilde{C}) \text{ is non-contextual})
                \leq \text{NCF}(s_*(\tilde{C}))
        \end{align}
    Or equivalently
    \begin{align}
        \text{CF}(s_*(\tilde{C})) \leq s(t_*(\tilde{C}) \text{ is contextual})
    \end{align}
    By Lemma \ref{lemma:lightcones} if
    $t_*(\tilde{C})$ is contextual then $\text{In}_n(t) ; \lcb{C} ; \text{Out}_n(t) \not\subset \text{id}_I$.
    \begin{align}
        s(t_*(e) \text{ is contextual}) \leq
            s(\text{In}_n(t) ; \lcb{C} ; \text{Out}_n(t) \not\subset \text{id}_I)
    \end{align}
    By Lemma \ref{lemma:bgk} the probability
    of $\text{In}_n(t) ; \lcb{C} ; \text{Out}_n(t) \not\subset \text{id}_J$
    is at most $\supp{J}^2 K^D A \epsilon$, as required.
\end{proof}

We now consider the two simulations $s(S,G,d), s(I,G,d)$ from Section \ref{section:distributing-graph}.

\begin{lemma}
    Let $G$ be a rooted graph, $d \geq 2$ a dimension, $I$ a finite set, and 
    $S = (I, X, Y)$ a multipartite scenario.
    \begin{enumerate}
        \item Let $C$ a two-round classical circuit  strategy for the scenario $T(S,G,d)$ of depth $D$ and maximal fan-in $K$.
            \begin{equation}
                \cf{s(S,G,d)_*(\tilde{C})} \leq \supp{I}^2 K^D \supp{G}^{-1}
            \end{equation}
        \item Let $C$ be a single-round classical circuit strategy for the scenario $T(I, G, d)$
            of depth $D$ and maximal fan-in $K$.
            \begin{equation}
                \cf{s(I, G, d)_*(\tilde{C})} \leq \supp{I}^2 K^D \supp{G}^{-1}\text{rad}(G)
            \end{equation}
    \end{enumerate}
\end{lemma}
\begin{proof}
    In each case the simulation is defined as a convex combination
    \begin{align}
        \sum_{v \in \text{Paths}(G)^I} \big[ \prod_{i \in I} \pathdist(v_i) \big] \cdot t_v
    \end{align}
    where $t_v$ is a simulation with either
    \begin{align}
        \text{In}_1(t)(i) &= \emptyset\\
        \text{In}_2(t)(i) &= \{(i, v_{il_i})\}\\
        \text{Out}_2(t)(i) &= \{(i, v_{il_i})\}
    \end{align}
    in the two-round case, or 
    \begin{align}
        \text{In}_1(t)(i) &= \{(i, v_{il_i})\}\\
        \text{Out}_1(t)(i) &= \{(i, v_{ij})\}_{j=1}^{l_i}
    \end{align}
    in the single-round case.
    
    In either case condition 1.\ of Lemma \ref{lemma:bounding-cf} is satisfied.
    For the second condition
    we have $A = 1$ in the first case and $A \leq \text{rad}(G)$
    in the second case.
    
    Note that choice of paths are independent for different $i,i' \in I$.
    For the third condition we therefore have bound
    $\epsilon = 1 / \supp{G}$.
\end{proof}

\subsection{Randomised restrictions} \label{section:restrictions}
A common technique in complexity theory is to look at randomised
restrictions of circuits. In this section we explain
why this is not sufficient to prove a separation between shallow
quantum and classical circuits.

\begin{definition}
    Let $S = (I, X, Y)$ be a multipartite scenario
    and $B$ a single-round behaviour.
    A \emph{restriction} is a pair $(U, I')$
    where $I' \subset I$ is a subset of measurement
    sites and $U$ is a joint measurement
    for the remaining measurement sites $I \setminus I'$.
    Write  $\res{S}{I'}$ for the multipartite
    scenario $(I', (X_i)_{i \in I}, (Y_{i,x})_{i \in I', x \in X_i})$
    and $\res{e}{U,I'}:\res{S}{I'}$
    for the behaviour given by fixing
    the measurement setting at measurement sites
    $I \setminus I'$ according to $x$.
\end{definition}

Let $C$ be a classical circuit strategy of depth $D$
and maximal fan-in $K$. We consider the probability
that $\res{\tilde{C}}{U, I'}$ is contextual
when $(U, I')$ is selected randomly.
Classically we require communication to produce contextuality:
If $\tilde{C}$ is contextual then there must
be some some $i \neq i' \in I$ such that input 
wire $i$ communicates to output wire
$i'$ through $C$, i.e. $\text{in}_i \in \lcb{C}(\text{out}_{i'})$.
Hence if the restriction $\res{\tilde{C}}{U,I'}$
is contextual then there are $i \neq i' \in I'$
such that $\text{in}_i \in \lcb{C}(\text{out}_{i'})$.
Because $I'$ is small and uniformly distributed
it can be shown that the probability that
this occurs is at most $\epsilon K^D$ for some small $\epsilon$.
$\res{\tilde{C}}{U,I'}$ is therefore 
non-contextual with probability
at most $\epsilon K^D$.
\begin{align}
    \text{Prob}(\res{e}{U,I'} \text{ is contextual})
    &\leq \text{Prob}(\exists i \neq i' \in I'.\ i \in \lcb{C}(i'))\\
    &\leq  \epsilon K^D
\end{align}

Suppose next that $Q$ is a quantum circuit model
of depth $D$ and maximal fan-in $K$.
Quantum mechanically we require shared entanglement
to produce contextuality. Because we are interested in unconditional
separations we do not allow quantum circuits to start 
with an entangled state. If $\tilde{Q}$ is contextual
we therefore have some $i \neq i' \in I$
such that $\lcb{Q}(\text{out}_i) \cap \lcb{Q}(\text{out}_{i'}) \neq \emptyset$.
It can be shown that the probability of this occurring
is at most $\epsilon (K^D)^2$,
where $\epsilon$ is the same small parameter.
We, therefore, have that the probability of the restriction
$\res{\tilde{Q}}{U,I'}$ is contextual is bounded by 
$\epsilon (K^D)^2$.
\begin{align}
    \text{Prob}(\res{\tilde{Q}}{U,I'} \text{ is contextual})
     &\leq \text{Prob}(\exists i \neq i' \in I'.\ \lcb{Q}(i) \cap \lcb{Q}(i') \neq \emptyset) \\
    &\leq  \epsilon (K^D)^2
\end{align}

Hence by looking at the restrictions we can only detect
a constant difference between shallow quantum and classical circuits.

\section{Quantum advantage with shallow circuits}
\label{section:qa-result}
\begin{figure}
    \centering
    \hfill
    \begin{subfigure}{0.3\textwidth}
        \centering
        \includegraphics[scale=1]{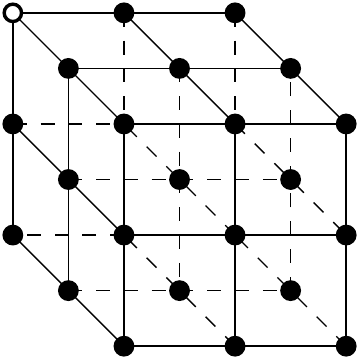}
        \caption{}\label{fig:3x3x3-grid}
   \end{subfigure}
    \hfill
    \vrule
    \hfill
    \begin{subfigure}{0.3\textwidth}
        \centering
        \includegraphics[scale=1]{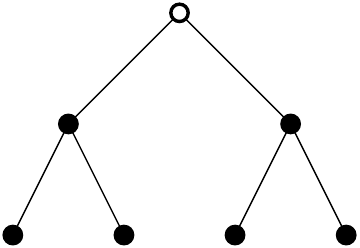}
        \caption{}\label{fig:t22}
    \end{subfigure}
    \hfill{}
    \caption{
        Examples of rooted graphs. The circle denotes the root node.
        (a) A 3-dimensional hypergrid. (b) A binary tree. 
    }
    \label{fig:graph_examples}
\end{figure}

A shallow circuit is a family of circuits, not a single circuit.
We now move from single circuits to families of circuits
by considering families of rooted graphs $\{G_n\}_{n \in \nats}$.
We first present the two results. We then explain the proof of both
results, and in section \ref{section:examples} we give some examples
of families of graphs and shallow circuits.

Suppose first that $(e, \Phi)$ is a non-local game with classical bound $\gamma$.
In Section \ref{section:distribute-2r} we define a quantum circuit $Q_n$
that depends only on the number of qudits of $\psi$ and the degree of $G_n$.
If $\{G_n\}_{n \in \nats}$ has bounded degree then $\{Q_n\}_{n \in \nats}$ is shallow.
We also define a game $\Phi_n$ such that the success probability of $Q_n$ exceeds $\gamma$.
In Section \ref{section:bounding-cf} we showed that for any
classical circuit the success probability is bounded.
We, therefore, have the following result.
\begin{theorem}
    Let $(e_{\psi,\pi}, \Phi, \gamma)$ be a qudit non-local game 
    and $\{G_n\}_{n \in \nats}$ a family of rooted
    graphs of bounded degree and unbounded size.
    There exists two-round interactive
    games $\{\Phi_{G_n}\}_{n \in \nats}$
    and a shallow quantum circuit strategy
    $\{U_{G_n, \psi,\pi}\}_{n \in \nats}$
    such that 
    \begin{enumerate}
        \item The success probability of $Q_{G_n, \psi,\pi}$
        on $\Phi_{G_n}$ violates $\gamma$:
        \begin{align}
            p_S(Q_{G_n,\psi,\pi}, \Phi_{G_n}) > \gamma
        \end{align}
        \item Let $\{C_n\}_{n \in \nats}$ be any
            shallow circuit strategy.
            The success probability of $C_n$
            on $\Phi_{G_n}$ tends to $\gamma$ as $n$ increases.
        \begin{align*}
            p_S(C_n, \Phi_{G_n}) &\leq \gamma + \epsilon_n
        \end{align*}
        \end{enumerate}
        where $\gamma < 1$ is the classical bound
        of $\Phi$ and 
        $\epsilon_n \in O(\supp{G_n}^{-1})$.
\end{theorem}

Before giving the proof we state the analogous result for Weyl measurement games.
Suppose that the quantum strategy $e$ is realised by Weyl measurements on a single-qudit state.
We can then instead consider the quantum circuit $Q_n$ and non-local game $\Phi_n$
defined in Section \ref{section:distribute-2r}. This gives the following result.
\begin{theorem}
    Let $(\Phi, \psi, \gamma)$ be a Weyl measurement game,
    and $\{G_n\}_{n \in \nats}$ a family of rooted graphs
    such that $\text{deg}(G_n) \in O(1)$
    and $\text{lim}_{n \to \infty} \text{rad}(G_n) / \supp{G_n} = 0$.
    There exists non-local games $\{\Phi_{G_n}\}_{n \in \nats}$ and
    quantum circuit strategies $\{U_{\psi, G_n}\}_{n \in \nats}$ such that 
    \begin{enumerate}
        \item $\{U_{\psi, G_n}\}_{n \in \nats}$ is shallow, and the success probability
        of $U_{\psi,G_n}$ on $\Phi_{G_n}$ exceeds $\gamma$:
        \begin{align}
            p_S(U_{\psi,G_n}, \Phi_{G_n}) > \gamma
        \end{align}
        \item If $\{C_n\}_{n \in \nats}$ is any shallow classical circuit
            strategy then the violation of $\gamma$ by $C_n$
            tends to $0$ for large $n$:
            \begin{align*}
                p_S(C_n, \Phi_{G_n}) \leq \gamma + \epsilon_n
            \end{align*}
        \end{enumerate}
        where $\epsilon_n \in O(\text{rad}(G_n) / \supp{G_n})$.
\end{theorem}

The proofs take the same form.

\begin{proof}(Theorem I and II)\\
The proof of both statements follow the same pattern.
We take a non-local game $(e, \Phi)$ with a particular quantum realisation
and a family of graphs $\{G_n\}_{n \in \nats}$.
Write $\{e_n\}_n$ for the associated family of quantum
realised empirical models, $\{s_n\}_{n}$ for the family of simulations.
Because $e_n$ simulates $e$ its success probability on the pullback
problem $(s_n)^*(\Phi)$ is equal to the success probability
of $e$, which violates the classical bound of the non-local game.

Under the assumptions on the graphs $\{G_n\}_n$ the model
these empirical models can be recast as a shallow quantum circuit
$\{Q_n\}_n$.

For the classical bound, we can use Section \ref{section:bounding-cf}
to derive a bound on the form
\begin{align}
    \cf{(s_n)_*(\tilde{C}_n)} \leq \epsilon_n
\end{align}
for any classical shallow circuit $\{C_n\}_{n \in \nats}$.
The violation of $(s_n)_*(\tilde{C_n})$ of the classical bound
$\gamma$ for the non-local game $(e,\Phi)$ is therefore at most
$\epsilon_n$.
It follows that the violation of $C_n$ of $\gamma$ on the pullback
problem is also bounded by $\epsilon_n$.

We also comment that the circuits we use are equivalent up to a constant
factor in-depth and fan-in to circuits using only input/output wires
with bits and only quantum wires that are qubits.
An important point is that the number of measurement settings
at each measurement site does not blow up.
\end{proof}

\subsection{Examples} \label{section:examples}
We will now present some concrete examples of circuits
arising from the construction we have presented.
We first define two classes of rooted graphs (Figure \ref{fig:graph_examples}).
\begin{definition}
    The \emph{hypergrid graph} $[n]^k$, where $n,k \in \nats$,
    has nodes $\{(a_1, \dots, a_k) \mid 1 \leq a_1, \dots, a_k \leq n\}$,
    root $(1, \dots, 1)$,
    and an edge $\{(a_1, \dots, a_k), (b_1,\dots, b_k)\}$
    whenever $\supp{a_j - b_j} = 1$ for some $j$
    and $a_{j'} = b_{j'}$ for all $j' \neq j$.
\end{definition}
Note that $[n]^1$ is a line, $[n]^2$ is a square grid, $[n]^3$ is a 3D grid, etc.
\begin{definition}
    Let $n,k \in \nats$.
    The \emph{$k$-ary tree $T_{k,n}$ of depth $n$} is
    the rooted graph
    with nodes 
    $\{(i,j) \mid i \in \{1,\dots,n\}, j \in \{1, \dots, k^{i-1}\}\}$,
    root $(1,1)$,
    and edges $\{(i,j), (i',j')\}$ whenever 
    $i' = i + 1$ and $k(j-1) < j' \leq kj$.
\end{definition}

Let $n,k \in \nats$. The hypergrid $[n]^k$
has degree $2k$ and the tree $T_{n,k}$ has degree $k$.
For a fixed $k \in \nats$ the families
$\{[n]^k\}_{n \in \nats}$ or $\{T_{n,k}\}_{n \in \nats}$
therefore have bounded degree.
If we use either family to define a family of quantum circuits $\{Q_n\}_{n \in \nats}$
then the resulting circuit is shallow.
Furthermore, $[n]^k$ has radius $nk$ and size $n^k$.
while $T_{n,k}$ has radius $k$ and size $k^{n+1}-1$.
In Theorem I the parameters $\epsilon_n$
therefore converge at a rate of
$O(1/n^{k})$ or $O(1 / {k^n})$ respectively.
And in Theorem II the rate of converge is
$O(1/n^{k-1})$ or $O(1 / {k^{n-1}})$ respectively.

For the hypergrid graphs the quantum circuits
have polynomial size $O(n^k)$, while for tree graphs it has exponential size $O(k^n)$.

The game $\Phi_G$ and the quantum circuit strategy
$\tilde{Q}_{\psi,G}$ is a distributed circuit version
of any non-local game. We showed that this game
is solved with high probability by $Q_{\psi,G}$
whose depth and maximal fan-in only depend on the size of $I$
and the maximal degree of $G$.
However, for any classical circuit strategy
the success probability is bounded by a bound involving the
size of $G$ and the radius of $G$.

The quantum circuits arising from the Magic Square game
and 2D graphs and binary trees are shown in the following figures
(Figures \ref{fig:my_label_one} - \ref{fig:my_label_four}).

\begin{figure}
    \centering
    \includegraphics{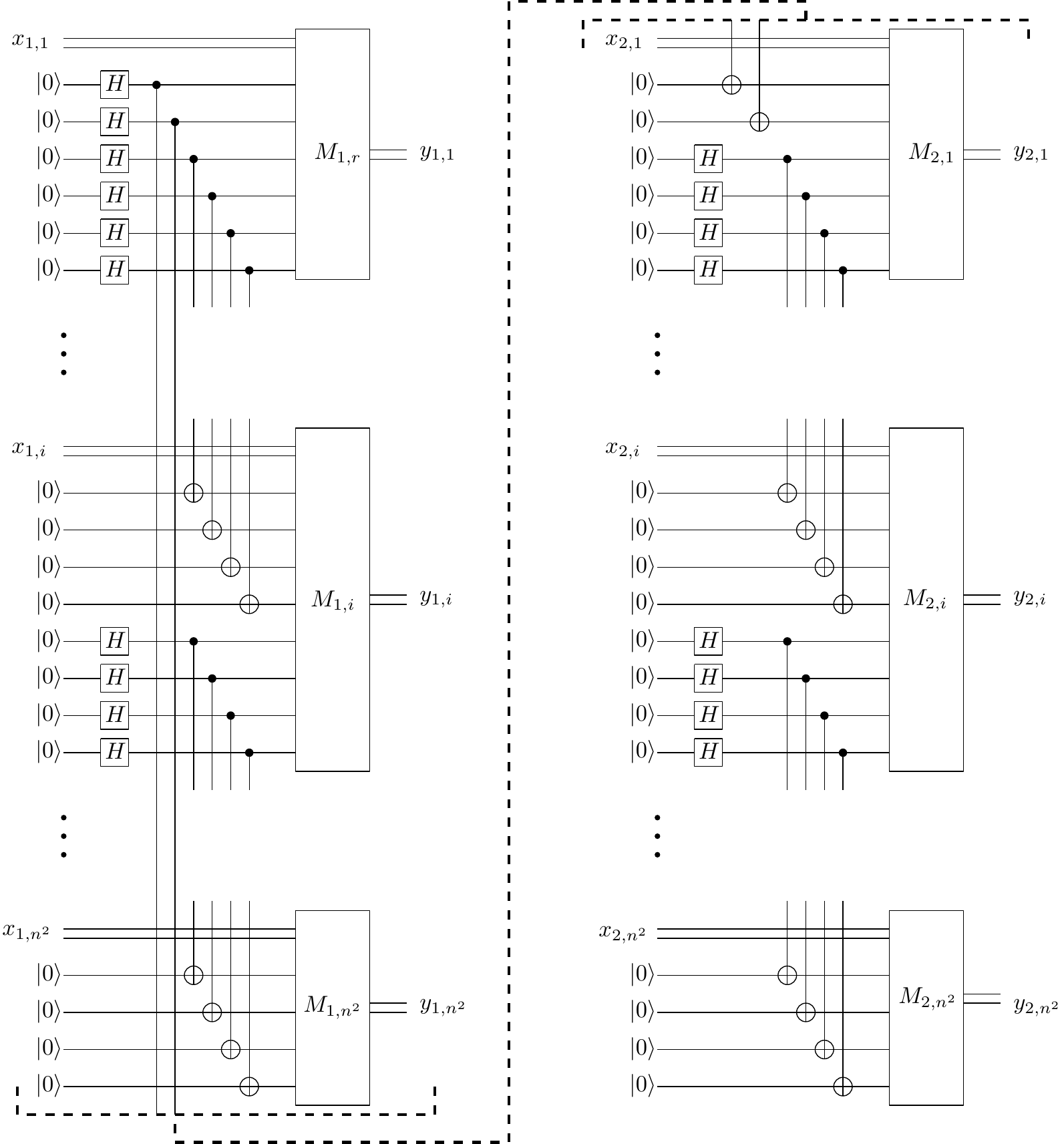}
    \caption{Circuit version of Magic Square game played on a 2D grid. }
    \label{fig:my_label_one}
\end{figure}

\begin{figure}
    \centering
    \includegraphics{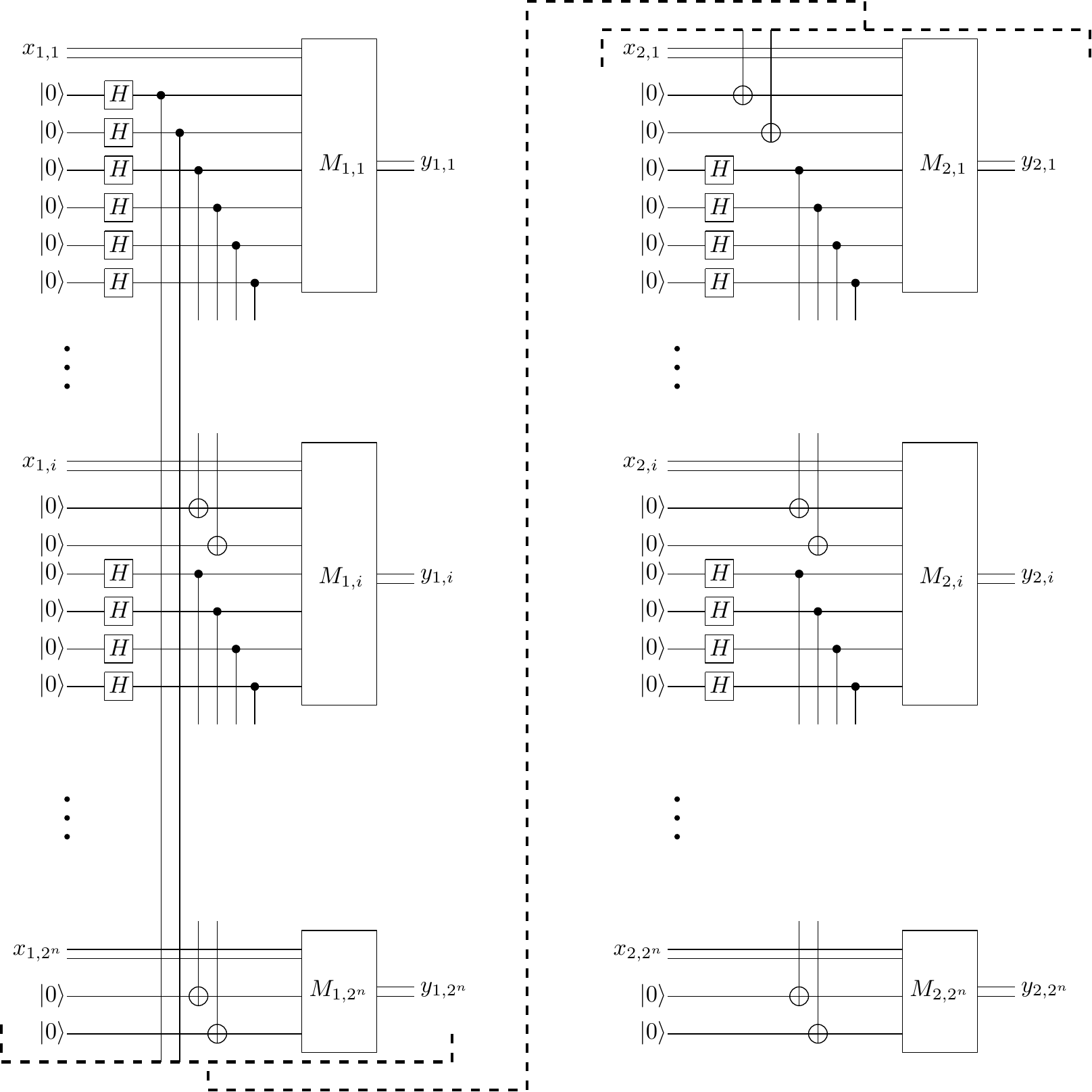}
    \caption{Circuit version of the Magic Square game played on
    binary trees.}
    \label{fig:my_label_two}
\end{figure}

\begin{figure}
    \centering
    \includegraphics[width=1\textwidth]{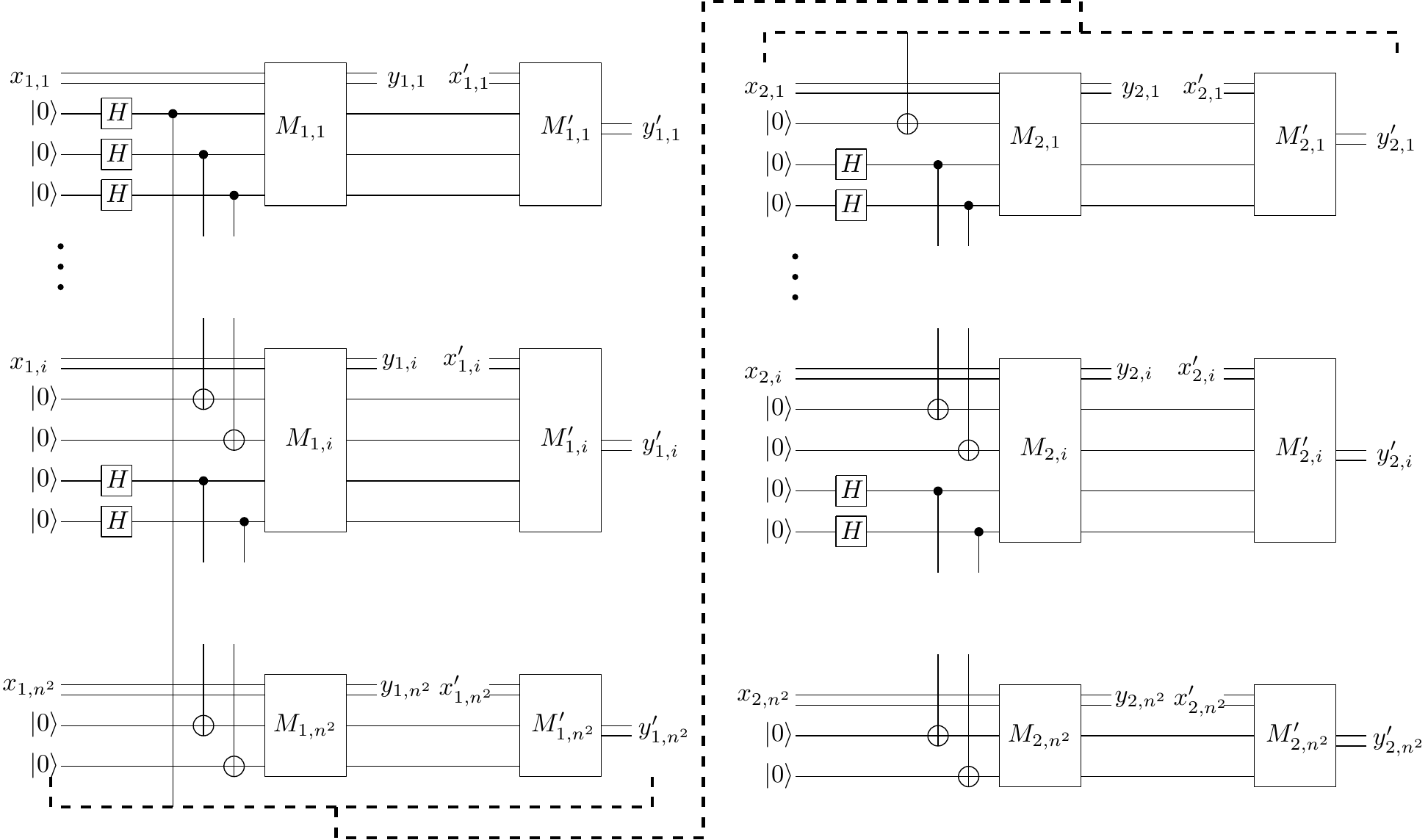}
    \caption{Magic square game played on a grid in two rounds.}
    \label{fig:my_label_three}
\end{figure}

\begin{figure}
    \centering
    \includegraphics[width=1\textwidth]{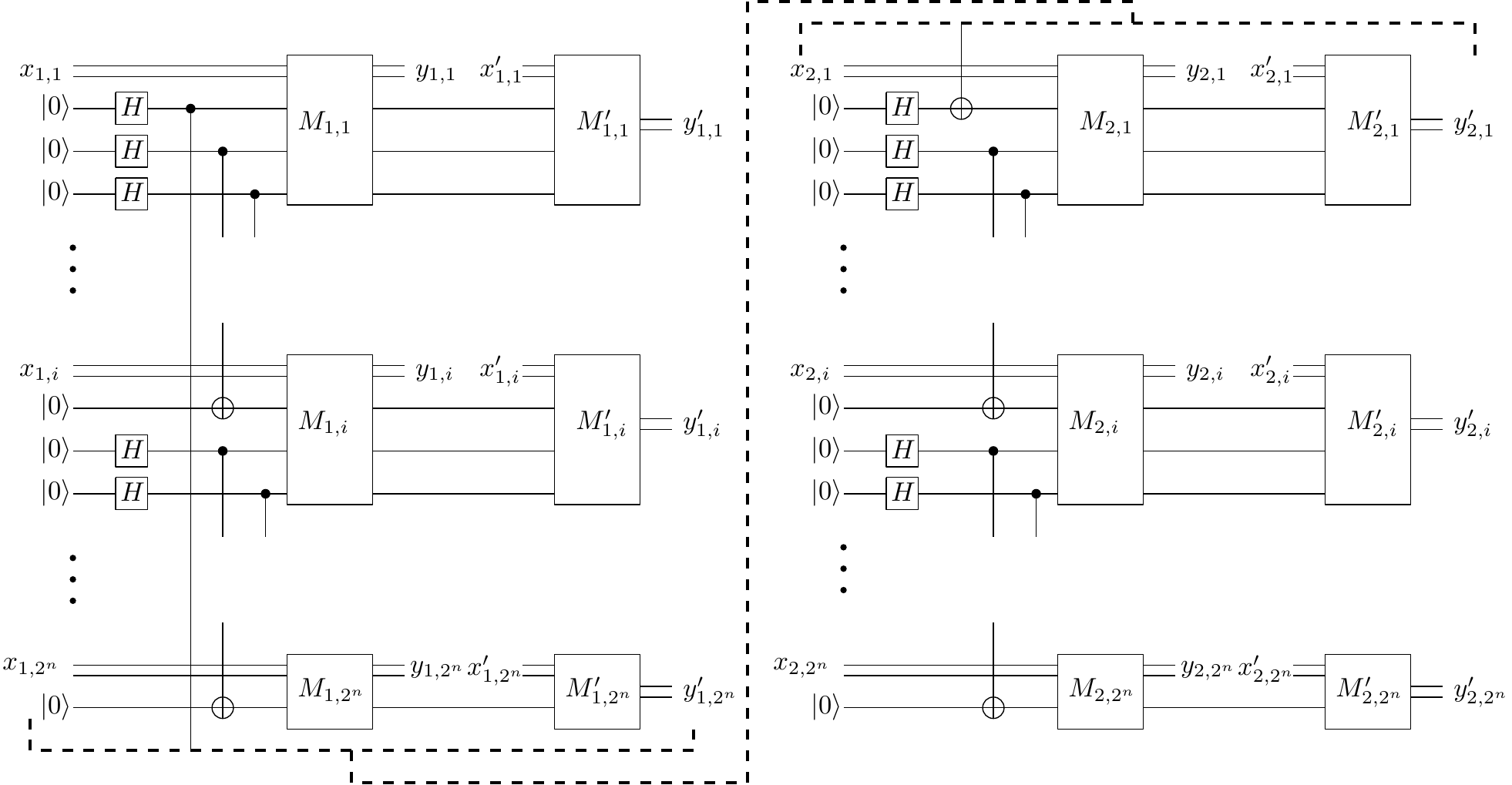}
    \caption{Magic square game played on a tree in two rounds.}
    \label{fig:my_label_four}
\end{figure}
\chapter{Final Remarks}
The motivation behind this thesis is the idea that contextuality
can be a useful phenomena for understanding quantum advantage
in computing. The results that we have presented touches on two different
aspects of this idea. We need more refined ways of identifying
types of contextuality, and we need concrete examples relating
these types of contextuality to quantum advantage.

Cohomology is in many settings a powerful technique for identifying useful structure in data.
For example, the simplicial cohomology of a topological space is related to its number of ``holes''.
In the sheaf theoretic framework it is natural to consider \v{C}ech cohomology
as an invariant of contextuality. \v{C}ech cohomology can detect
contextuality in a range of examples, indicating that it could
be a useful way of identifying types of contextuality.
Another promising technique for studying contextuality
is the topological approach of Okay et al.\
We have shown that any false negative of the \v{C}ech approach
induces a false negative of the topological approach.
As far as detecting contextuality
the topological approach therefore cannot go further than the \v{C}ech cohomology approach.

Bravyi, Gosset, and K\"{o}nig's quantum advantage result with
shallow circuits currently is the strongest example of quantum
advantage using contextuality. We have highlighted the role
of simulations in their result as a way of bounding
the success probability of classical circuits.
We have extended their result by giving a systematic way
of promoting any quantum realised multipartite empirical model to a quantum advantage result
with shallow circuits.
The construction is parametrised by a family of graphs that are used as templates
to spread entanglement using teleportation.
By considering different families of graphs we can achieve a different
tradeoff in size vs strength of separation.

There are measurements making any (pure) entangled $n$-qudit state
contextual \cite{hardy_nonlocality_1993, abramsky_hardy_2016}.
One can therefore use any (pure) entangled state.

\section{Further work}
An unconditional quantum advantage result for a general computational model
appears to be far away. However, BGK's
result has spawned a new interest in trying to prove unconditional
results for models of computation with structural restrictions.
Perhaps by studying such models we can identify the common structures
that are important. A potential candidate for such a model
is the \emph{cell probe model}.

Contextuality has shown itself to be useful for proving unconditional bounds on memory
complexity. Karanjai et al. \cite{karanjai_contextuality_2018} proves
that the memory complexity of the Gottesmann Knill algorithm \cite{aaronson_improved_2004}
is asymptotically optimal. The proof appears to exploit a particular type of contextuality.
It might be interesting to see if this structure can be given
a concise description using methods like cohomology or simulations
and if it can similarly be generalised to give a general connection
between contextuality and quantum advantage.

\addcontentsline{toc}{chapter}{Bibliography}
\bibliography{references}        
\bibliographystyle{alpha}  

\end{document}